\documentclass[12pt, a4paper]{article}

\title{Field-free Line Magnetic Particle Imaging: Radon-based Artifact Reduction with Motion Models}
\author{Stephanie Blanke\footnote{stephanie.blanke@uni-hamburg.de}~ and Christina Brandt}
\date{}

\usepackage{graphicx}        

\usepackage{subcaption}

\usepackage{amsfonts,amsmath}
\usepackage{amssymb}

\usepackage{enumitem}

\usepackage[pdftex,colorlinks, citecolor=black,linkcolor=black,urlcolor=black]{hyperref}

\usepackage{color}

\usepackage[a4paper, textwidth=15.5cm]{geometry}

\DeclareMathOperator{\R}{\mathbb R}
\DeclareMathOperator{\dmathInt}{\, d}
\DeclareMathOperator{\Dmath}{D}
\DeclareMathOperator{\BV}{BV}
\DeclareMathOperator{\TV}{TV}
\DeclareMathOperator{\supp}{supp}

\newcommand{\brackets}[1]{\left( #1 \right)}
\newcommand{\abs}[1]{\left| #1 \right|}
\newcommand{\norm}[1]{\left\| #1 \right\|}
\newcommand{\dmath}[1]{\frac{\text{d}}{\text{d}#1}}
\newcommand{\tr}[1]{\text{tr}\left[ #1\right] }

\def\angleRange{{\left[ 0,2\pi\right] }}

\newtheorem{definition}{Definition}[section]
\newtheorem{lemma}[definition]{Lemma}
\newtheorem{remark}[definition]{Remark}
\newtheorem{example}[definition]{Example}
\newtheorem{theorem}[definition]{Theorem}
\newtheorem{corollary}[definition]{Corollary}

\newtheorem{assumption}[definition]{Assumption}

\newenvironment{proof}{\textbf{Proof:}}{\hfill$\square$ \\}

\begin{document}

\maketitle

\vspace{-0.7cm}

\begin{center}
	\small
	Department of Mathematics, Universit\"at Hamburg, Germany
\end{center}

\vspace{0.05cm}

\begin{abstract}
	\hspace{-0.58cm}Magnetic particle imaging is a promising medical imaging technique. Applying changing magnetic fields to tracer material injected into the object under investigation results in a change in magnetization. Measurement of related induced voltage signals enables reconstruction of the particle distribution. For the field-free line scanner the scanning geometry is similar to the one in computerized tomography. We make use of these similarities to derive a forward model for dynamic particle concentrations. We validate our theoretical findings for synthetic data. By utilizing information about the object's dynamics in terms of a diffeomorphic motion model, we are able to jointly reconstruct the particle concentration and the corresponding dynamic Radon data without or with reduced motion artifacts. Thereby, we apply total variation regularization for the concentration and an optional sparsity constraint on the Radon data. 
\end{abstract}

\small

\textbf{Keywords:} Magnetic Particle Imaging, Computerized Tomography, Dynamic Inverse \phantom{\hspace{2.8cm}} Problems, Total Variation Regularization

%%%%%%%%%%%%%%%%%%%%%%%%%%%%%%%%%%
\section{Introduction}
%%%%%%%%%%%%%%%%%%%%%%%%%%%%%%%%%%
Magnetic particle imaging (MPI) is a fast tracer-based medical imaging technique allowing for quantitative imaging of magnetic nanoparticles (MNPs) within the patient's body. Relying on the non-linear magnetization response of magnetic particles to changing magnetic fields, in MPI the composition of a so-called selection field $\mathbf{H}_S$ featuring a low-field volume (LFV) and a drive field $\mathbf{H}_D$ steering the LFV through the field of view (FOV) is applied to the region of interest. The change in magnetization leads to an induced voltage signal captured by dedicated receive coils. Since only MNPs in close vicinity to the LFV are contributing to the signal, as all other particles stay in magnetic saturation, spatial information is encoded in the measurements. Having access to data for a suitable amount of different LFV positions enables reconstruction of the spatial distribution of the particles. We recommend to consult~\cite{knopp2012magnetic} for an elaborate introduction to MPI. In view of achievements and milestones in MPI history we refer to~\cite{knopp2017magnetic}. \\
Most scanner implementations make use of a field-free point (FFP) centering the LFV. Aside from that,  Weizenecker, Gleich, and Borgert suggested a field-free line (FFL) encoding scheme in~\cite{weizenecker2008magnetic}. They stated a possible increase in sensitivity for an FFL scanner as, due to the larger LFV, more particles are contributing to the signal. The first FFL imagers were developed by~\cite{goodwill2012projection} and~\cite{bente2014electronic}. Information about a small-bore imaging system using a field-free line is shared within the scope of	the open-source project OS-MPI~\cite{mattingly2020mpi}. MPI leads to the task of solving an ill-posed inverse problem. However, according to~\cite{kluth2018degree} the FFL setting may lead to a less ill-posed problem compared to the usage of an FFP. For more details concerning FFL scanner, we refer to~\cite{bringout2016field} and~\cite{erbe2014field}. \\
With the first article regarding MPI having been published in 2005~\cite{gleich2005tomographic}, MPI is still a rather new imaging modality leaving a lot of space for research. Lately, reduction of motion-based artifacts started to arouse further interest. In~\cite{gdaniec2017detection} periodic particle concentrations are considered, whose dynamic properties e.g. result from organ movements like the beating heart or respiration. They estimate the motion frequency from the measured data, which is then grouped according to the different motion states. After their proposed data processing, usual MPI reconstruction methods assuming static tracer distributions can be applied. This approach is extended to multi-patch MPI in~\cite{gdaniec2020suppression}. An adaption, and corresponding significance examination, of the MPI forward model allowing possibly fast and non-periodic timely changing tracer distributions can be found in~\cite{brandt2021modeling}. In order to circumvent motion and multi-patch artifacts, in~\cite{brandt2022motion} the authors further developed a reconstruction method based on an expression of the dynamic MNP distribution in terms of spline curves. For numerical results all of the aforementioned works regard the FFP encoding scheme. In this article, we investigate modeling and reconstruction of dynamic particle concentrations for MPI using a field-free line. \\
For MPI-FFL devices the scanning geometry resembles the one in computerized tomography (CT), both rely on measuring data along lines. While MPI is a tracer-based imaging modality, CT accesses morphological information about the tissue itself~\cite{knopp2012magnetic}. Thus, combining these two modalities results in a powerful instrument for clinical diagnostics and accordingly the idea for a hybrid MPI-CT scanner came up~\cite{vogel2019magnetic}. For sequential data generation, i.e. acquiring data while translating the field-free line through the FOV and rotating the FFL in between measurements, it was derived in~(\cite{knopp2011fourier},~\cite{bringout2020new}) for static concentrations that MPI data can be linked to corresponding Radon data and well-established reconstruction methods developed for CT become also available for MPI. Remember the Radon transform, mapping a function to the set of its line integrals, gives the forward operator in computerized tomography. With respect to their open-sided FFL scanner prototype, the authors of~\cite{kilic2022inverse} compared images obtained via system matrix based reconstruction using Kaczmarz to inverse Radon transform based reconstruction results. In~\cite{https://doi.org/10.48550/arxiv.2211.11683} the relation between MPI and Radon data was investigated for the acquisition scheme suggested in~\cite{weizenecker2008magnetic} that is the FFL simultaneously rotates with its translation. It was found that this setting is connected to a non-rotating oscillating FFL in combination with a continuously rotating phantom. On that account, in this work we extend the result of~\cite{knopp2011fourier} towards time-dependent particle concentrations not restricted to phantom rotation. \\\\
Because of the geometrical similarity to CT, we model time-dependent concentrations based on a framework applied in dynamic CT (e.g.~\cite{+2014+323+339},~\cite{hahn2017motion},~\cite{hahn2021motion}). More precisely, at each time point we link the dynamic particle concentration to a static reference concentration, for example the initial MNP distribution, via diffeomorphic motion functions. For multi-patch MPI using a field-free point and restricting to rigid motion, a registration-based approach also led to diffeomorphic transformation model~\cite{ehrhardt2019temporal}. While the authors jointly regarded image and motion estimation, within the scope of this work, we assume the dynamics to be known. Nevertheless, in practice the motion needs to be reconstructed beforehand or simultaneously with image reconstruction. We derive a connection between MPI and dynamic Radon data. Since the resulting adapted Radon transform is the same as derived e.g. in~\cite{hahn2021motion}, related methods and results from dynamic CT become available. \\
Instead of first reconstructing the adapted Radon data and then applying dedicated reconstruction techniques, inspired by~\cite{Tovey_2019} we simultaneously derive particle concentration and associated dynamic Radon data applying total variation (TV) regularization. Variational approaches are suitable for many applications, including the treatment of inverse problems~\cite{hauptmann2021image}. Further, also in view of MPI, TV already has been used in order to incorporate a priori information of the object in the reconstruction process~(\cite{storath2016edge},~\cite{zdun2021fast},~\cite{bathke2017improved},~\cite{ilbey2017comparison}). For an introduction to TV regularization we refer e.g. to~\cite{burger2013guide}. \\\\
We keep the structure and notation of the article similar to those in~\cite{https://doi.org/10.48550/arxiv.2211.11683} such that results for the different settings can be easily compared. In Section~\ref{Sec:FFL} we give a brief overview of MPI. For static MNP distributions we state the corresponding forward model for the FFL encoding scheme and recapitulate the  connection between MPI and Radon data known from~\cite{knopp2011fourier}. Afterwards, in Section~\ref{Sec:DynModel} we derive a forward model for dynamic concentrations and build the link towards an adapted version of the Radon transform in Section~\ref{Sec:Relation}. In Section~\ref{Sec:ImageReco} we present our TV-based joint reconstruction approach similar to~\cite{https://doi.org/10.48550/arxiv.2211.11683}. Finally, in Section~\ref{Sec:Results} we state numerical results for simulated data.

\section{Principles of MPI using a field-free line} \label{Sec:FFL} 

According to~\cite{knopp2012magnetic}, the signal equation relating the time-dependent MNP distribution $c: \, \R^3\times\R^+_{0} \to \R^+_{0} $ with $\R^+_{0} := \R^+ \cup \left\lbrace 0\right\rbrace$ to the measured voltage signal induced in the $l$-th receive coil $u_l: \, \R^+_0 \to \R,\; l\in\left\lbrace 1,\dots,L\right\rbrace $ can be written as
\begin{equation}
u_l\brackets{t} = -\mu_0 \dmath{t}\int_{\R^3} c\brackets{\mathbf{r},t} \overline{\mathbf{m}}\brackets{\mathbf{r}, t} \cdot \mathbf{p}_l\brackets{\mathbf{r}} \dmathInt\mathbf{r}
\label{SignalEquation}
\end{equation}
with magnetic permeability $\mu_0$, mean magnetic moment $\overline{\mathbf{m}}: \, \R^3 \times\R^+_0 \to \R^3$, and receive coil sensitivities $\mathbf{p}_l: \, \R^3 \to \R^3$. In practice, the magnetic excitation field additionally contributes to the voltage signal, which thus needs to be adjusted. However, we omit signal filtering and assume to have preprocessed the data. Further, for the moment we suppose the concentration to be static. Then, equation~\eqref{SignalEquation} transforms to
\begin{equation}
u_l\brackets{t} = -\mu_0 \int_{\R^3} c\brackets{\mathbf{r}} \frac{\partial}{\partial t}\overline{\mathbf{m}}\brackets{\mathbf{r}, t} \cdot \mathbf{p}_l\brackets{\mathbf{r}} \dmathInt\mathbf{r}.
\label{SignalEquation_static}
\end{equation}
In order to guarantee that~\eqref{SignalEquation_static} is well-defined, we propose that $c$ and $\frac{\partial}{\partial t}\overline{\mathbf{m}}\brackets{\cdot, t} \cdot \mathbf{p}_l\brackets{\cdot}$ are $L_2$-functions. For our considerations we need a model for the mean magnetic moment $\overline{\mathbf{m}}$. The state-of-the-art approach, which we will use in this work, is given by applying the Langevin theory of paramagnetism. Then, following~\cite{knopp2011fourier} we can write
\begin{equation}
\overline{\mathbf{m}}\brackets{\mathbf{r}, t} = \overline{m}\brackets{\left\|  \mathbf{H}\brackets{\mathbf{r}, t} \right\|} \frac{\mathbf{H}\brackets{\mathbf{r}, t}}{\left\|  \mathbf{H}\brackets{\mathbf{r}, t} \right\|}
\label{Eqn:ModelinMagnetization}
\end{equation}
with modulus of the mean magnetic moment $\overline{m}\brackets{H}= m\mathcal{L}\brackets{\frac{\mu_0 m}{k_{\text{B}} T_p}H}$. Thereby, $m$ denotes the magnetic moment of a single particle, $k_{\text{B}}$ the Boltzmann constant, $T_p$ the particle temperature, and $\mathcal{L}$ the Langevin function 
\begin{eqnarray}
\mathcal{L}:\; \R \to \left[ -1,1\right], \quad
\mathcal{L}\brackets{\lambda} := \begin{cases}
\coth \brackets{\lambda} - \frac{1}{\lambda}&, \; \lambda \neq 0, \\
0&, \; \lambda=0.
\end{cases}
\label{eq:Langevin}
\end{eqnarray}
\begin{remark}
	The Langevin model makes the assumption that particles are in thermal equilibrium~\cite{knopp2012magnetic}, which is not an adequate characterization of the particle dynamics. Accordingly, finding more accurate models shows an important field of research. We refer to (\cite{weizenecker2018fokker},~\cite{kluth2019towards},~\cite{kaltenbacher2021parameter},~\cite{albers2022modeling}) for some works heading in this direction. Further,~\cite{kluth2018mathematical} is a survey paper regarding mathematical modeling of the signal chain.
\end{remark}
We are interested in MPI using an FFL scanner. More precisely, for our considerations, we regard the scanning geometry visualized in Figure~\ref{Fig:SeqRot}, which was also used e.g. in~\cite{knopp2011fourier}. Thereby, the FFL is sequentially translated through the FOV and rotated in between measurements. We now derive the signal equation for the specific setting of using a field-free line for spatial encoding. To simplify matters, we suppose the magnetic particles to be contained within the xy-plane. Thus, for data generation the FFL is moved within this plane and the problem formulation can be considered as two dimensional,~cf.~\cite{kluth2018degree} for further details. As already suggested the total magnetic field $\mathbf{H}(\mathbf{r},\varphi,t)$ is composed of a static selection field $\mathbf{H}_S(\mathbf{r},\varphi)$ and a time-dependent drive field $\mathbf{H}_D(\varphi,t)$. The angle $\varphi$ determines the orientation of the FFL within the xy-plane (cf. Figure~\ref{Fig:Geom}). Based on~\cite{knopp2011fourier} we model these fields as
\begin{equation}
\mathbf{H}(\mathbf{r},\varphi,t) = \mathbf{H}_S(\mathbf{r},\varphi) + \mathbf{H}_D(\varphi,t) = \left( -G \; \mathbf{r}
\cdot \mathbf{e}_{\varphi}
+ A\Lambda_\varphi(t)\right) \mathbf{e}_{\varphi}
\label{Eqn:MagneticField}
\end{equation}
with unit vector $\mathbf{e}_\varphi:=\brackets{-\sin\varphi,\cos\varphi}^T$ orthogonal to the FFL. More components are the gradient strength $G$ determining the width of the LFV, the drive peak amplitude $A$, and an excitation function $\Lambda_\varphi$ typically chosen to be sinusoidal. Computing the nulls of the magnetic field gives the field-free line 
\begin{equation*}
\text{FFL}\brackets{\mathbf{e_{\varphi}}, s_{\varphi,t}}:=\left\lbrace \mathbf{r} \in \R^2 : \; \mathbf{r} \cdot \mathbf{e_{\varphi}} = s_{\varphi,t} \right\rbrace
\end{equation*}
with displacement $ s_{\varphi,t}:= \frac{A}{G} \Lambda_\varphi\brackets{t}$ to the origin and normal vector $\mathbf{e}_\varphi$ visualized in Figure~\ref{Fig:Geom}.
From~\eqref{Eqn:MagneticField} we further get that the magnetic field is constant along lines parallel to the FFL~\cite{knopp2011fourier}. Note that we consider an idealistic setting. In real measurements field imperfections result in deformed LFVs~(e.g.~\cite{bringout2016field},~\cite{bringout2020new}) leading to artifacts in the reconstructed images when being ignored.
\begin{figure}[htbp]
	\begin{subfigure}[b]{0.45\textwidth}
		\centering 
		\includegraphics[width=1\linewidth]{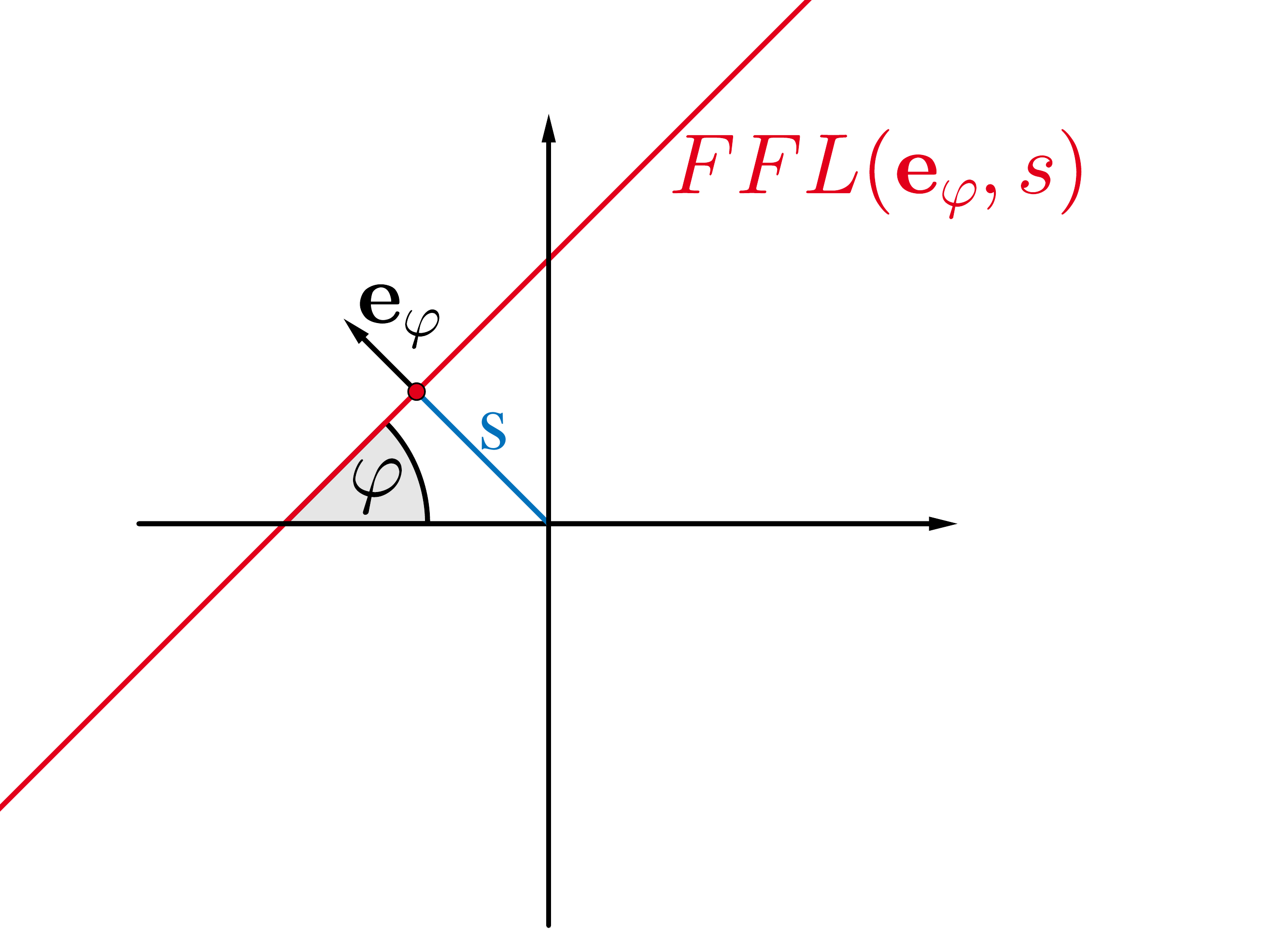}
		\subcaption{\small FFL orthogonal to $\mathbf{e}_\varphi$ and with displacement $s$ to the origin}
		\label{Fig:Geom}
	\end{subfigure}
	\hfill
	\begin{subfigure}[b]{0.45\textwidth}
		\centering
		\includegraphics[width=1\linewidth]{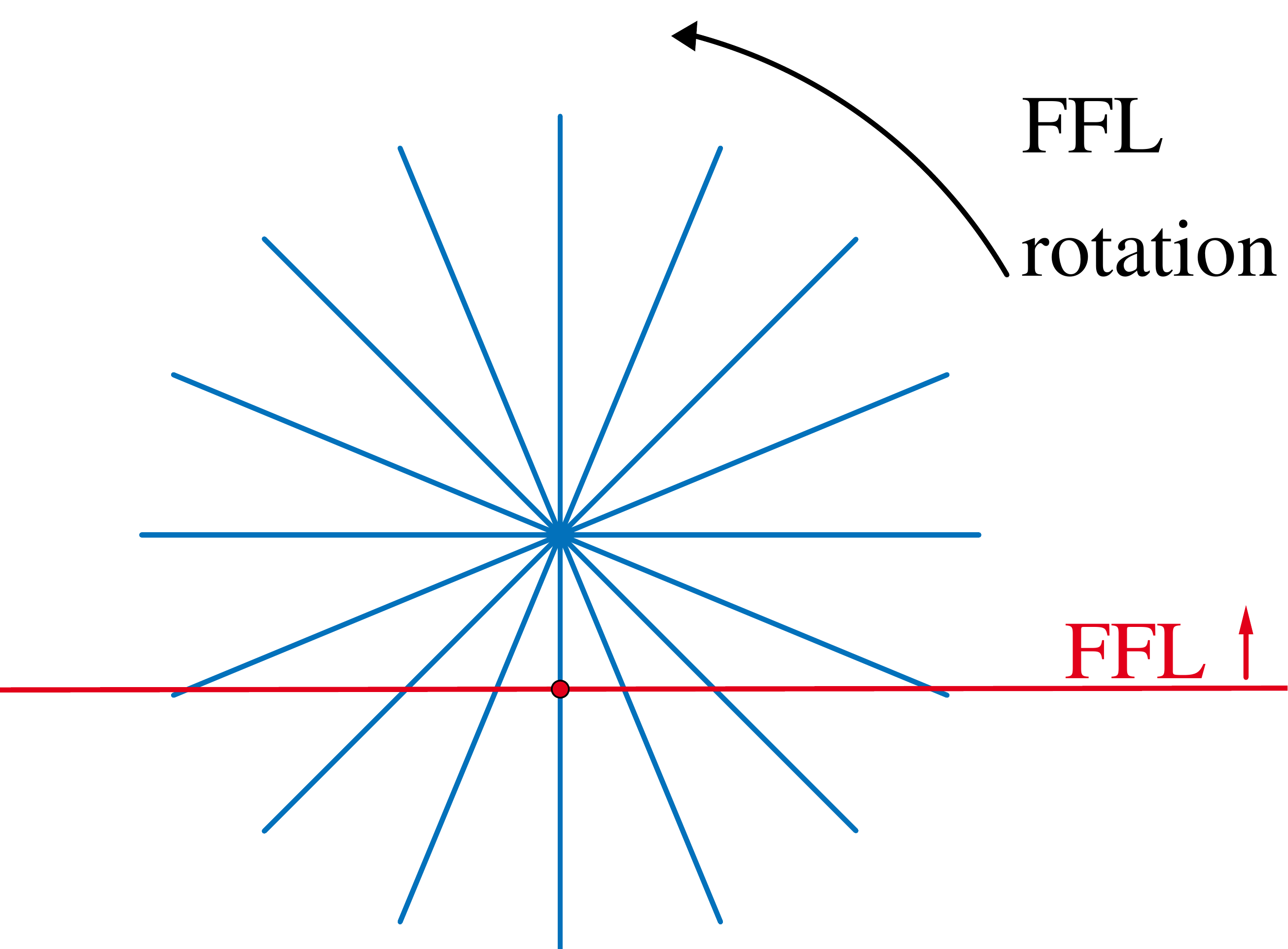}
		\subcaption{\small Sequential translation and rotation of the FFL}
		\label{Fig:SeqRot}
	\end{subfigure}
	\caption{\small Visualization of the scanning geometry.}
	\label{Fig:ScanningGeometry}
\end{figure}
\\
Now, we are able to derive the signal equation for an FFL scanner (cf.~\cite{knopp2011fourier}). From~\eqref{eq:Langevin} we get that $\overline{m}\brackets{-H}=-\overline{m}\brackets{H}$ and thus
by inserting the magnetic field~\eqref{Eqn:MagneticField} into~\eqref{Eqn:ModelinMagnetization}, we obtain 
\begin{equation*}
\overline{\mathbf{m}}\brackets{\mathbf{r}, t} = \overline{m} \left( 
-G \; \mathbf{r}
\cdot \mathbf{e}_\varphi
+ A\Lambda_\varphi(t)\right) \mathbf{e}_\varphi
\end{equation*}
and~\eqref{SignalEquation_static} transforms to
\begin{equation}
u_l(\varphi, t) = -\mu_0 \int_{\R^2} c(\mathbf{r}) \frac{\partial }{\partial t} \overline{m} \left( 
-G \; \mathbf{r}
\cdot \mathbf{e}_\varphi
+ A\Lambda_\varphi(t)\right) \mathbf{e}_\varphi \cdot \mathbf{p}_l\brackets{\mathbf{r}} \dmathInt\mathbf{r}.
\label{Eqn:ModelFFL}
\end{equation}
For the remainder, let $T>0$ denote the measurement time utilized per scanning direction for FFL translation and define $Z_T:=\angleRange \times \left[ 0,T\right] $. Moreover, we assume the MNPs to be contained within the circle $B_R$ of radius $R>0$ around the origin, i.e. $\supp \brackets{c} \subset B_R \subset \R^2$. Thus, henceforth we regard concentrations $c\in L_2\brackets{B_R, \R^+_{0} }$ with continuous extension $c\brackets{\mathbf{r}}:=0$ for $\mathbf{r}\in \R^2 \setminus B_R$ (cf. support condition in~\cite{hahn2021motion}) and define the forward model for field-free line magnetic particle imaging as follows.
\begin{definition}\label{Def:MPI_FFL}
	Define $\mathcal{A}_l: \; L_2\brackets{B_R,\R^+_{0} } \to L_2\brackets{Z_T, \R}$ to be 
	\begin{equation*}
	\mathcal{A}_l c\brackets{\varphi,t} := -\mu_0  \int_{\R^2} c(\mathbf{r}) \frac{\partial}{\partial t} \overline{m} \left( 
	-G \; \mathbf{r}
	\cdot \mathbf{e}_{\varphi}
	+ A\Lambda_\varphi(t)\right) \mathbf{e}_{\varphi} \cdot \mathbf{p}_l\brackets{\mathbf{r}} \dmathInt\mathbf{r}.
	\label{Def:MPIForwardOp}
	\end{equation*}
	Therewith, the forward operator for a MPI-FFL scanner is given as $\mathcal{A}: \; L_2\brackets{B_R,\R^+_{0} } \to L_2\brackets{Z_T, \R^L}$ with $\mathcal{A} c\brackets{\varphi,t} = \brackets{\mathcal{A}_lc\brackets{\varphi,t}}_{l=1,\dots,L}$.
\end{definition}
Thus, to reconstruct the static MNP distribution $c$ from measured data $\mathbf{u} = \brackets{u_l}_{l=1,\dots,L}$ an ill-posed linear inverse problem $$\mathcal{A} c = \mathbf{u} $$ with forward operator $\mathcal{A}$ needs to be solved. \\\\
A prominent breakthrough in history of medical imaging was the invention of the nowadays well-known diagnostic tool computerized tomography. Using the parallel scanning geometry, the radiation source emits a bunch of parallel X-rays passing through the specimen. The corresponding intensity loss is measured via a detector panel. Afterwards, radiation source and detector are rotated. This process is repeated for a suitable amount of different directions in order to enable reconstruction of the tissue density. For a mathematical treatment of CT, we refer to~\cite{natterermathematics}. \\
Looking at Figure~\ref{Fig:ScanningGeometry} with FFLs being substituted by X-rays, it becomes obvious that the FFL shifting resembles the parallel scanning geometry for computerized tomography. However, in CT the X-ray direction is generally determined by the angle between axis and normal vector $\mathbf{e_{\varphi}}$ instead to the ray itself. Thus, as observed in~\cite{knopp2011fourier} there is an angle shift of $\frac{\pi}{2}$ in the definition of $\varphi$ compared to common CT literature. Because of the aforementioned geometric similarity between MPI and CT, a connection between the associated forward operators is self-evident. For CT this operator is given by the Radon transform $\mathcal{R}:\, L_2\brackets{B_R, \R^+_{0} } \to L_2\brackets{Z,\R^+_{0}  }$ with $Z:=\left[ 0,2\pi\right] \times \R$ and
\begin{equation*}
\mathcal{R}c\brackets{\varphi, s} = \int_{\R^2} c\brackets{\mathbf{r}} \delta\brackets{\mathbf{r}\cdot\mathbf{e}_\varphi-s} \dmathInt \mathbf{r}.
\label{Def:RadonTransform}
\end{equation*}
For static particle distributions, it was indeed shown that MPI-FFL data can be traced back to the Radon transform of the particle concentration.
\begin{theorem}[{\cite{knopp2011fourier}}] \label{Thm:FourierSlice}
	Under the assumption of spatially homogeneous receive coil sensitivities and sequential line rotation, signal equation~\eqref{Eqn:ModelFFL} can be reformulated as
	\begin{equation*}
	u_l\brackets{\varphi, t} 
	= -\mu_0 A \Lambda_\varphi'(t) \; \mathbf{e}_\varphi \cdot \mathbf{p}_l
	\Big[   \overline{m}' \left(G \; \cdot\right) * \mathcal{R}c\brackets{\varphi,\cdot} \Big]  \left(\frac{A}{G}\Lambda_\varphi\brackets{t}\right).
	\label{Eqn:FourierSlice}
	\end{equation*}
\end{theorem}
Our main goal is to derive a similar result for the dynamic setting. As a first step, we adapt the forward operator from Definition~\ref{Def:MPI_FFL} towards time-dependent particle concentrations in the next section.

\section{Dynamic particle concentrations} \label{Sec:DynModel}
For dynamic particle concentrations the time derivative in~\eqref{SignalEquation} does not only act on the mean magnetic moment but also on the concentration itself, i.e.
\begin{equation*}
u_l\brackets{t} 
= -\mu_0 \int_{\R^3} \brackets{\overline{\mathbf{m}}\brackets{\mathbf{r}, t}\frac{\partial}{\partial t}c\brackets{\mathbf{r},t}  
	+c\brackets{\mathbf{r},t} \frac{\partial }{\partial t}\overline{\mathbf{m}}\brackets{\mathbf{r}, t}} \cdot   
\mathbf{p}_l\brackets{\mathbf{r}} \dmathInt\mathbf{r}.
\end{equation*}
While in~\cite{gdaniec2020suppression} they propose that the part of the signal equation corresponding to the concentration change is insignificant, the authors of~\cite{brandt2022motion} treated this part by expressing the particle distribution in terms of spline curves. In this work, we neither want to deal with the concentration derivative itself nor do we want to neglect it a priori. Instead, we propose that the object's dynamics can be described via so-called motion functions $\Gamma$, which have been successfully applied for example in dynamic CT (e.g.~\cite{+2014+323+339},~\cite{hahn2017motion},~\cite{hahn2021motion}). The idea is to trace back the dynamic inverse problem of reconstructing the time-dependent concentration $c$ towards the reconstruction of a static reference concentration $c_0$. Therewith, no computation of the concentration derivative is required. For the scope of this work, we take these motion functions as given but note that in practice the motion is a priori unknown and needs to be estimated either beforehand or simultaneously with the concentration reconstruction~\cite{hahn2017motion}. \\\\
We now go back to the specific setting of using a field-free line for spatial encoding and derive a respective dynamic forward model. As we regard sequential line rotation, the state of the particle distribution does not only depend on $t\in\left[ 0,T\right] $ but also on the angle $\varphi\in\angleRange$ determining the direction of the FFL. Suppose that $c\brackets{\cdot,\varphi,t}\in L_2\brackets{B_R,\R^+_{0} }$ for all $\brackets{\varphi,t}\in Z_T$. Further, let $c_0\in L_2\brackets{B_R,\R^+_{0} }$ be a static reference concentration, which could for example be chosen as
\begin{equation*}
c_0\brackets{\mathbf{r}} = c\brackets{\mathbf{r}, \varphi_0, t_0} \label{Eqn:RefConc}
\end{equation*}
for some reference time $t_0\in\left[ 0,T\right] $ and angle $\varphi_0\in\angleRange$. Let $\Gamma: \, \R^2 \times Z_T \to \R^2$ be such that $\Gamma_{\varphi,t}\mathbf{r} := \Gamma\brackets{\mathbf{r}, \varphi,t}$ is a diffeomorphism for all $\brackets{\varphi,t}\in Z_T$ and
\begin{eqnarray}
\textit{(intensity pres.)} \hspace{3.78cm} c\brackets{\mathbf{r},\varphi,t} &=& c_0\brackets{\Gamma_{\varphi,t}\mathbf{r}}, \hspace{3.7cm}  \label{CDyn:IntPres} \\
\textit{(mass pres.)\phantom{......}} \hspace{3.7cm} c\brackets{\mathbf{r},\varphi,t} &=& c_0\brackets{\Gamma_{\varphi,t}\mathbf{r}} \abs{\det \Dmath \Gamma_{\varphi,t}\mathbf{r}}. \hspace{2.8cm} \label{CDyn:MassPres}
\end{eqnarray}
Thereby, assumption~\eqref{CDyn:IntPres} is used for the supposition of intensity and~\eqref{CDyn:MassPres} for mass preservation~\cite{hahn2021motion}. Using
\begin{eqnarray*}
	h_{\varphi,t}\brackets{\mathbf{y}} &:=&
	\begin{cases}
		\abs{\det \Dmath \Gamma^{-1}_{\varphi,t}\mathbf{y}}, \; \text{for $c$ as in~\eqref{CDyn:IntPres}}, \\
		1,\;  \text{for $c$ as in~\eqref{CDyn:MassPres}},
	\end{cases}
\end{eqnarray*}
we summarize~\eqref{CDyn:IntPres} and~\eqref{CDyn:MassPres} via
\begin{eqnarray}
c\brackets{\mathbf{r},\varphi,t} &=& c_0\brackets{\Gamma_{\varphi,t}\mathbf{r}} h_{\varphi,t}\brackets{\Gamma_{\varphi,t}\mathbf{r}}\abs{\det \Dmath \Gamma_{\varphi,t}\mathbf{r}}.\label{Eqn:CDyn}
\end{eqnarray}
\begin{remark}\label{Rem:GammaReg}
	Applying motion models $\Gamma$ may already result in a regularization as it fixes the objects dynamics~\cite{hauptmann2021image}.
\end{remark}
\begin{assumption} \label{Assumption}
	For the remainder, we assume that $\Gamma^{-1}_{\varphi,t}$ as well as $ D\Gamma^{-1}_{\varphi,t}$ are differentiable with respect to t. We denote $\displaystyle \frac{\partial}{\partial t}\Gamma^{-1}_{\varphi,t}=:\brackets{\Gamma^{-1}_{\varphi,t}}'$ with similar notation for $\displaystyle \frac{\partial}{\partial t}D\Gamma^{-1}_{\varphi,t}$ and $\displaystyle \frac{\partial}{\partial t}h_{\varphi,t}$. Further, we suppose for all $\displaystyle \mathbf{y}\in \R^2$ and $ \brackets{\varphi,t}\in Z_T$ that
	\begin{equation}
	\norm{\brackets{\Gamma^{-1}_{\varphi,t}}'\mathbf{y}} \leq \frac{A}{G} C_{\varphi,t}, \quad
	\abs{h'_{\varphi,t}\brackets{\mathbf{y}}} \leq D_{\varphi,t}
	\label{Eqn:MotionBounds}
	\end{equation}
	for some constants $C_{\varphi,t},\;D_{\varphi,t} >0.$
\end{assumption}
\begin{remark}
	Assumption~\ref{Assumption} guarantees that the signal equation is well-defined. Note that the bounds in~\eqref{Eqn:MotionBounds} are naturally satisfied as in practice the speed and extent of motion are limited due to physiological constraints.
\end{remark}
Using~\eqref{Eqn:CDyn} and substituting $\mathbf{y} := \Gamma_{\varphi,t}\mathbf{r}$, the MPI-FFL signal equation transforms to
\begin{eqnarray*}
	u_l\brackets{\varphi, t} =&-&\mu_0 \dmath{t} \int_{\R^2} c(\mathbf{r},\varphi,t) \overline{m} \left( 
	-G \; \mathbf{r}
	\cdot \mathbf{e}_\varphi
	+ A\Lambda_\varphi(t)\right) \mathbf{e}_\varphi \cdot \mathbf{p}_l\brackets{\mathbf{r}} \dmathInt\mathbf{r} \nonumber \\
	=& -&\mu_0 \dmath{t} \int_{\R^2} c_0\brackets{\Gamma_{\varphi,t}\mathbf{r}}
	h_{\varphi,t}\brackets{\Gamma_{\varphi,t}\mathbf{r}} \abs{\det \Dmath \Gamma_{\varphi,t}\mathbf{r}} \nonumber \\
	&\phantom{-}&\hspace{1.3cm}\times \; \overline{m} \left( 
	-G \; \mathbf{r}
	\cdot \mathbf{e}_\varphi
	+ A\Lambda_\varphi(t)\right) \mathbf{e}_\varphi \cdot \mathbf{p}_l\brackets{\mathbf{r}} \dmathInt\mathbf{r}  \nonumber \\
	=& -&\mu_0 \,  \int_{\R^2} c_0\brackets{\mathbf{y}}  \frac{\partial}{\partial t}   h_{\varphi,t}\brackets{\mathbf{y}} \overline{m}\brackets{-G \; \Gamma_{\varphi,t}^{-1}\mathbf{y} \cdot \mathbf{e}_{\varphi} + A\Lambda_\varphi(t)} \mathbf{e}_\varphi \cdot \mathbf{p}_l\brackets{\Gamma_{\varphi,t}^{-1}\mathbf{y}}  \dmathInt \mathbf{y}. \nonumber \\ \label{Eqn:DynSignalEqn}
\end{eqnarray*}
With this all time-dependencies lie in quantities supposed to be known. 

\begin{definition}\label{Def:Dyn_MPI_FFL}
	Define $\mathcal{A}^\Gamma_l: \; L_2\brackets{B_R,\R^+_{0} } \to L_2\brackets{Z_T, \R}$ to be
	%\begin{equation}
	\begin{multline}
	\mathcal{A}^\Gamma_l c_0\brackets{\varphi,t} := -\mu_0  \int_{\R^2} c_0\brackets{\mathbf{y}}  \frac{\partial}{\partial t}   h_{\varphi,t}\brackets{\mathbf{y}} \overline{m}\brackets{-G \; \Gamma_{\varphi,t}^{-1}\mathbf{y} \cdot \mathbf{e}_{\varphi} + A\Lambda_\varphi(t)} \mathbf{e}_\varphi \cdot \mathbf{p}_l\brackets{\Gamma_{\varphi,t}^{-1}\mathbf{y}}  \dmathInt \mathbf{y}.
	\label{Def:Dyn_MPIForwardOp}
	\end{multline}
	%\end{equation}
	Therewith, the forward operator for MPI-FFL scanner is given as $\mathcal{A}^\Gamma: \; L_2\brackets{B_R,\R^+_{0} } \to L_2\brackets{Z_T, \R^L}$ with $\mathcal{A}^\Gamma c_0\brackets{\varphi,t} = \brackets{\mathcal{A}^\Gamma_l c_0\brackets{\varphi,t}}_{l=1,\dots,L}$.
\end{definition}
Hence, for concentration reconstruction in the dynamic case, we need to solve the linear ill-posed inverse problem $\mathcal{A}^\Gamma c_0 = \mathbf{u} $ with measured data $\mathbf{u} = \brackets{u_l}_{l=1,\dots,L}$ and forward operator $\mathcal{A}^\Gamma$. 

\section{Relation between MPI and Radon data}\label{Sec:Relation}
We aim at a formulation of MPI data in terms of the Radon transform of the particle concentration also when being confronted with moving objects. To this end, as for the MPI forward operator we derive an adapted version of the Radon transform (cf. \cite{hahn2021motion}), i.e.
\begin{eqnarray}
\mathcal{R} c\brackets{\varphi, t, s} &=& \int_{\R^2} c\brackets{\mathbf{r},\varphi,t} \delta\brackets{\mathbf{r} \cdot \mathbf{e}_{\varphi} - s} \dmathInt \mathbf{r} \nonumber \\  
&=& \int_{\R^2} c_0\brackets{\Gamma_{\varphi,t}\mathbf{r}} h_{\varphi,t}\brackets{\Gamma_{\varphi,t}\mathbf{r}} \abs{\det \Dmath \Gamma_{\varphi,t}\mathbf{r}} \delta\brackets{\mathbf{r} \cdot \mathbf{e}_{\varphi} - s} \dmathInt \mathbf{r} \nonumber\\
&=& \int_{\R^2} c_0\brackets{\mathbf{y}} h_{\varphi,t}\brackets{\mathbf{y}} \delta\brackets{\Gamma_{\varphi,t}^{-1}\mathbf{y} \cdot \mathbf{e}_{\varphi} - s} \dmathInt \mathbf{y} =: \mathcal{R}^\Gamma c_0\brackets{\varphi, t, s}. \nonumber \\
\label{Eqn:DynRadonMPI}
\end{eqnarray}
Additionally, we define a weighted version of this dynamic Radon transform
\begin{eqnarray}
\mathcal{R}_w^\Gamma c_0\brackets{\varphi, t, s} := \int_{\R^2} w_{\varphi,t}\brackets{\mathbf{y}}c_0\brackets{\mathbf{y}} h_{\varphi,t}\brackets{\mathbf{y}} \delta\brackets{\Gamma_{\varphi,t}^{-1}\mathbf{y} \cdot \mathbf{e}_{\varphi} - s} \dmathInt \mathbf{y}
\label{Eqn:WeightedRadon}
\end{eqnarray}
with bounded weight function $w_{\varphi,t}: \, \R^2 \to \R$. In the following, we further denote
\begin{eqnarray*}
	\mathbf{R}^{{\varphi}}=\begin{pmatrix}
		\cos{\varphi}&-\sin{\varphi} \\
		\sin{\varphi} &\phantom{-}\cos{\varphi}
	\end{pmatrix}, \quad
	\mathbf{e}_{{\varphi}}=\begin{pmatrix}
		-\sin{\varphi} \\
		\phantom{-}\cos{\varphi}
	\end{pmatrix}, \quad
	\mathbf{e}^\perp_{{\varphi}}=-\begin{pmatrix}
		\cos{\varphi} \\
		\sin{\varphi}
	\end{pmatrix}.
\end{eqnarray*}
\begin{theorem} \label{Thm:DynConv}
	Under the assumption of spatially homogeneous receive coil sensitivities and sequential line rotation, we can reformulate~\eqref{Def:Dyn_MPIForwardOp} as
	\begin{equation*}
	\mathcal{A}^\Gamma_l  =  \mathcal{K}_{1,l} \circ \mathcal{R}^\Gamma + \mathcal{K}_{2,l} \circ \mathcal{R}^\Gamma_\alpha + \mathcal{K}_{3,l} \circ \mathcal{R}^\Gamma_\beta
	\label{Eqn:Thm}
	\end{equation*}
	with convolution operators $\mathcal{K}_{i,l}: \; L_2\brackets{Z_T\times \R ,\R} \to L_2\brackets{Z_T,\R}$ for $i=1,2,3$ and $l\in \left\lbrace 1,\dots,L\right\rbrace $ defined as
	\begin{eqnarray*}
		\mathcal{K}_{1,l} f\brackets{\varphi,t}
		&=&-\mu_0 \; \mathbf{e}_{\varphi}\cdot \mathbf{p}_l \; A \Lambda_\varphi'(t) \; \overline{m}' \left(G \; \cdot\right) * f\brackets{\varphi, t, \cdot}  \left(s_{\varphi,t}\right), \\
		\mathcal{K}_{2,l} f\brackets{\varphi,t}
		&=& \phantom{-} \mu_0  \; \mathbf{e}_{\varphi}\cdot \mathbf{p}_l \; G \; \overline{m}' \left(G \; \cdot\right) * f\brackets{\varphi, t, \cdot}  \left(s_{\varphi,t}\right), \\
		\mathcal{K}_{3,l} f\brackets{\varphi,t}
		&=& -\mu_0 \; \mathbf{e}_{\varphi}\cdot \mathbf{p}_l \; \overline{m} \left(G \; \cdot\right) 
		* f\brackets{\varphi, t, \cdot}  \left(s_{\varphi,t}\right)
	\end{eqnarray*}
	and weight functions
	\begin{eqnarray*} 
		\alpha_{\varphi,t}\brackets{\mathbf{y}}
		&=& \brackets{\Gamma_{\varphi,t}^{-1}}'\mathbf{y} \cdot \mathbf{e}_{\varphi}, \quad
		\beta_{\varphi,t}\brackets{\mathbf{y}}
		=  \frac{h_{\varphi,t}'\brackets{\mathbf{y}}}{h_{\varphi,t}\brackets{\mathbf{y}}}.
	\end{eqnarray*}
\end{theorem}
\begin{proof}
	Computing the derivative in~\eqref{Def:Dyn_MPIForwardOp}, we obtain
	\begin{eqnarray*}
		\mathcal{A}^\Gamma_l c_0\brackets{\varphi,t}
		=& - &\mu_0 \; \mathbf{e_\varphi} \cdot \mathbf{p}_l \; A\Lambda_\varphi'\brackets{t} \,  \int_{\R^2} c_0\brackets{\mathbf{y}} h_{\varphi,t}\brackets{\mathbf{y}} \overline{m}^\prime \left( -G \; \Gamma_{\varphi,t}^{-1}\mathbf{y} \cdot \mathbf{e}_{\varphi} + A\Lambda_\varphi(t)\right)  \dmathInt \mathbf{y} \\
		&+& \mu_0 \; \mathbf{e_\varphi} \cdot \mathbf{p}_l \; G   \,  \int_{\R^2} c_0\brackets{\mathbf{y}} h_{\varphi,t}\brackets{\mathbf{y}} \\
		&\phantom{+}& \hspace{2.5cm} \times \;  \overline{m}^\prime \left( -G \; \Gamma_{\varphi,t}^{-1}\mathbf{y} \cdot \mathbf{e}_{\varphi} + A\Lambda_\varphi(t)\right)  \brackets{\Gamma_{\varphi,t}^{-1}}'\mathbf{y} \cdot \mathbf{e}_{\varphi}  \dmathInt \mathbf{y} \\
		&-& \mu_0 \; \mathbf{e_\varphi} \cdot \mathbf{p}_l   \,  \int_{\R^2} c_0\brackets{\mathbf{y}} \overline{m} \left( -G \; \Gamma_{\varphi,t}^{-1}\mathbf{y} \cdot \mathbf{e}_{\varphi} + A\Lambda_\varphi(t)\right) h_{\varphi,t}'\brackets{\mathbf{y}} \dmathInt \mathbf{y} \\
		= &-&\mu_0 \; \mathbf{e_\varphi} \cdot \mathbf{p}_l \Big( A\Lambda_\varphi'\brackets{t} \; I\brackets{\varphi,t}
		- G \; II\brackets{\varphi,t}
		+ III\brackets{\varphi,t} \Big) .
	\end{eqnarray*}
	We look at the integrals separately and proceed similar to \cite{knopp2011fourier} by substituting $\mathbf{R}^{-\varphi}\Gamma_{\varphi,t}^{-1}\mathbf{y} =: \mathbf{r}' =: \brackets{v',s'}^T$. It holds
	\[\mathbf{R}^{\varphi}\mathbf{r}' = \mathbf{R}^{\varphi}  \begin{pmatrix}v' \\s'\end{pmatrix} = s'\,\mathbf{e}_{\varphi} - v'\,\mathbf{e}_{\varphi}^\perp\] 
	yielding
	\begin{eqnarray*}
		\text{FFL}\brackets{\mathbf{e_{\varphi}}, s'} &=& \left\lbrace \mathbf{r} \in \R^2 : \; \mathbf{r} \cdot \mathbf{e}_{\varphi} = s'  \right\rbrace = 
		\left\lbrace s'\,\mathbf{e}_{\varphi} - v'\,\mathbf{e}_{\varphi}^\perp : \; v' \in \R\right\rbrace \\
		&=& \left\lbrace \mathbf{R}^{\varphi}  \begin{pmatrix}
			v' \\
			s'
		\end{pmatrix} : \; v' \in \R\right\rbrace.
	\end{eqnarray*}
	Together with~\eqref{Eqn:DynRadonMPI} we obtain 
	\begin{eqnarray*}
		I\brackets{\varphi,t} 
		&:=& \int_{\R^2} c_0\brackets{\mathbf{y}} h_{\varphi,t}\brackets{\mathbf{y}} \overline{m}^\prime \left( -G \; \Gamma_{\varphi,t}^{-1}\mathbf{y} \cdot \mathbf{e}_{\varphi} + A\Lambda_\varphi(t)\right)  \dmathInt \mathbf{y}\\
		&\hspace{1mm}=&\int_{\R^2} 
		c_0\brackets{\Gamma_{\varphi,t} \mathbf{R}^{\varphi}\mathbf{r}'} h_{\varphi,t}\brackets{\Gamma_{\varphi,t} \mathbf{R}^{\varphi}\mathbf{r}'} \abs{\det \Dmath \Gamma_{\varphi,t}\brackets{\mathbf{R}^{\varphi}\mathbf{r}'}} \\
		&\phantom{\hspace{1mm}=}& \hspace{0.5cm} \times \; \overline{m}^\prime \left( -G \; \mathbf{R}^{{\varphi}} \mathbf{r}' \cdot  \mathbf{e}_\varphi + A\Lambda_\varphi(t)\right)  \dmathInt\mathbf{r}' \\
		&\hspace{1mm}=&  \int_{\R} \overline{m}^\prime \left( -G \; s' + A\Lambda_\varphi(t)\right)\\
		&\phantom{\hspace{1mm}=}& \hspace{0.5cm} \times \; \int_{\R^2}
		c_0\brackets{\Gamma_{\varphi,t} \mathbf{r}} h_{\varphi,t}\brackets{\Gamma_{\varphi,t}\mathbf{r}} \abs{\det \Dmath \Gamma_{\varphi,t}\mathbf{r}} \delta\brackets{\mathbf{r}\cdot \mathbf{e}_\varphi-s'}   \dmathInt \mathbf{r} \dmathInt s'  \\
		&\hspace{1mm}=& \int_{\mathbb{R}}  
		\mathcal{R}^{\Gamma}c_0\left( \varphi, t, s'\right)  \overline{m}' \left( -G \; s'
		+ A\Lambda_\varphi(t)\right)  \dmathInt s' \\
		&\hspace{1mm}=& \Big[ \overline{m}' \left(G \; \cdot\right) * \mathcal{R}^\Gamma c_0\brackets{\varphi,t,\cdot} \Big] \left(\frac{A}{G}\Lambda_\varphi\brackets{t}\right).
	\end{eqnarray*}
	Analogously we get for the second and third integral 
	\begin{eqnarray*}
		II\brackets{\varphi,t}	
		&:=& \int_{\R^2} c_0\brackets{\mathbf{y}} h_{\varphi,t}\brackets{\mathbf{y}} \overline{m}^\prime \left( -G \; \Gamma_{\varphi,t}^{-1}\mathbf{y} \cdot \mathbf{e}_{\varphi} + A\Lambda_\varphi(t)\right)  \brackets{\Gamma_{\varphi,t}^{-1}}'\mathbf{y} \cdot \mathbf{e}_{\varphi}  \dmathInt \mathbf{y} \\
		&\hspace{1mm}=&\int_{\R^2} 
		c_0\brackets{\Gamma_{\varphi,t} \mathbf{R}^{\varphi}\mathbf{r}'} h_{\varphi,t}\brackets{\Gamma_{\varphi,t} \mathbf{R}^{\varphi}\mathbf{r}'} \abs{\det \Dmath \Gamma_{\varphi,t}\brackets{\mathbf{R}^{\varphi}\mathbf{r}'}} \\
		&\hspace{1mm}\phantom{=}& \hspace{0.5cm}\times \; \overline{m}^\prime \left( -G \; \mathbf{R}^{{\varphi}} \mathbf{r}' \cdot  \mathbf{e}_\varphi + A\Lambda_\varphi(t)\right)  \brackets{\Gamma_{\varphi,t}^{-1}}' \brackets{\Gamma_{\varphi,t}\mathbf{R}^{\varphi}\mathbf{r}'} \cdot \mathbf{e}_{\varphi} \dmathInt\mathbf{r}' \\
		%			&\hspace{1mm}=& \int_{\R} \alpha_{\varphi,t}\brackets{s'} \overline{m}^\prime \left( -G \; s' + A\Lambda_\varphi(t)\right) \int_{\R^2}
		%			c_0\brackets{\Gamma_{\varphi,t} \mathbf{r}} \delta\brackets{\mathbf{r}\cdot \mathbf{e}_\varphi-s'}   \dmathInt \mathbf{r} \dmathInt s' \\
		&\hspace{1mm}=&  \int_{\mathbb{R}}  
		\mathcal{R}_\alpha^{\Gamma}c_0\left( \varphi, t, s'\right)  \overline{m}' \left( -G \; s'
		+ A\Lambda_\varphi(t)\right)  \dmathInt s'\\
		&\hspace{1mm}=& \Big[ \overline{m}' \left(G \; \cdot\right) * \mathcal{R}_\alpha^\Gamma c_0\brackets{\varphi,t,\cdot}\Big] \left(\frac{A}{G}\Lambda_\varphi\brackets{t}\right)
	\end{eqnarray*}
	as well as 
	\begin{eqnarray*}
		III\brackets{\varphi,t}	
		&:=& \int_{\R^2} c_0\brackets{\mathbf{y}} \overline{m} \left( -G \; \Gamma_{\varphi,t}^{-1}\mathbf{y} \cdot \mathbf{e}_{\varphi} + A\Lambda_\varphi(t)\right) h_{\varphi,t}'\brackets{\mathbf{y}} \dmathInt \mathbf{y}\\
		&\hspace{1mm}=& \int_{\R^2} 
		c_0\brackets{\Gamma_{\varphi,t} \mathbf{R}^{\varphi}\mathbf{r}'} h_{\varphi,t}\brackets{\Gamma_{\varphi,t} \mathbf{R}^{\varphi}\mathbf{r}'} \abs{\det \Dmath \Gamma_{\varphi,t}\brackets{\mathbf{R}^{\varphi}\mathbf{r}'}}  \\
		&\hspace{1mm}\phantom{=}& \hspace{0.5cm} 
		\times \; \overline{m} \left( -G \; \mathbf{R}^{{\varphi}} \mathbf{r}' \cdot  \mathbf{e}_\varphi + A\Lambda_\varphi(t)\right)   \frac{h_{\varphi,t}'\brackets{\Gamma_{\varphi,t} \mathbf{R}^{\varphi}\mathbf{r}'}}{h_{\varphi,t}\brackets{\Gamma_{\varphi,t} \mathbf{R}^{\varphi}\mathbf{r}'}} \dmathInt\mathbf{r}' \\
		%			&\hspace{1mm}=& \int_{\R^2} 
		%			c_0\brackets{\Gamma_{\varphi,t} \mathbf{R}^{\varphi}\mathbf{r}'}  \overline{m} \left( -G \; \mathbf{R}^{{\varphi}} \mathbf{r}' \cdot  \mathbf{e}_\varphi + A\Lambda_\varphi(t)\right) 
		%			\beta_{\varphi,t}\brackets{\mathbf{R}^{{\varphi}} \mathbf{r}' \cdot  \mathbf{e}_\varphi} \dmathInt\mathbf{r}' \\
		&\hspace{1mm}=& \int_{\mathbb{R}}   
		\mathcal{R}_\beta^{\Gamma}c_0\left( \varphi, t, s'\right) \overline{m} \left( -G \; s'
		+ A\Lambda_\varphi(t)\right)  \dmathInt s' \\
		&\hspace{1mm}=&  \Big[ \overline{m} \left(G \; \cdot\right) * \mathcal{R}_\beta^\Gamma c_0\brackets{\varphi,t,\cdot}\Big] \left(\frac{A}{G}\Lambda_\varphi\brackets{t}\right).
	\end{eqnarray*}
	Thus, we have
	\begin{eqnarray*}
		\mathcal{A}^\Gamma_l c_0\brackets{\varphi,t}
		&=& - \mu_0 \; \mathbf{e_\varphi} \cdot \mathbf{p}_l \Big( A\Lambda_\varphi'\brackets{t} \, I\left( \varphi, t\right)  
		- G   \, II\left( \varphi,t\right) + III\left( \varphi,t\right)\Big) \\
		&=& \left[ \mathcal{K}_{1,l} \brackets{\mathcal{R}^\Gamma c_0} + \mathcal{K}_{2,l} \brackets{\mathcal{R}^\Gamma_\alpha c_0} + \mathcal{K}_{3,l} \brackets{ \mathcal{R}^\Gamma_\beta c_0}\right] \brackets{\varphi,t}.
		%			=& - &\mu_0 \; \mathbf{e_\varphi} \cdot \mathbf{p}_l \; A\Lambda_\varphi'\brackets{t} \, \left( \brackets{\overline{m}}' \left(G \; \cdot\right) * \mathcal{R}^\Gamma c_0\brackets{\varphi,t,\cdot}\right) \left(s_t\right)  \\
		%			&+& \mu_0 \; \mathbf{e_\varphi} \cdot \mathbf{p}_l \; G   \, \left( \brackets{\overline{m}}' \left(G \; \cdot\right) * \mathcal{R}_\alpha^\Gamma c_0\brackets{\varphi,t,\cdot}\right) \left(s_t\right)  \\
		%			&-& \mu_0 \; \mathbf{e_\varphi} \cdot \mathbf{p}_l   \,  \left( \overline{m} \left(G \; \cdot\right) * \mathcal{R}_\beta^\Gamma c_0\brackets{\varphi,t,\cdot}\right) \left(s_t\right).
	\end{eqnarray*}
\end{proof}
\begin{remark}
	Note that for static objects we have $\Gamma_{\varphi,t} = \text{Id}$ and Theorem~\ref{Thm:DynConv} reduces to Theorem~\ref{Thm:FourierSlice}.  
\end{remark}
\begin{corollary} \label{corollary}
	In case that there exist scalar functions $\tilde{\alpha}_{\varphi,t},\tilde{\beta}_{\varphi,t}: \, \R \to \R$ such that 
	\begin{eqnarray*}
		\alpha_{\varphi,t}\brackets{\mathbf{y}} 
		&=& \tilde{\alpha}_{\varphi,t}\brackets{\Gamma^{-1}_{\varphi,t} \mathbf{y} \cdot \mathbf{e}_\varphi}, \quad
		\beta_{\varphi,t}\brackets{\mathbf{y}} 
		= \tilde{\beta}_{\varphi,t}\brackets{\Gamma^{-1}_{\varphi,t} \mathbf{y} \cdot \mathbf{e}_\varphi}
	\end{eqnarray*}
	we can write
	\begin{equation}
	\mathcal{A}^\Gamma_l  =  \mathcal{K}_{l} \circ \mathcal{R}^\Gamma
	\label{Eqn:Cor}
	\end{equation}
	with $\mathcal{K}_{l}: \; L_2\brackets{Z_T\times \R,\R} \to L_2\brackets{Z_T,\R}$ for $l\in \left\lbrace 1,\dots,L\right\rbrace $ 
	\begin{eqnarray*}
		\mathcal{K}_{l} f \brackets{\varphi,t}
		&=& \left[ \mathcal{K}_{1,l}f+ \mathcal{K}_{2,l} \brackets{\tilde{\alpha}_{\varphi,t}f} + \mathcal{K}_{3,l} \brackets{\tilde{\beta}_{\varphi,t}f}\right] \brackets{\varphi,t}.
	\end{eqnarray*}
\end{corollary}
\begin{proof}
	Under the given assumption, it directly follows from the definition of the weighted Radon transform~\eqref{Eqn:WeightedRadon} that
	\begin{eqnarray*}
		\mathcal{R}_\alpha^\Gamma c_0\brackets{\varphi, t, s} 
		&=& \tilde{\alpha}_{\varphi,t}\brackets{s} \mathcal{R}^\Gamma c_0\brackets{\varphi, t, s}, \quad
		\mathcal{R}_\beta^\Gamma c_0\brackets{\varphi, t, s}
		=\tilde{\beta}_{\varphi,t}\brackets{s} \mathcal{R}^\Gamma c_0\brackets{\varphi, t, s}
	\end{eqnarray*}
	and the representation~\eqref{Eqn:Cor} is then a consequence of the last theorem.
\end{proof}
\begin{remark} \label{Rem:beta}
	For the assumption of mass preservation it holds $h_{\varphi,t}\brackets{\mathbf{y}} =  1$ and thus
	\begin{equation*}
	\beta_{\varphi,t}\brackets{\mathbf{y}} 
	=  \frac{h_{\varphi,t}'\brackets{\mathbf{y}}}{h_{\varphi,t}\brackets{\mathbf{y}}} = 0 \quad \Longrightarrow \quad \mathcal{K}_{3,l} \circ \mathcal{R}^\Gamma_\beta = 0,
	\end{equation*}
	i.e. the third component of the signal equation vanishes. For intensity preservation we have $h_{\varphi,t}\brackets{\mathbf{y}} = \abs{\det \Dmath \Gamma^{-1}_{\varphi,t}\mathbf{y}}$ and
	\begin{equation*}
	h'_{\varphi,t}\brackets{\mathbf{y}} 
	= \frac{\partial}{\partial t} \abs{\det \Dmath \Gamma^{-1}_{\varphi,t}\mathbf{y}}
	= \frac{\det \Dmath \Gamma^{-1}_{\varphi,t}\mathbf{y} \, \frac{\partial}{\partial t} \det \Dmath \Gamma^{-1}_{\varphi,t}\mathbf{y}}{\abs{\det \Dmath \Gamma^{-1}_{\varphi,t}\mathbf{y}}}.
	\end{equation*}
	According to Jacobi's formula it holds
	\begin{equation*}
	\frac{\partial}{\partial t} \det \Dmath \Gamma^{-1}_{\varphi,t}\mathbf{y}
	= \det \Dmath \Gamma^{-1}_{\varphi,t}\mathbf{y} \; \tr{\brackets{\Dmath \Gamma^{-1}_{\varphi,t}}^{-1}\brackets{\mathbf{y}} \, \brackets{\Dmath \Gamma^{-1}_{\varphi,t}}'\brackets{\mathbf{y}}},
	\end{equation*}
	which leads to
	\begin{eqnarray*}
		\beta_{\varphi,t}\brackets{\mathbf{y}} 
		&=& \frac{\frac{\partial}{\partial t} \abs{\det \Dmath \Gamma^{-1}_{\varphi,t}\mathbf{y}}}{\abs{\det \Dmath \Gamma^{-1}_{\varphi,t}\mathbf{y}}}
		= \frac{\det \Dmath \Gamma^{-1}_{\varphi,t}\mathbf{y} \, \frac{\partial}{\partial t} \det \Dmath \Gamma^{-1}_{\varphi,t}\mathbf{y}}{\abs{\det \Dmath \Gamma^{-1}_{\varphi,t}\mathbf{y}}^2}
		= \frac{\frac{\partial}{\partial t} \det \Dmath \Gamma^{-1}_{\varphi,t}\mathbf{y}}{\det \Dmath \Gamma^{-1}_{\varphi,t}\mathbf{y}}\\
		&=& \tr{\brackets{\Dmath \Gamma^{-1}_{\varphi,t}}^{-1}\brackets{\mathbf{y}} \, \brackets{\Dmath \Gamma^{-1}_{\varphi,t}}'\brackets{\mathbf{y}}}
		= \tr{\Dmath \Gamma_{\varphi,t}\brackets{\Gamma^{-1}_{\varphi,t}\mathbf{y}} \, \brackets{\Dmath \Gamma^{-1}_{\varphi,t}}'\brackets{\mathbf{y}}}.
	\end{eqnarray*}
\end{remark}
Next, we want to regard the specific case of affine motions. Hence, suppose $\Gamma^{-1}_{\varphi,t}\mathbf{y} = \mathbf{A}_{\varphi,t} \mathbf{y} + \mathbf{b}_{\varphi,t}$ for some $\mathbf{b}_{\varphi,t} \in \R^{2},\;\mathbf{A}_{\varphi,t} \in \R^{2 \times 2}$ being differentiable with respect to time and satisfying $\det \mathbf{A}_{\varphi,t}\neq 0$. We know from the last remark that for the mass preserving setting we have $\beta_{\varphi,t} = 0$ and for the intensity preserving case we get
\begin{eqnarray}
\beta_{\varphi,t}\brackets{\mathbf{y}}
= \tr{\mathbf{A}^{-1}_{\varphi,t} \, \mathbf{A}'_{\varphi,t}}, \label{Ex:beta}
\end{eqnarray}
which is constant for fixed angle $\varphi$ and time $t$. Even for affine motions the weighting function $\alpha_{\varphi,t}\brackets{\mathbf{y}}$ cannot be written as in Corollary~\ref{corollary} in general but we can extract the part, which does not depend on $\mathbf{y}$
\begin{equation*}
\alpha_{\varphi,t}\brackets{\mathbf{y}} = \brackets{\mathbf{A}'_{\varphi,t}\mathbf{y} + \mathbf{b}_{\varphi,t}'}\cdot\mathbf{e}_\varphi =: \mathbf{A}'_{\varphi,t}\mathbf{y} \cdot\mathbf{e}_\varphi + \tilde{\alpha}_{\varphi,t}
\end{equation*}
and write
\begin{equation*}
\mathcal{A}^\Gamma_l  =  \mathcal{K}_{l} \circ \mathcal{R}^\Gamma + \mathcal{K}_{2,l} \circ \mathcal{R}^\Gamma_{\mathbf{A}'_{\varphi,t}\mathbf{y} \cdot\mathbf{e}_\varphi}.
\end{equation*}
\begin{example} \label{Ex:alpha}
	Let for $a_{\varphi,t} \neq 0$ being differentiable with respect to $t$ the deformation matrix be given as 
	\begin{equation*}
	\mathbf{A}_{\varphi,t} := \mathbf{R}^\varphi \begin{pmatrix}
	a_{\varphi,t}& 0 \\
	0 & 1
	\end{pmatrix} \mathbf{R}^{-\varphi}
	\end{equation*}
	resulting in an object deformation restricted to directions orthogonal to the FFL \textcolor{black}{movement}. It holds
	\begin{eqnarray*}
		\mathbf{A}_{\varphi,t} \mathbf{y} = 
		%			\mathbf{R}^\varphi \begin{pmatrix}
		%				a_{\varphi,t}& 0 \\
		%				0 & 1
		%			\end{pmatrix} \mathbf{R}^{-\varphi} \mathbf{y} = 	
		a_{\varphi,t}\brackets{ \mathbf{y} \cdot \mathbf{e}_\varphi^\perp } \mathbf{e}_\varphi^\perp + \brackets{\mathbf{y} \cdot \mathbf{e}_\varphi } \mathbf{e}_\varphi  			 
	\end{eqnarray*}
	yielding $\mathbf{A}'_{\varphi,t} \mathbf{y} = a'_{\varphi,t} \brackets{ \mathbf{y} \cdot \mathbf{e}_\varphi^\perp } \mathbf{e}_\varphi^\perp$ and thus $\mathbf{A}'_{\varphi,t}\mathbf{y} \cdot \mathbf{e}_{\varphi}=0.$ We further have 
	\begin{eqnarray*}
		\mathcal{R}^\Gamma c_0\brackets{\varphi, t, s} 
		&=& \int_{\R^2} c_0\brackets{\mathbf{y}} h_{\varphi,t}\brackets{\mathbf{y}} \delta\brackets{\mathbf{A}_{\varphi,t}\mathbf{y} \cdot \mathbf{e}_{\varphi} + \mathbf{b}_{\varphi,t} \cdot \mathbf{e}_{\varphi}  - s} \dmathInt \mathbf{y} \\
		&=& \int_{\R^2} c_0\brackets{\mathbf{y}} \textcolor{black}{\abs{a_{\varphi,t}}} \delta\brackets{\mathbf{y} \cdot \mathbf{e}_\varphi + \mathbf{b}_{\varphi,t} \cdot \mathbf{e}_{\varphi}  - s} \dmathInt \mathbf{y} \\
		&=& \textcolor{black}{\abs{a_{\varphi,t}}} \mathcal{R} \textcolor{black}{c_0} \brackets{\varphi, t, s-\mathbf{b}_{\varphi,t} \cdot \mathbf{e}_{\varphi}}.
	\end{eqnarray*}
\end{example}
We conclude from the last example that for object translation and further for objects deforming only orthogonal to the FFL movement  $\alpha_{\varphi,t}\brackets{\mathbf{y}}=\tilde{\alpha}_{\varphi,t}$ is spatially independent.
\begin{example}
	As a second example, we regard rigid phantom rotations $\Gamma^{-1}_{\varphi,t}\mathbf{y}=\mathbf{R}^{a_{\varphi,t}}\mathbf{y}$. We obtain
	\begin{equation*}
	\Gamma^{-1}_{\varphi,t}\mathbf{y} \cdot \mathbf{e}_\varphi = \mathbf{y} \cdot \mathbf{e}_{\varphi - a_{\varphi,t}},\quad \det \Dmath\Gamma^{-1}_{\varphi,t}\mathbf{y}=1
	\end{equation*}
	and thus
	\begin{equation*}
	\alpha_{\varphi,t}=\brackets{\Gamma^{-1}_{\varphi,t}}'\mathbf{y} \cdot \mathbf{e}_\varphi = -a'_{\varphi,t}\mathbf{y} \cdot \mathbf{e}^\perp_{\varphi - a_{\varphi,t}}, \quad \beta_{\varphi,t} = 0.
	\end{equation*}
\end{example}
\begin{remark}
	The previous example builds a link towards the setting of static phantoms together with the simultaneous line rotation scanning geometry, i.e. the FFL rotates concurrently with its translation, investigated in~\cite{https://doi.org/10.48550/arxiv.2211.11683}.
\end{remark}
\begin{remark}
	Being able to derive dynamic Radon data $\mathcal{R}^\Gamma c_0$ from MPI measurements not only allows application of corresponding reconstruction methods from dynamic CT but also gives access towards results regarding the amount of information contained in the data as e.g. given in~(\cite{doi:10.1137/16M1057917},~\cite{hahn2021microlocal}). For example choosing $a_{\varphi,t} = \varphi$ for the phantom rotation ends in highly insufficient data since we only get information for one angle though data is acquired for different FFL directions. 
\end{remark}
Finally, we want to develop estimates for $\mathcal{K}_{2,l}\circ \mathcal{R}^\Gamma_\alpha$ and $\mathcal{K}_{3,l}\circ \mathcal{R}^\Gamma_\beta$ for $l\in\left\lbrace 1,\dots,L\right\rbrace$. 
\begin{lemma} \label{Lem:Estimate1}
	It holds 
	\begin{eqnarray*}
		\abs{\mathcal{K}_{1,l} \mathcal{R}^\Gamma c_0\brackets{\varphi,t}}
		&\geq& \abs{\frac{\Lambda_\varphi'(t)}{C_{\varphi,t}}}\abs{\mathcal{K}_{2,l} \mathcal{R}_\alpha^\Gamma c_0\brackets{\varphi,t}}.
	\end{eqnarray*}
\end{lemma}
\begin{proof}
	First, we regard an estimate for the weighted Radon transform
	\begin{eqnarray*}
		\mathcal{R}_\alpha^\Gamma c_0\brackets{\varphi, t, s} 
		&=& \int_{\R^2} \alpha_{\varphi,t}\brackets{\mathbf{y}}c_0\brackets{\mathbf{y}} h_{\varphi,t}\brackets{\mathbf{y}} \delta\brackets{\Gamma_{\varphi,t}^{-1}\mathbf{y} \cdot \mathbf{e}_{\varphi} - s} \dmathInt \mathbf{y} \\ 
		&\leq& \frac{A}{G}C_{\varphi,t} \mathcal{R}^\Gamma c_0\brackets{\varphi, t, s},
	\end{eqnarray*}
	where we used $\displaystyle	\alpha_{\varphi,t}\brackets{\mathbf{y}} \leq \abs{\alpha_{\varphi,t}\brackets{\mathbf{y}}} \leq \norm{\brackets{\Gamma_{\varphi,t}^{-1}}'\mathbf{y}}\leq \frac{A}{G} C_{\varphi,t}$, which holds due to Assumption~\ref{Assumption}. Therewith, we obtain
	\begin{eqnarray*}
		\abs{\mathcal{K}_{1,l} \mathcal{R}^\Gamma c_0\brackets{\varphi,t}}
		&=& \abs{\mu_0 \; \mathbf{e}_{\varphi}\cdot \mathbf{p}_l \; A \Lambda_\varphi'(t) \; \overline{m}' \left(G \; \cdot\right) * \mathcal{R}^\Gamma c_0\brackets{\varphi, t, \cdot}  \left(s_{\varphi,t}\right)} \\
		&\geq& \abs{\mu_0 \; \mathbf{e}_{\varphi}\cdot \mathbf{p}_l \; G \,\frac{\Lambda_\varphi'(t)}{C_{\varphi,t}} \; \overline{m}' \left(G \; \cdot\right) * \mathcal{R}_\alpha^\Gamma c_0\brackets{\varphi, t, \cdot}  \left(s_{\varphi,t}\right)} \\
		&=& \abs{\frac{\Lambda_\varphi'(t)}{C_{\varphi,t}}}\abs{\mathcal{K}_{2,l} \mathcal{R}_\alpha^\Gamma c_0\brackets{\varphi,t}}.
	\end{eqnarray*}
\end{proof}
\begin{lemma} \label{Lem:Estimate2}
	Let $\overline{m}'_\infty:=\underset{{\lambda\in\R}}{\max}\;\overline{m}'\brackets{\lambda}$ and $c_{max}$ the maximal particle concentration during measurements, it holds that
	\begin{eqnarray*}
		\abs{\mathcal{K}_{2,l} \mathcal{R}^\Gamma_\alpha c_0\brackets{\varphi,t}}
		&\leq&	\mu_0 \norm{\mathbf{p}_l} \overline{m}'_\infty \; A \; C_{\varphi,t} \;c_{\max}\pi R^2, \\
		\abs{\mathcal{K}_{3,l} \mathcal{R}^\Gamma_\beta c_0\brackets{\varphi,t}} 
		&\leq&	\mu_0 \norm{\mathbf{p}_l} m \; D_{\varphi,t} \;c_{\max}\pi R^2.
	\end{eqnarray*}
\end{lemma}
\begin{proof}
	From the proof of Theorem~\ref{Thm:DynConv} we know that using Assumption~\ref{Assumption} we can write 
	\begin{eqnarray*}
		\abs{\mathcal{K}_{2,l} \mathcal{R}_\alpha^\Gamma c_0\brackets{\varphi,t}}
		&=&	\left|  \mu_0 \; \mathbf{e_\varphi} \cdot \mathbf{p}_l \; G  \int_{\R^2} c_0\brackets{\mathbf{y}} h_{\varphi,t}\brackets{\mathbf{y}} \right. \\ 
		&\phantom{\hspace{1mm}=}& \hspace{2.5cm} \times \;
		\left. 	\overline{m}^\prime \left( -G \; \Gamma_{\varphi,t}^{-1}\mathbf{y} \cdot \mathbf{e}_{\varphi} + A\Lambda_\varphi(t)\right)  \brackets{\Gamma_{\varphi,t}^{-1}}'\mathbf{y} \cdot \mathbf{e}_{\varphi}  \dmathInt \mathbf{y}\right|  \\
		%			&\leq&	\mu_0 \norm{\mathbf{p}_l}  \; A \, C_{\varphi,t}\abs{ \int_{\R^2} c_0\brackets{\Gamma_{\varphi,t}\mathbf{r}}
		%			h_{\varphi,t}\brackets{\Gamma_{\varphi,t}\mathbf{r}} \abs{\det \Dmath \Gamma_{\varphi,t}\mathbf{r}} \; \overline{m} \left( -G \; \mathbf{r}\cdot \mathbf{e}_\varphi	+ A\Lambda_\varphi(t)\right) \dmathInt\mathbf{r}} \\
		&\leq&	\mu_0 \norm{\mathbf{p}_l}  \; A \; C_{\varphi,t}\abs{ \int_{\R^2} c\brackets{\mathbf{r}, \varphi, t} \; \overline{m}' \left( -G \; \mathbf{r}\cdot \mathbf{e}_\varphi	+ A\Lambda_\varphi(t)\right) \dmathInt\mathbf{r}} \\
		&\leq& \mu_0 \norm{\mathbf{p}_l} \overline{m}'_\infty \; A \; C_{\varphi,t}\abs{ \int_{\R^2} c\brackets{\mathbf{r}, \varphi, t}  \dmathInt \mathbf{r}} \\
		&\leq& \mu_0 \norm{\mathbf{p}_l} \overline{m}'_\infty \; A \; C_{\varphi,t} \;c_{\max}\pi R^2.
	\end{eqnarray*} 
	The modulus of the mean magnetic moment is bounded by $m$ (cf. Langevin model in Section~\ref{Sec:FFL}). Thus, similarly we obtain
	\begin{eqnarray*}
		\abs{\mathcal{K}_{3,l} \mathcal{R}_\alpha^\Gamma c_0\brackets{\varphi,t}}
		&=&	\abs{ \mu_0 \; \mathbf{e_\varphi} \cdot \mathbf{p}_l   \int_{\R^2} c_0\brackets{\mathbf{y}} \overline{m} \left( -G \; \Gamma_{\varphi,t}^{-1}\mathbf{y} \cdot \mathbf{e}_{\varphi} + A\Lambda_\varphi(t)\right) h_{\varphi,t}'\brackets{\mathbf{y}} \dmathInt \mathbf{y}} \\
		&\leq& \mu_0 \norm{\mathbf{p}_l} m \; D_{\varphi,t}\abs{ \int_{\R^2} c_0\brackets{\mathbf{y}}  \dmathInt \mathbf{y}}  \\
		&\leq& \mu_0 \norm{\mathbf{p}_l} m \; D_{\varphi,t} \;c_{\max}\pi R^2.
	\end{eqnarray*}
\end{proof}\\
From Lemma~\ref{Lem:Estimate1} we get that, if the motion speed is suitably slow compared to the FFL translation, the part $\mathcal{K}_{2,l} \circ \mathcal{R}_\alpha^\Gamma$ of the forward model might be neglected. Furthermore, if we have estimations for the speed and extent of the object's dynamics, the bounds in Lemma~\ref{Lem:Estimate2} can be computed and thus give insight, which components of the forward model should be incorporated and which might be negligible within the image reconstruction process. 

%%%%%%%%%%%%%%%%%%%%%%%%%%%%%%%%%%
\section{Radon-based image reconstruction using TV regularization} \label{Sec:ImageReco}
%%%%%%%%%%%%%%%%%%%%%%%%%%%%%%%%%%
We aim to reconstruct the reference concentration $c_0$. Knowledge of the motion functions $\Gamma$, then allows computation of the time-dependent MNP distribution at all other states. According to Remark~\ref{Rem:GammaReg} usage of motion models and thus, incorporating prior motion information, already might regularize the solution as side effect. However, additionally we apply TV regularization. Therefore, we define the variation of functions as
\begin{equation*}
\TV\brackets{c} := \sup\left\lbrace \int_{\R^2} c\brackets{\mathbf{r}}  \operatorname{div}\brackets{\mathbf{g}}\brackets{\mathbf{r}} \dmathInt \mathbf{r} : \; \mathbf{g} \in C^\infty_0\brackets{B_R,\R^2}, \; \abs{\mathbf{g}\brackets{\mathbf{r}}}_2 <1 \; \text{for all} \; \mathbf{r} \right\rbrace.	
\end{equation*}
Therewith, we define the space $\BV$ of functions with bounded variation on $B_R$
\begin{eqnarray*}
	\BV\brackets{B_R} := \left\lbrace c \in L_1\brackets{B_R,\R} : \; \TV\brackets{c} < \infty\right\rbrace,  
\end{eqnarray*}
which provided with $\norm{\cdot}_{\BV} := \norm{\cdot}_{L_1} + \TV\brackets{\cdot}$ becomes a Banach space and satisfies $\BV\brackets{B_R} \subset L_2\brackets{B_R}$. Further information regarding total variation can be found e.g. in~\cite{acar1994analysis} and \cite{burger2013guide}.\\
Similar to~\cite{https://doi.org/10.48550/arxiv.2211.11683} we set $\mathcal{D} := L_2\brackets{B_R,\R}\times L_2\brackets{Z_T \times \R,\R }$ as well as $A^\Gamma_l: \mathcal{D}   \to L_2\brackets{Z_T,\R}$ such that
\begin{equation*}
A^\Gamma_l\brackets{c_0,v^\Gamma} := \mathcal{K}_{1,l} v^\Gamma + \brackets{\mathcal{K}_{2,l} \circ \mathcal{R}^\Gamma_\alpha}c_0 + \brackets{\mathcal{K}_{2,l} \circ \mathcal{R}^\Gamma_\beta}c_0
\label{Eqn:NewOp}
\end{equation*}
with operators $\mathcal{K}_{i,l},\; i=1,2,3$ and $\mathcal{R}^\Gamma_\alpha,\;\mathcal{R}^\Gamma_\beta$ from Theorem~\ref{Thm:DynConv}. Therewith, we have for $c_0\geq0\,$ that $\,A^\Gamma_l\brackets{c_0,\mathcal{R}^\Gamma c_0}=\mathcal{A}^{\Gamma}_l c_0$. Note that if the assumptions of Corollary~\ref{corollary} are satisfied, it holds that
\begin{equation*}
\left[ A^\Gamma_l\brackets{c_0,v^\Gamma}\right] \brackets{\varphi,t} 
= \left[ \mathcal{K}_{1,l} v^\Gamma+ \mathcal{K}_{2,l} 	
\brackets{\tilde{\alpha}_{\varphi,t}v^\Gamma} + \mathcal{K}_{3,l} \brackets{\tilde{\beta}_{\varphi,t}v^\Gamma}\right] \brackets{\varphi,t}
\end{equation*}
only depends on $v^\Gamma$. \\
Analogue to~\cite{https://doi.org/10.48550/arxiv.2211.11683}, inspired by~\cite{Tovey_2019}, we jointly reconstruct the reference concentration $c_0$ of the MNP distribution and corresponding dynamic Radon data $v^\Gamma :=\mathcal{R}^\Gamma c_0$ via solving

\begin{equation}
\min_{\brackets{c_0,v^\Gamma} \in \mathcal{C}} \frac{1}{2} \sum_l \norm{A^\Gamma_l\brackets{c_0,v^\Gamma} - u_l}^2_{L_2} + \frac{\alpha_1}{2}\norm{\mathcal{R}^\Gamma c_0 - v^\Gamma}^2_{L_2} + \alpha_2 \TV\brackets{c_0} + \alpha_3 P\brackets{v^\Gamma}.
\label{Eqn:OptProblem}
\end{equation}
Compared to~\cite{https://doi.org/10.48550/arxiv.2211.11683} we included an additional penalty term acting on the adapted Radon data $v^\Gamma$. Further components are the feasible set 
\begin{equation*}
\mathcal{C} := \left\lbrace \brackets{c_0,v^\Gamma} \in \mathcal{D} : \; c_0\geq 0,\; v^\Gamma\geq0\right\rbrace, 
\end{equation*}
the given voltage signals $u_l,\; l=1,\dots,L$, the weighting parameter $\alpha_1>0$, as well as regularization parameters $\alpha_2,\;\alpha_3>0$. At this point, we do not specify the choice of $P$. In~\cite{Tovey_2019} directional TV regularization was used. In our numerical examples we will either regard $P=0$, i.e. we only penalize the concentration and therewith indirectly also the adapted Radon data, or we apply an additional sparsity constraint on $v^\Gamma$. Furthermore, the TV regularization can be exchanged by any penalty term suiting the regarded imaging problem.
\begin{theorem}
	For convex, proper, and weakly lower semicontinuous $P$ the minimization problem~\eqref{Eqn:OptProblem} has a solution $\brackets{\hat{c}_0,\hat{v}^\Gamma} \in \mathcal{C}$ with $\hat{c}_0 \in \BV\brackets{B_R}.$
\end{theorem}
The proof basically follows standard arguments from convex optimization and has been stated in~\cite{https://doi.org/10.48550/arxiv.2211.11683} for $P=0$.

%Additionally, the Poincaré-Wirtinger inequality holds
%\begin{equation}
%\norm{c - \overline{c}}_{L_2} \leq C\, \TV\brackets{c}, \quad \overline{c} = \frac{1}{\pi R^2} \int_{B_R} c\brackets{\mathbf{r}} \dmathInt \mathbf{r}
%\label{Eqn_Poincare}
%\end{equation}
%for some constant $C>0$.

%%%%%%%%%%%%%%%%%%%%%%%%%%%%%%%%%%
\section{Numerical results} \label{Sec:Results}
%%%%%%%%%%%%%%%%%%%%%%%%%%%%%%%%%%
For our numerical examples we use a simulation framework based on the one developed by Gael Bringout~\cite{bringout2016field} available at~\url{https://github.com/gBringout}, which we adapted to suit our purposes. \\\\
We cross over to a partly discrete setting. More precisely, we regard measurements for $p$ different FFL directions determined by angles
\begin{equation*}
\varphi_j := \brackets{j-1}\frac{\pi}{p},\quad j=1,\dots,p.
\end{equation*}
For our experiments we set the excitation function to be
\begin{equation*}
\Lambda_{\varphi_j}\brackets{t} =
\begin{cases}
\phantom{-} \cos \brackets{2\pi f_{\text{d}} t}, \quad j \;\text{odd},  \\
- \cos \brackets{2\pi f_{\text{d}} t} , \quad j \;\text{even},
\end{cases} j = 1,\dots,p, \quad t\in\left[ 0, T \right] 
\end{equation*}
with drive frequency $f_{\text{d}}$ and $T=\frac{1}{2f_d}$. This means that the FFL starts with a signed distance $\frac{A}{G}$ to the origin, is translated once through the FOV until it reaches a displacement of $-\frac{A}{G},$ is then rotated, translated again back through the FOV, and so on. Further, we assume an instantaneous rotation. This is not a realistic proposition but acceptable since we only opt for a proof-of-concept. Additionally, we suppose the particle concentration to be normalized to one and located within a circle of radius $\frac{A}{G}$ around zero during the whole data acquisition. For data simulation the FOV $\left[-\frac{A}{G},\frac{A}{G} \right] \times \left[-\frac{A}{G},\frac{A}{G} \right]$ is partitioned into $501 \times 501$ pixel, whereas we aim to reconstruct on a coarser grid with $201\times 201$ sampling points. The tracer is modeled as a solution with $0.5 \; \frac{\text{mol}\brackets{\text{Fe}_3\text{O}_4}}{\text{m}^3}$ concentration of magnetite with $30$ nm core diameter and $\frac{0.6}{\mu_0}$ T saturation magnetization. We consider two receive coils oriented orthogonal to each other. The choice of parameter values can be found in Table~\ref{table:SimulationParameters}.
\begin{table}[htbp]
	\centering
	\caption{\small Simulation parameters }
	\renewcommand{\arraystretch}{1.3}
	\begin{tabular}{|c|c|c|c|}
		\hline
		\textbf{Parameter} & \textbf{Explanation} & \textbf{Value} & \textbf{Unit} \\
		\hline \hline
		& & & \\[-1.2em]
		$\mu_0$ & magnetic permeability & $4\pi\cdot 10^{-7}$ & $\text{T}\text{m}\text{A}^{-1}$ \\
		$k_{\text{B}}$ & Boltzmann constant & $1.380650424\cdot 10^{-23}$ & $\text{J}\text{K}^{-1}$ \\
		\hline
		$G$ & gradient strength & $4$ & $\text{T}\brackets{\text{m}\mu_0}^{-1}$ \\
		$A$ & drive peak amplitude & $15$  & $\text{mT}\mu_0^{-1}$\\
		$\mathbf{p}_1$ & sensitivity of the first receive coil & $\left[0.015/293.29, 0 \right]^T $  & $\text{m}^{-1}$ \\
		$\mathbf{p}_2$ & sensitivity of the second receive coil & $\left[0, 0.015/379.71\right]^T $  & $\text{m}^{-1}$\\
		\hline
		$f_{\text{d}}$ & drive-field frequency & 25 & kHz\\
		$f_{\text{s}}$ & sampling frequency & 8 & MHz\\
		\hline
		$p$ & amount of FFL directions & 25 & \\
		\hline
	\end{tabular}
	\label{table:SimulationParameters}
	\renewcommand{\arraystretch}{1}
\end{table}\\\\
We consider an affine motion $\Gamma_{\varphi_j,t}\mathbf{r}:=\mathbf{A}_{\varphi_j,t}\mathbf{r} + \mathbf{b}_{\varphi_j,t}$ choosing for $j=1,\dots,p$ and $t\in\left[ 0, T\right]$
\begin{eqnarray*}
	\mathbf{A}_{\varphi_j,t}&:=&
	\begin{pmatrix}
		a_{\varphi_j,t}& 0 \\
		0 & a_{\varphi_j,t}
	\end{pmatrix}, \quad a_{\varphi_j,t}:= 0.8 + 0.2 \cos\brackets{2 \pi f \brackets{\frac{j-1}{2f_d}+t}},\\
	\mathbf{b}_{\varphi_j,t}&:=&
	a_{\varphi_j,t}
	\begin{pmatrix}
		\Delta x -  \brackets{\frac{j-1}{2f_d}+t} \frac{2\Delta x}{\brackets{p-1}T}\\
		0
	\end{pmatrix},
\end{eqnarray*}
with $f=78$ kHz denoting the deformation frequency and $\Delta x = 2$ mm being the translation shift of the object at initial state compared to the reference concentration, which is shown in Figure~\ref{Fig:ReferenceState}. 
%Note that, especially under the assumption of intensity preservation, this is not a realistic motion example but suffices our aim to investigate the dynamic model and joint image reconstruction of Radon data and phantom. 
According to~\eqref{Ex:beta} and the proof of Corollary~\ref{corollary}, we obtain
\begin{equation*}
\brackets{\mathcal{K}_{3,l} \circ \mathcal{R}^\Gamma_\beta c_0} \brackets{\varphi_j,t} = \tilde{\beta}_{\varphi_j,t} \brackets{\mathcal{K}_{3,l} \circ \mathcal{R}^\Gamma c_0} \brackets{\varphi_j,t}=:\brackets{\widetilde{\mathcal{K}}_{3,l} \circ \mathcal{R}^\Gamma c_0 } \brackets{\varphi_j,t}
\end{equation*}
with $\tilde{\beta}_{\varphi_j,t}=0$ respectively $\displaystyle \tilde{\beta}_{\varphi_j,t} = \tr{\mathbf{A}^{-1}_{\varphi_j,t} \, \mathbf{A}'_{\varphi_j,t}} = 2\frac{a'_{\varphi_j,t}}{a_{\varphi_j,t}}$ proposing mass respectively intensity preservation. 
\begin{figure}[htbp]
	\begin{subfigure}[b]{0.32\textwidth}
		\centering
		\includegraphics[width=1.22\linewidth, trim=4.3cm 8.5cm 0 8cm, clip]{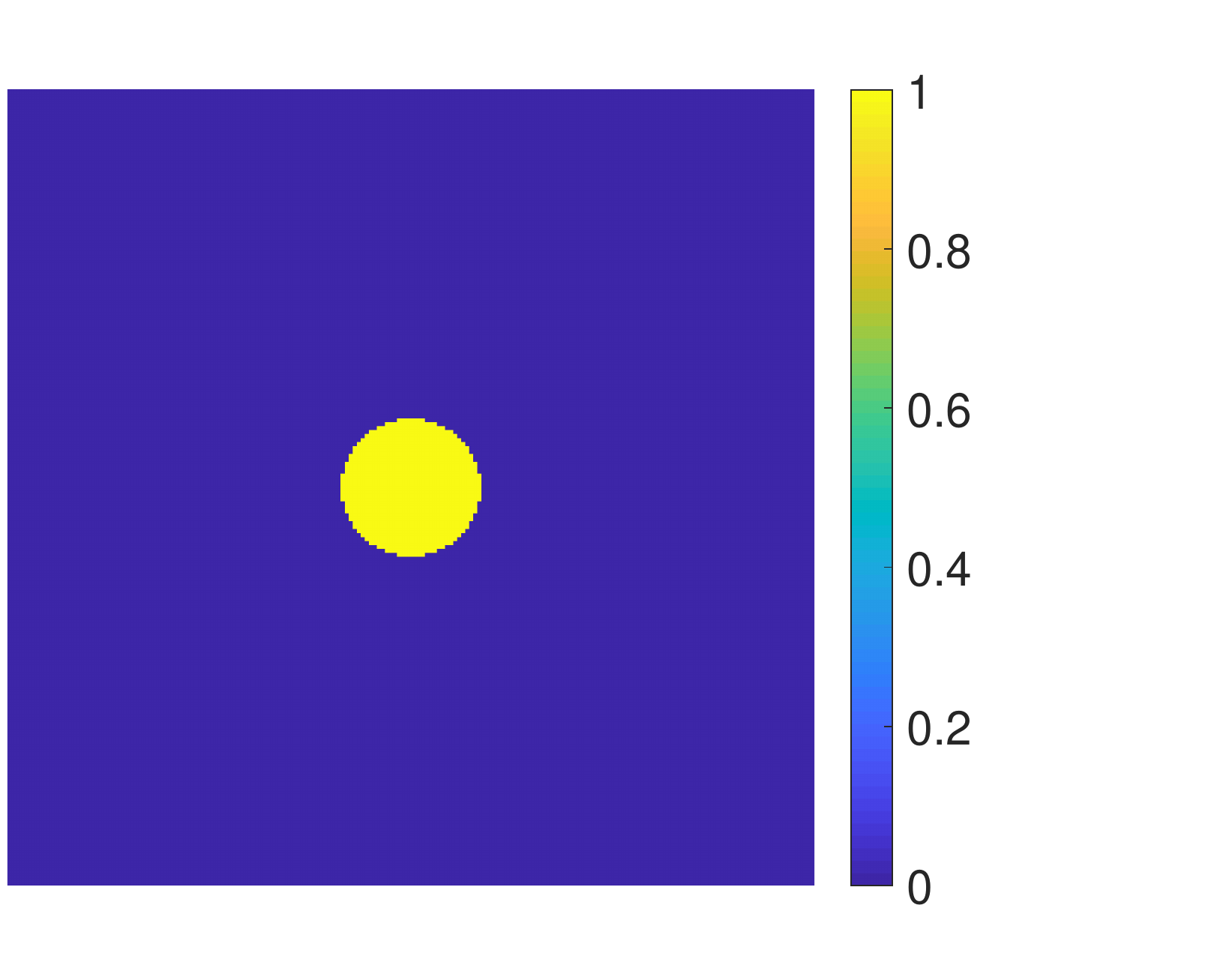}
		\subcaption{\small }
		\label{Fig:ReferenceState}
	\end{subfigure}
	\hfill
	\begin{subfigure}[b]{0.32\textwidth}
		\centering
		\includegraphics[width=1.22\linewidth, trim=4.3cm 8.5cm 0 8cm, clip]{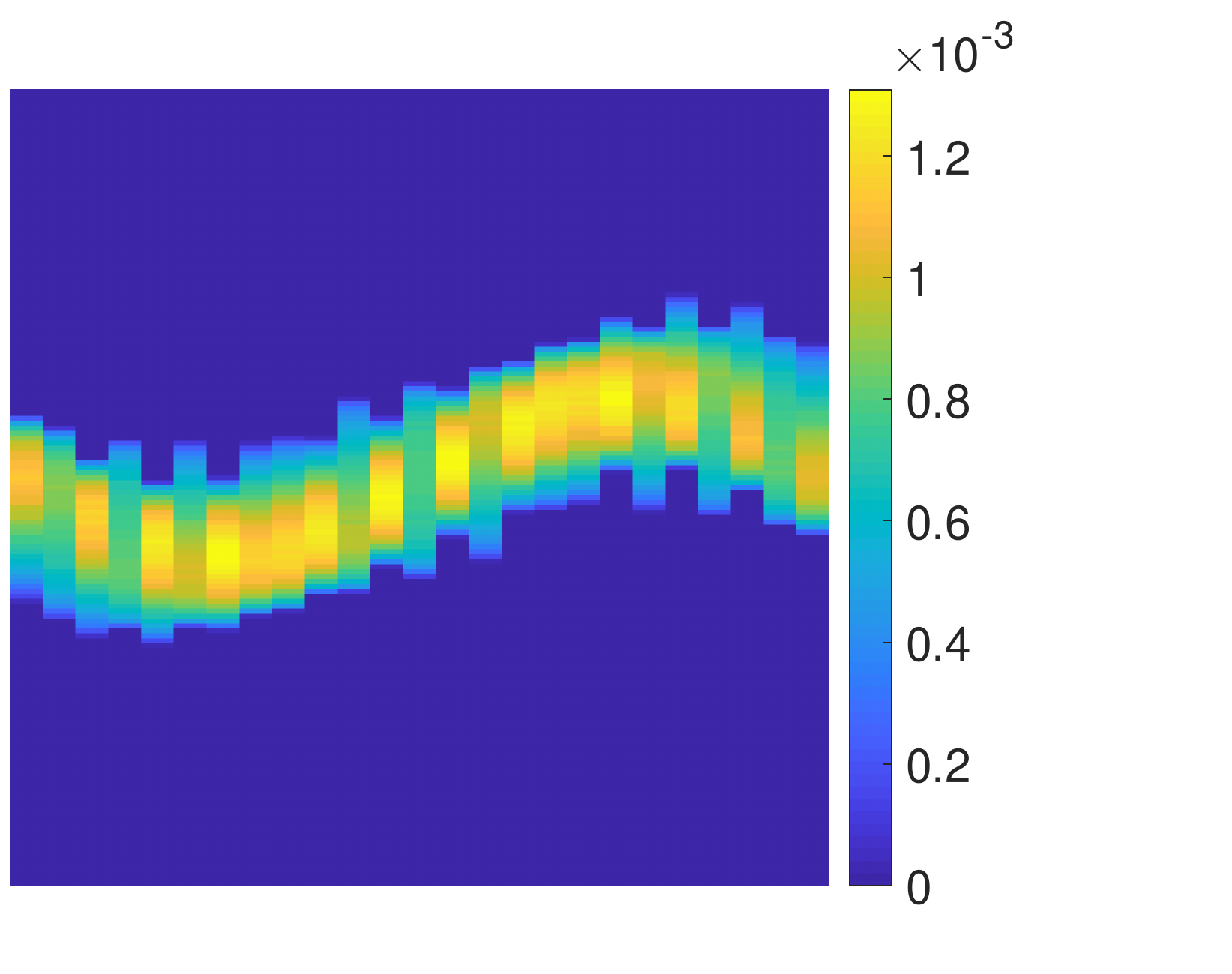}
		\subcaption{\small }
	\end{subfigure}
	\hfill
	\begin{subfigure}[b]{0.32\textwidth}
		\centering
		\includegraphics[width=1.22\linewidth, trim=4.3cm 8.5cm 0 8cm, clip]{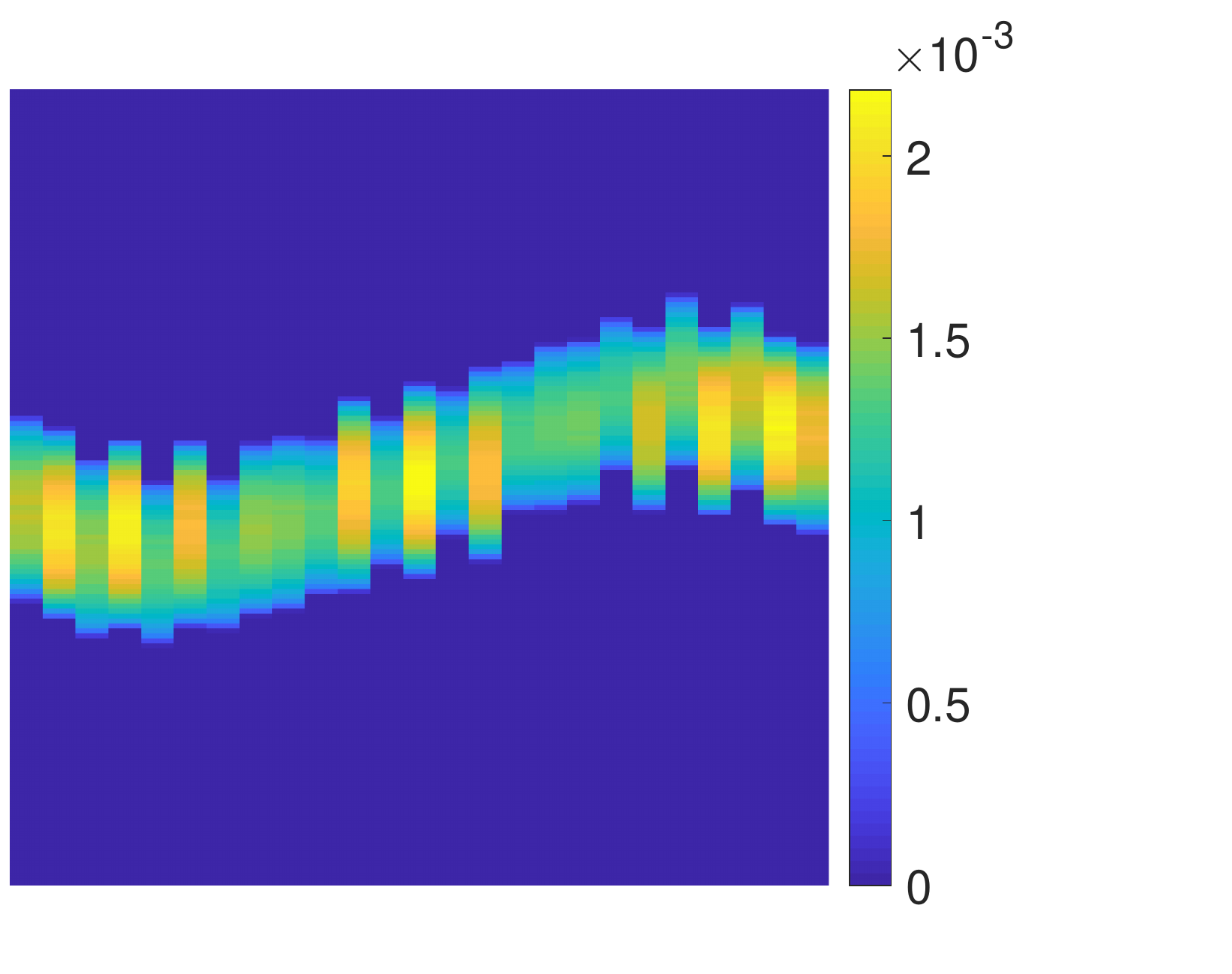}
		\subcaption{\small }
	\end{subfigure}
	\caption{\small Reference particle concentration (a) and sinograms assuming mass (b) or intensity (c) preservation.}
	\label{Fig:Sinograms}
\end{figure}\\\\
Discretizations are obtained via standard methods and indicated using bold letters. In order to evaluate $\mathbf{K}_{i,l}\mathbf{v}^\Gamma$, we would need a sinogram $\mathbf{v}^\Gamma$ with columns containing Radon data for the different FFL displacements and each column being dedicated to a fixed time-angle combination. Instead, to reduce the problem size, we opt for reconstructing the sinogram with entries $\mathcal{R}^\Gamma c_0\brackets{\varphi_m,t_n,s_n}$ for $m=1,\dots,p$ and $n=0,\dots,N:=\frac{f_s}{2f_d}$ as well as
\begin{eqnarray*}
	t_n = \frac{n}{f_s}, \quad
	s_n = \brackets{1-\frac{2n}{N}}\frac{A}{G}.
\end{eqnarray*} 
Therewith, we accept an additional model error, which we assume to be small enough to be manageable applying regularization methods. Note that $\overline{m}'$ converges to the dirac delta for particle diameters tending to infinity~\cite{knopp2012magnetic}. The resulting sinograms for mass and intensity preservation can be found in Figure~\ref{Fig:Sinograms}. \\
Next, we will investigate reconstruction results. We divide data and operators by the maximum absolute data value $u^*$
\begin{equation*}
\widehat{\mathbf{u}}_l := \frac{\mathbf{u}_l}{u^*}, \quad \widehat{\mathbf{K}}_{i,l} := \frac{\mathbf{K}_{i,l}}{u^*}, 
\quad \widehat{\widetilde{\mathbf{K}}}_{3,l} := \frac{\widetilde{\mathbf{K}}_{3,l}}{u^*}, \quad i,l=1,2
\end{equation*} 
and regard the following minimization problems
\begin{eqnarray*}
	\mathcal{M}_1: \hspace{1.1cm}
	\min_{\mathbf{c}_0\geq 0, \mathbf{v}^\Gamma\geq 0} \frac{1}{2} \sum_{l=1}^2 \norm{\widehat{\mathbf{K}}_{1,l}\mathbf{v}^\Gamma - \widehat{\mathbf{u}}_l}^2_{2} + \frac{\alpha_1}{2}\norm{\mathbf{R}\mathbf{c}_0 - \mathbf{v}^\Gamma}^2_{2} + \alpha_2 \norm{\abs{\nabla \mathbf{c}_0}_2}_1 + \alpha_3 \norm{\mathbf{v}^\Gamma}_1, 
\end{eqnarray*}
\begin{eqnarray*}
	\mathcal{M}_2: \hspace{1.1cm}
	\min_{\mathbf{c}_0\geq 0, \mathbf{v}^\Gamma\geq 0} \frac{1}{2}\sum_{l=1}^2 \norm{\widehat{\mathbf{K}}_{1,l}\mathbf{v}^\Gamma - \widehat{\mathbf{u}}_l}^2_{2} + \frac{\alpha_1}{2}\norm{\mathbf{R}^\Gamma\mathbf{c}_0 - \mathbf{v}^\Gamma}^2_{2} + \alpha_2 \norm{\abs{\nabla \mathbf{c_0}}_2}_1 + \alpha_3 \norm{\mathbf{v}^\Gamma}_1, 
\end{eqnarray*}
\begin{eqnarray*}
	\mathcal{M}_3: \hspace{1.1cm}
	\min_{\mathbf{c}_0\geq 0, \mathbf{v}^\Gamma\geq 0} \frac{1}{2}\sum_{l=1}^2 \norm{\brackets{\widehat{\mathbf{K}}_{1,l}+\widehat{\widetilde{\mathbf{K}}}_{3,l}}\mathbf{v}^\Gamma - \widehat{\mathbf{u}}_l}^2_{2} + \frac{\alpha_1}{2}\norm{\mathbf{R}^\Gamma\mathbf{c}_0 - \mathbf{v}^\Gamma}^2_{2} &+& \alpha_2 \norm{\abs{\nabla \mathbf{c}_0}_2}_1 \\ &+& \alpha_3 \norm{\mathbf{v}^\Gamma}_1.
\end{eqnarray*}
Hence, we neglect $\widehat{\mathbf{K}}_{2,l}\mathbf{R}_\alpha^\Gamma \mathbf{c}_0$ in the reconstruction process as its contribution to the signal was insignificant small compared to the other two signal components. To solve the regarded problem, we use CVX, a package for specifying and solving convex programs~(\cite{cvx},~\cite{gb08}), together with the MOSEK solver~\cite{mosek2010mosek}. If not mentioned differently, we compute results for parameters $\alpha_1 \in \left\lbrace 2 \cdot 10^i,\; i=1,\dots,8 \right\rbrace,\;\alpha_2 \in \left\lbrace 0.1^{5-0.2i},\; i=0,\dots,19 \right\rbrace,$ and $\alpha_3=0$ and give results for those choices with maximal structural similarity (SSIM) respectively maximal peak-signal-to-noise ratio (PSNR) of the reconstructed particle concentration compared with the reference phantom. Thereby, we do not aim for a parameter optimization in general but rather test for different magnitudes of $\alpha_1$. In the caption of reconstructions we only give the value of the similarity measure with respect to which we chose the parameters for that image. More values and corresponding parameters can be found at the end of each subsection. To state a guideline, applying the static reconstruction method $\mathcal{M}_1$ to static data we obtained at best a SSIM$\brackets{\mathbf{c}_0}$ value of $0.9875$ for $\alpha_1 = 2\cdot10^5,\;\alpha_2=0.1^{2.8},$ and a PSNR$\brackets{\mathbf{c}_0}$ value of $31.98$ for $\alpha_1 = 2\cdot10^6,\;\alpha_2=0.1^{3.5}$.
\subsection{Mass preservation}
In a first step, we simulate data assuming mass preservation. The just defined motion functions lead to a periodically expanding phantom being translated from left to right through the FOV (cf. Figure~\ref{Fig:Motion_mass}). 
\begin{figure}[htbp]
	\begin{subfigure}[b]{0.32\textwidth}
		\centering
		\includegraphics[width=1.15\linewidth, trim=3.6cm 8.5cm 0 8cm, clip]{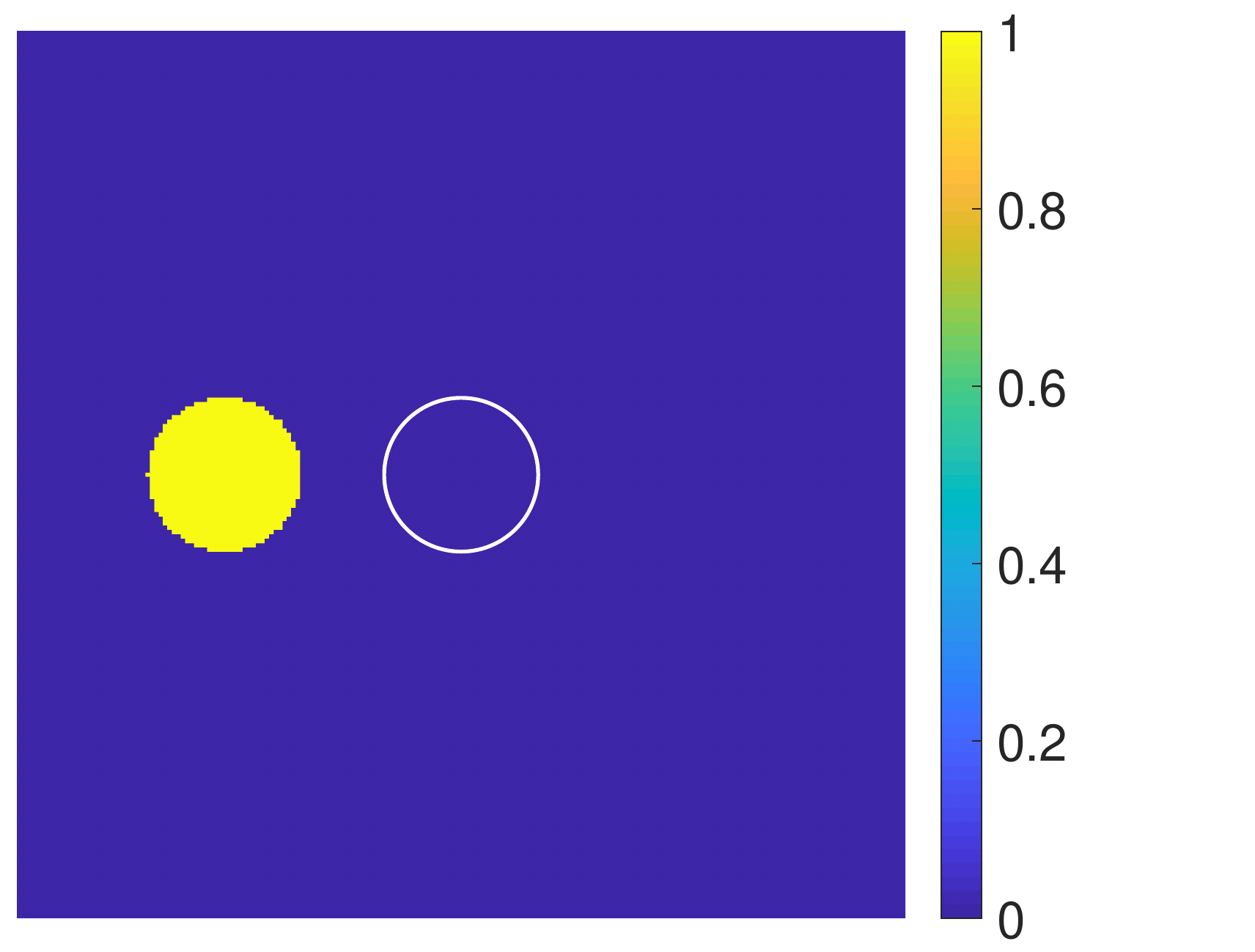}
		\subcaption{\small }
	\end{subfigure}
	\hfill
	\begin{subfigure}[b]{0.32\textwidth}
		\centering
		\includegraphics[width=1.15\linewidth, trim=3.6cm 8.5cm 0cm 8cm, clip]{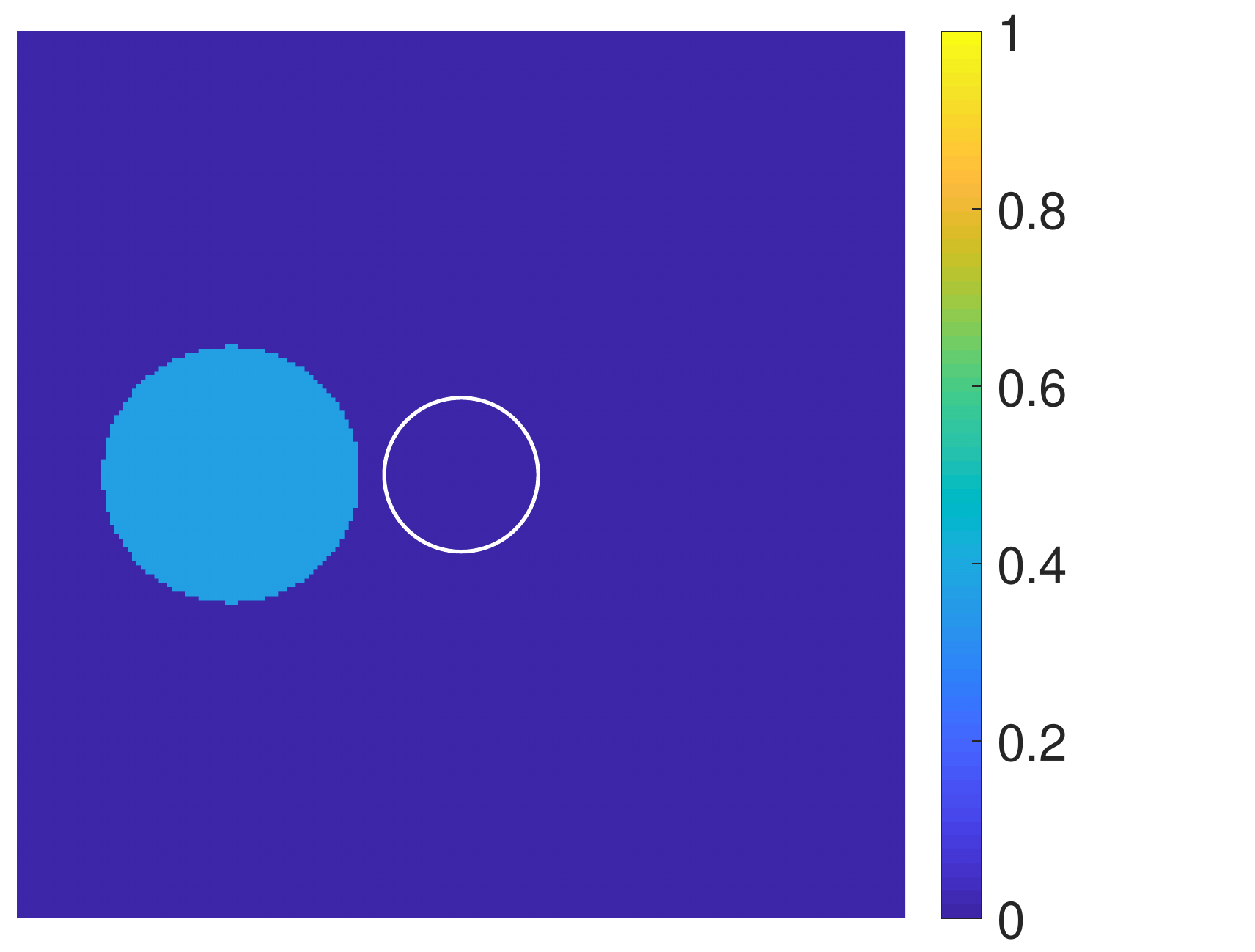}
		\subcaption{\small }
	\end{subfigure}
	\hfill
	\begin{subfigure}[b]{0.32\textwidth}
		\centering
		\includegraphics[width=1.15\linewidth, trim=3.6cm 8.5cm 0cm 8cm, clip]{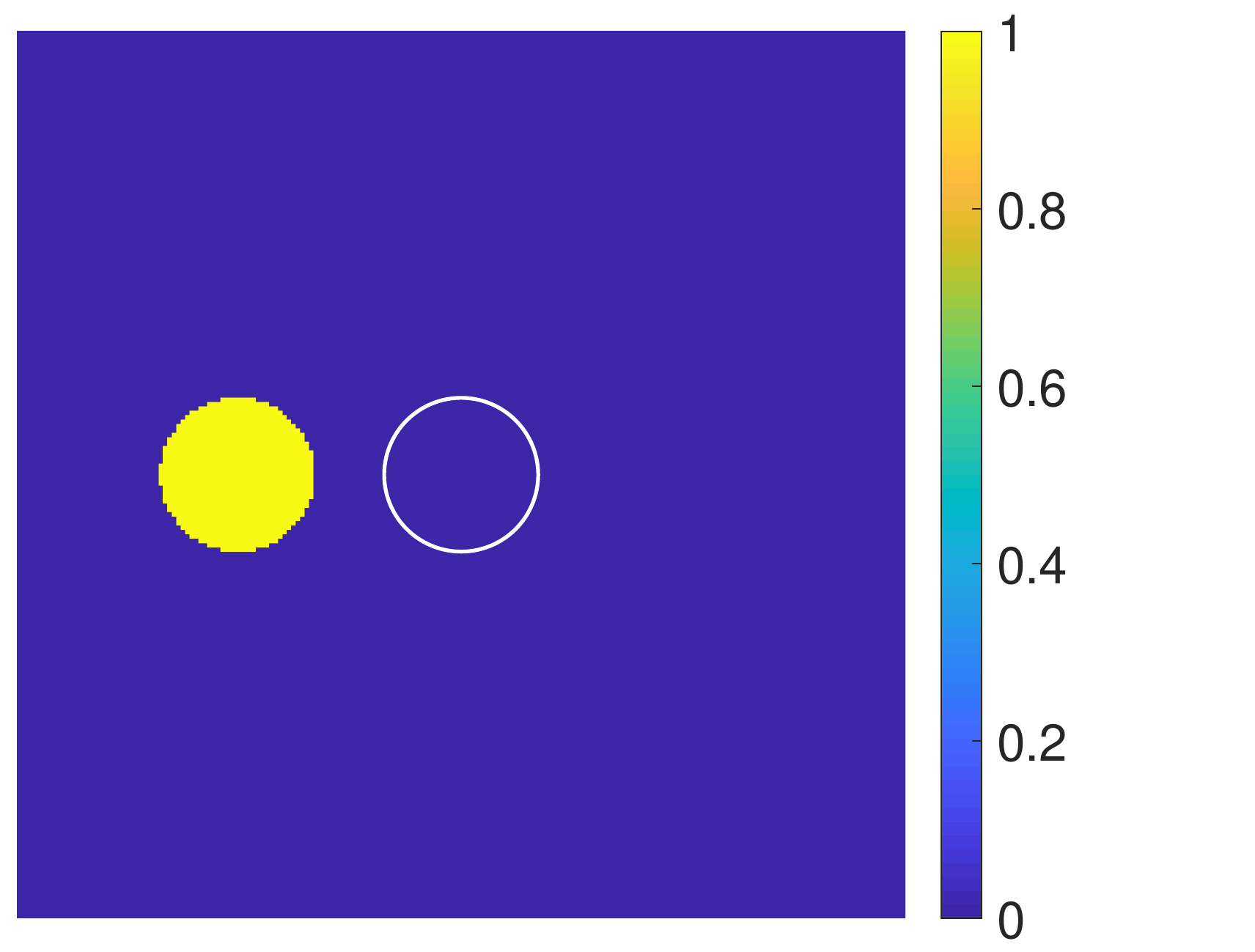}
		\subcaption{\small }
	\end{subfigure}
	\caption{\small Visualization of the object motion: The phantom starts at the left side of the FOV (a). While being shifted to the right, it expands until it reaches its maximum size (b). Still being translated to the right, it decreases again until the next expansion cycle starts (c). The white circle indicates the contour of the reference concentration.}
	\label{Fig:Motion_mass}
\end{figure}\\\\
Neglecting the phantom dynamics in the reconstruction and applying method $\mathcal{M}_1$ results in severe motion artifacts. Corresponding images are depicted in Figure~\ref{Fig:RecoPhantomMassStatic_ssim} respectively Figure~\ref{Fig:RecoPhantomMassStatic_psnr} with SSIM respectively PSNR as quality measure for comparing different parameter choices. To decrease the effect of outliers and enhance the contrast, Figure~\ref{Fig:RecoPhantomMassStatic_ssim_fullColorbar} shows a version of Figure~\ref{Fig:RecoPhantomMassStatic_ssim} where we restricted the colorbar to the range of our groundtruth.
\begin{figure}[htbp]
	\begin{subfigure}[b]{0.32\textwidth}
		\centering
		\includegraphics[width=1.15\linewidth, trim=3.6cm 8.5cm 0 8cm, clip]{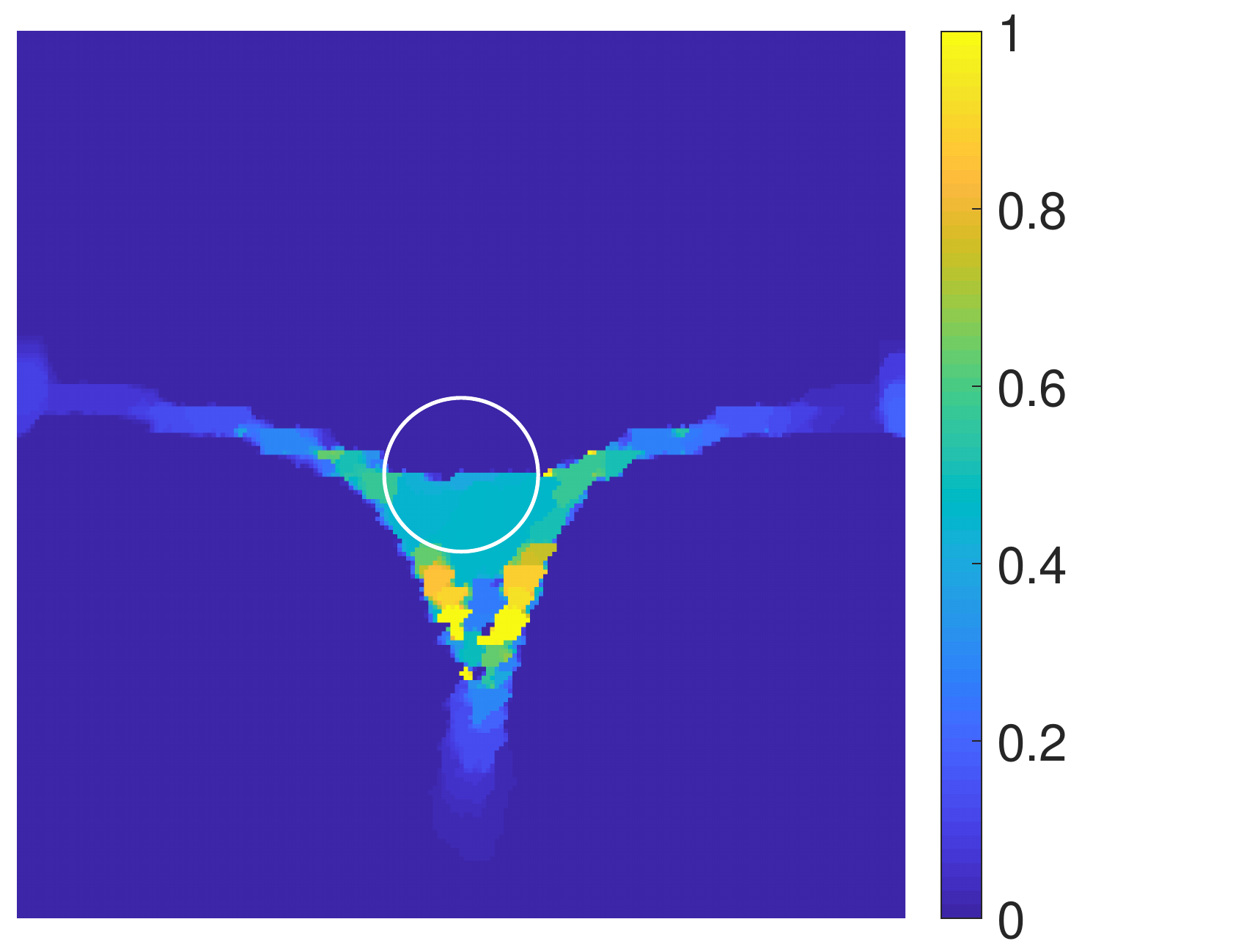}
		\subcaption{\small SSIM$\brackets{\mathbf{c}_0}$ = 0.8628}
		\label{Fig:RecoPhantomMassStatic_ssim_fullColorbar}
	\end{subfigure}
	\hfill
	\begin{subfigure}[b]{0.32\textwidth}
		\centering
		\includegraphics[width=1.15\linewidth, trim=3.6cm 8.5cm 0cm 8cm, clip]{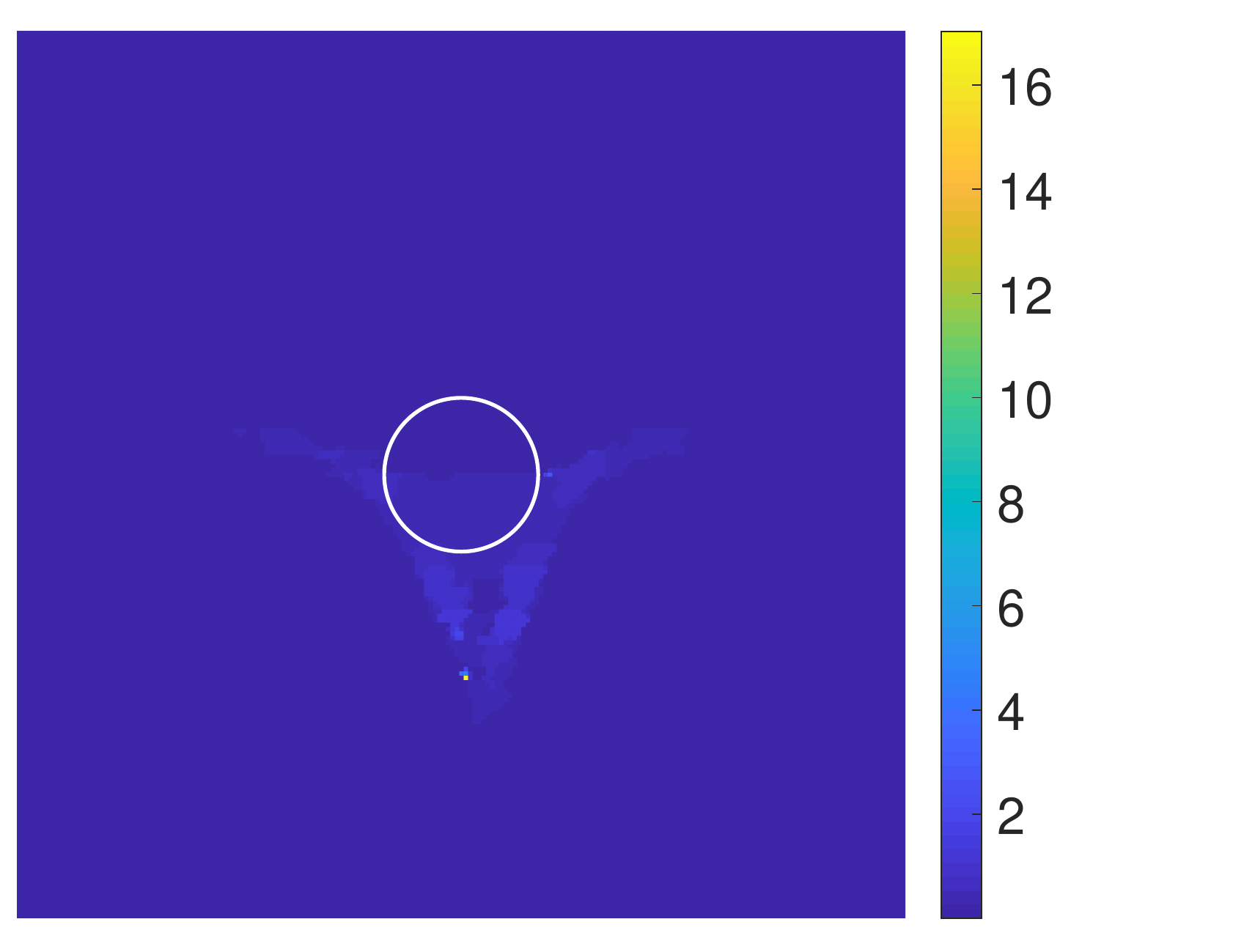}
		\subcaption{\small SSIM$\brackets{\mathbf{c}_0}$ = 0.8628}
		\label{Fig:RecoPhantomMassStatic_ssim}
	\end{subfigure}
	\hfill
	\begin{subfigure}[b]{0.32\textwidth}
		\centering
		\includegraphics[width=1.15\linewidth, trim=3.6cm 8.5cm 0cm 8cm, clip]{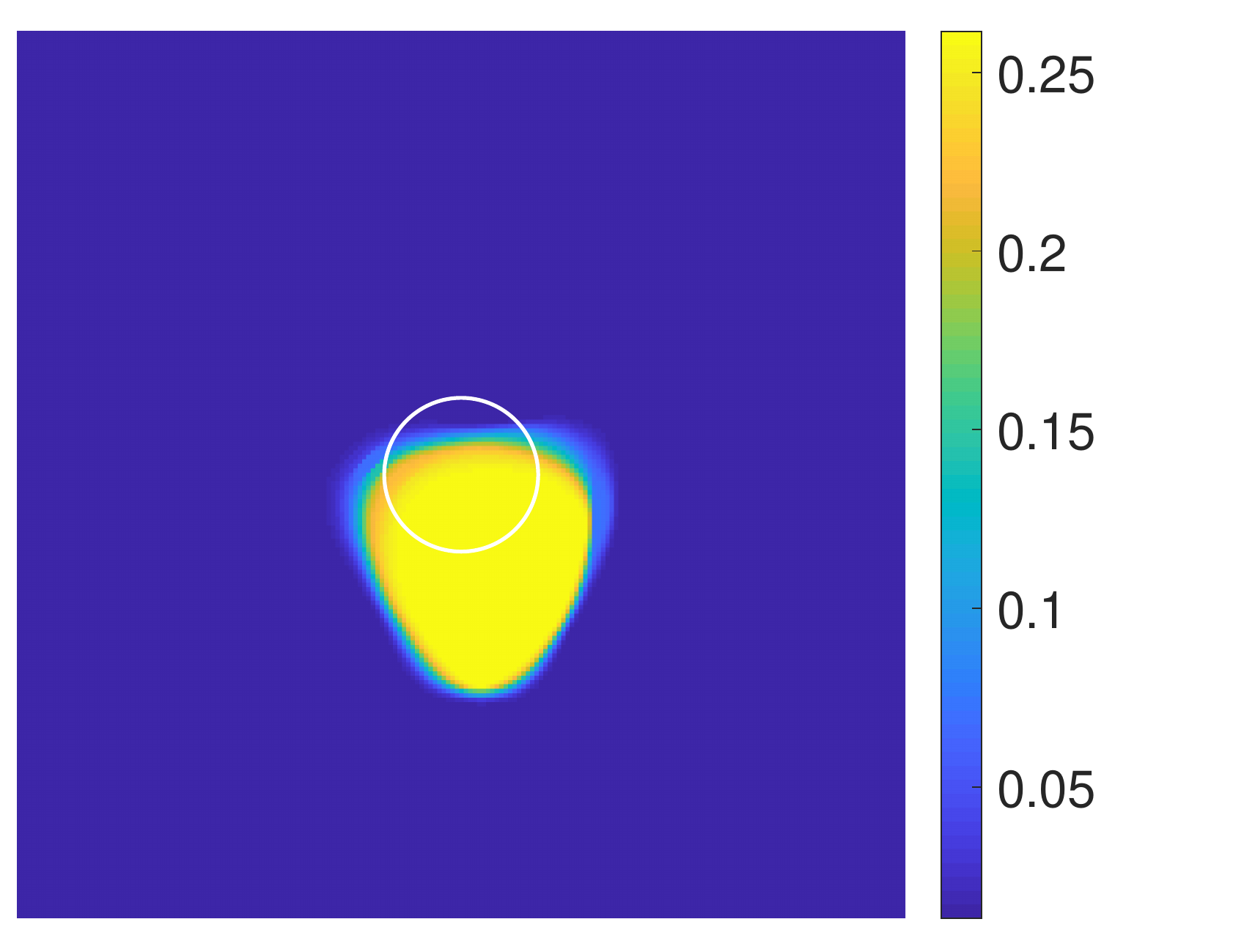}
		\subcaption{\small PSNR$\brackets{\mathbf{c}_0}$ = 17.43}
		\label{Fig:RecoPhantomMassStatic_psnr}
	\end{subfigure}
	\caption{\small \small Phantom reconstructions applying the static reconstruction method $\mathcal{M}_1$ to dynamic data. The white circle indicates shape and location of the chosen reference concentration. For contrast enhancement (a) shows a version of (b) with values larger than one being projected to one.}
	\label{Fig:RecoPhantomMassStatic}
\end{figure}\\
Thus, information about the object's motion clearly needs to be incorporated in the reconstruction in order to avoid motion artifacts and to get useful images. According to Remark~\ref{Rem:beta} it holds $\widehat{\widetilde{\mathbf{K}}}_{3,l}=0$ and $\mathcal{M}_3$ reduces to $\mathcal{M}_2$. Applying $\mathcal{M}_2$ yields the images in Figure~\ref{Fig:RecoPhantomMass}.
\begin{figure}[htbp]
	\begin{subfigure}[b]{0.32\textwidth}
		\centering
		\includegraphics[width=1.22\linewidth, trim=4.3cm 8.5cm 0 8cm, clip]{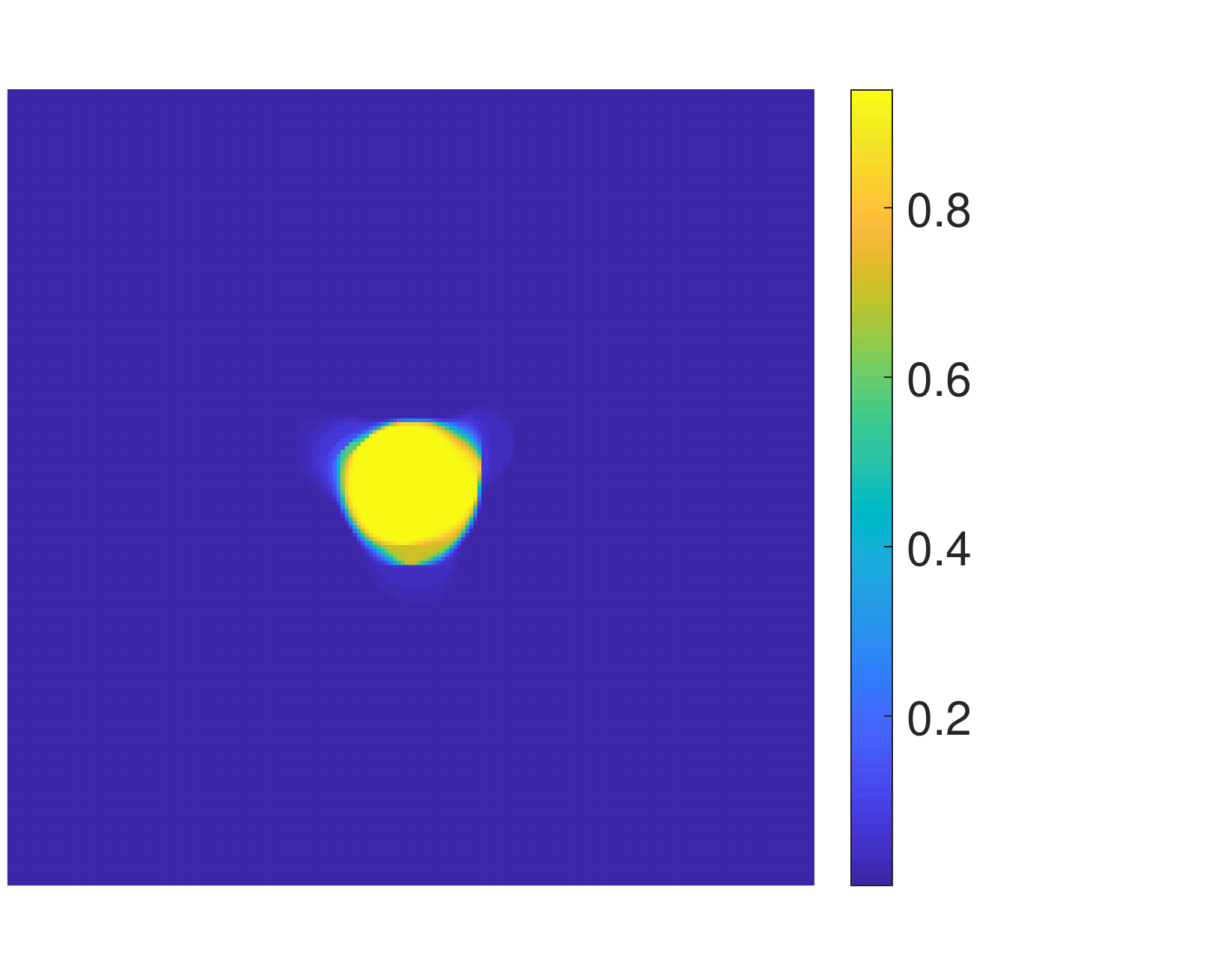}
		\subcaption{\small \small SSIM$\brackets{\mathbf{c}_0}$ = 0.9682}
		\label{Fig:RecoPhantomMass_a}
	\end{subfigure}
	\hfill
	\begin{subfigure}[b]{0.32\textwidth}
		\centering
		\includegraphics[width=1.22\linewidth, trim=4.3cm 8.5cm 0 8cm, clip]{images/Phantom.pdf}
		\subcaption{\small Groundtruth}
	\end{subfigure}
	\hfill
	\begin{subfigure}[b]{0.32\textwidth}
		\centering
		\includegraphics[width=1.22\linewidth, trim=4.3cm 8.5cm 0 8cm, clip]{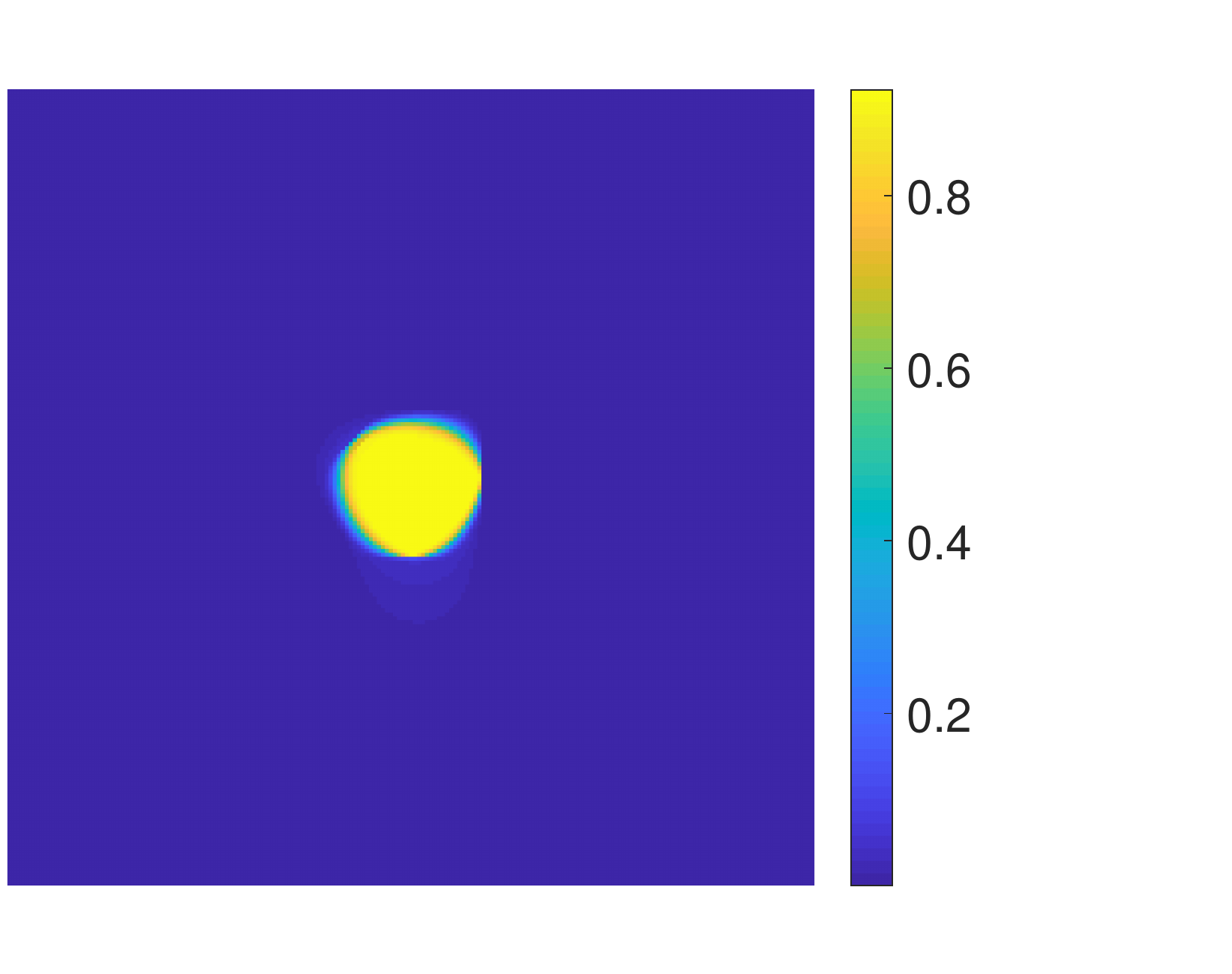}
		\subcaption{\small PSNR$\brackets{\mathbf{c}_0}$ = 28.97}
		\label{Fig:RecoPhantomMass_b}
	\end{subfigure}
	\\
	\begin{subfigure}[b]{0.32\textwidth}
		\centering
		\includegraphics[width=1.22\linewidth, trim=4.3cm 8.5cm 0 8cm, clip]{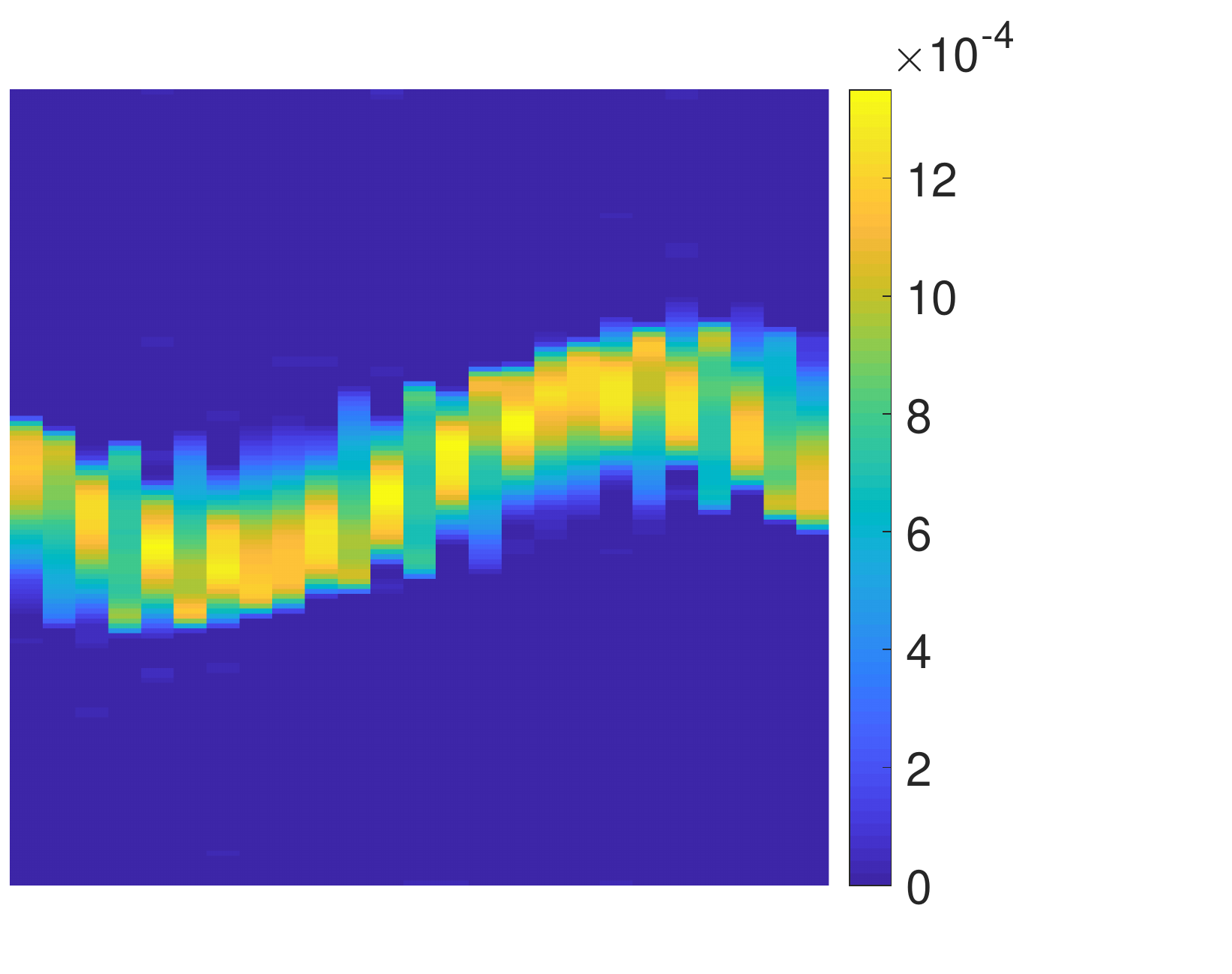}
		\subcaption{\small SSIM$\brackets{\mathbf{c}_0}$ = 0.9682}
		\label{Fig:RecoPhantomMass_c}
	\end{subfigure}
	\hfill
	\begin{subfigure}[b]{0.32\textwidth}
		\centering
		\includegraphics[width=1.22\linewidth, trim=4.3cm 8.5cm 0 8cm, clip]{images/Sino_mass.pdf}
		\subcaption{\small Groundtruth}
	\end{subfigure}
	\hfill
	\begin{subfigure}[b]{0.32\textwidth}
		\centering
		\includegraphics[width=1.22\linewidth, trim=4.3cm 8.5cm 0 8cm, clip]{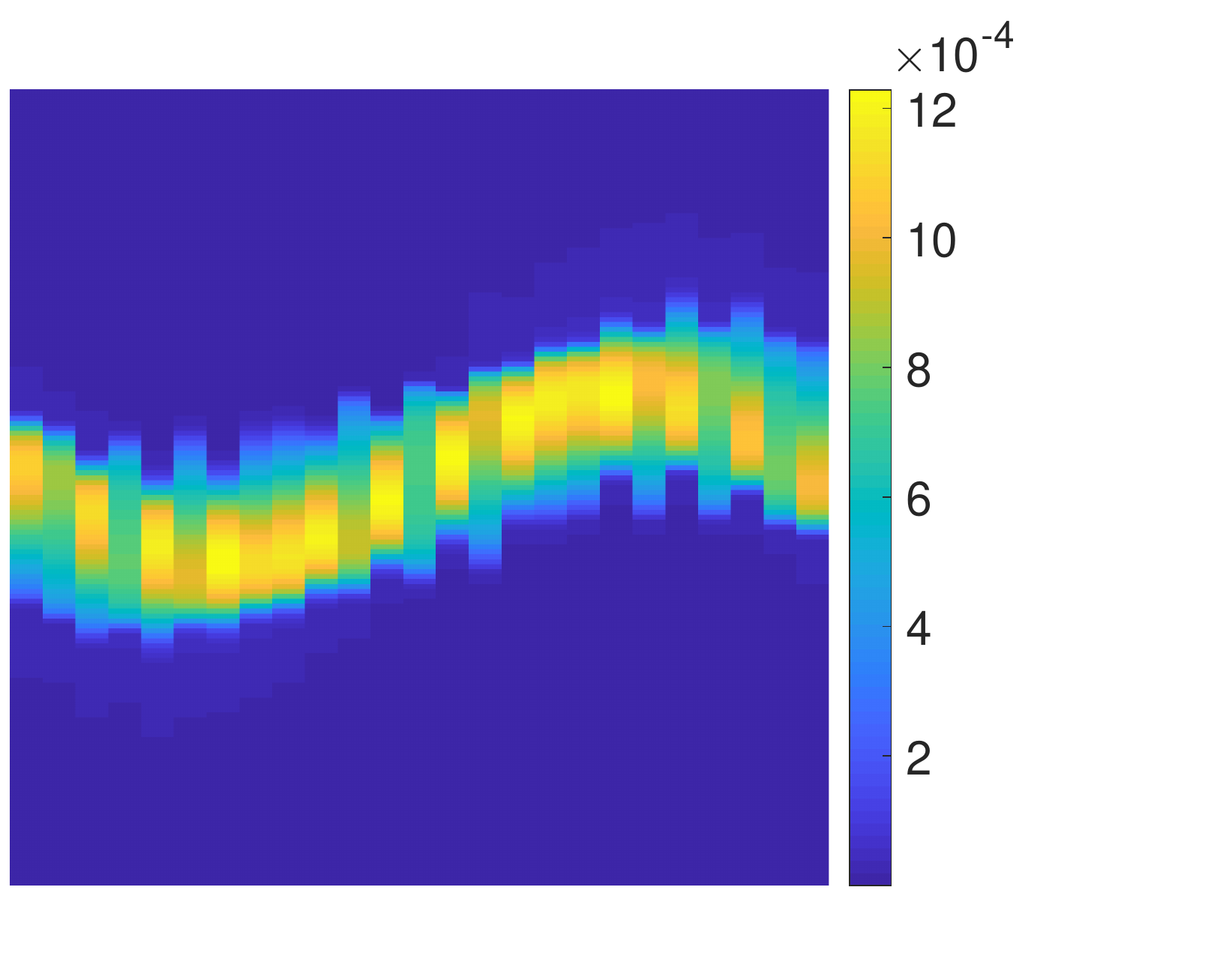}
		\subcaption{\small PSNR$\brackets{\mathbf{c}_0}$ = 28.97}
		\label{Fig:RecoPhantomMass_d}
	\end{subfigure}
	\caption{\small Phantom (first row) and corresponding sinogram (second row) reconstructions applying $\mathcal{M}_2$ incorporating motion information compared to the groundtruth phantom respectively sinogram.}
	\label{Fig:RecoPhantomMass}
\end{figure}
Both, usage of SSIM or PSNR as quality measure in the choice of parameters yields comparable reconstruction results. Phantom and sinogram are well reconstructed. The shape is not exactly a circle but can be recognized as one located at the right position. Also applying this method to data with added gaussian noise of $1\cdot 10^{-9}$ standard deviation (ca. 10\% of $u^*$) yields promising results as can be seen in Figure~\ref{Fig:RecoPhantomMass_noise}. Here, we do not compare SSIM with PSNR based image choices but we compare images with and without additional sparsity constraint on the Radon data. We find that for the sinogram obtained for $\alpha_3=0$ irritations are visible in the background, while by setting $\alpha_3>0$ we obtain a noise-free background in the sinogram and a clearer shape of the phantom. Note that for noisy data, we neglected executing a parameter comparison but rather used the parameters for the noise-free PSNR based image choice. Accordingly, we give the PSNR value in the captions, this time with higher precision to see a difference, whereas the SSIM value can be found in Table~\ref{table:RecoValues_mass} together with parameter choices and related SSIM and PSNR values for all of the figures of this section.
\begin{figure}[htbp]
	\begin{subfigure}[b]{0.32\textwidth}
		\centering
		\includegraphics[width=1.22\linewidth, trim=4.3cm 8.5cm 0 8cm, clip]{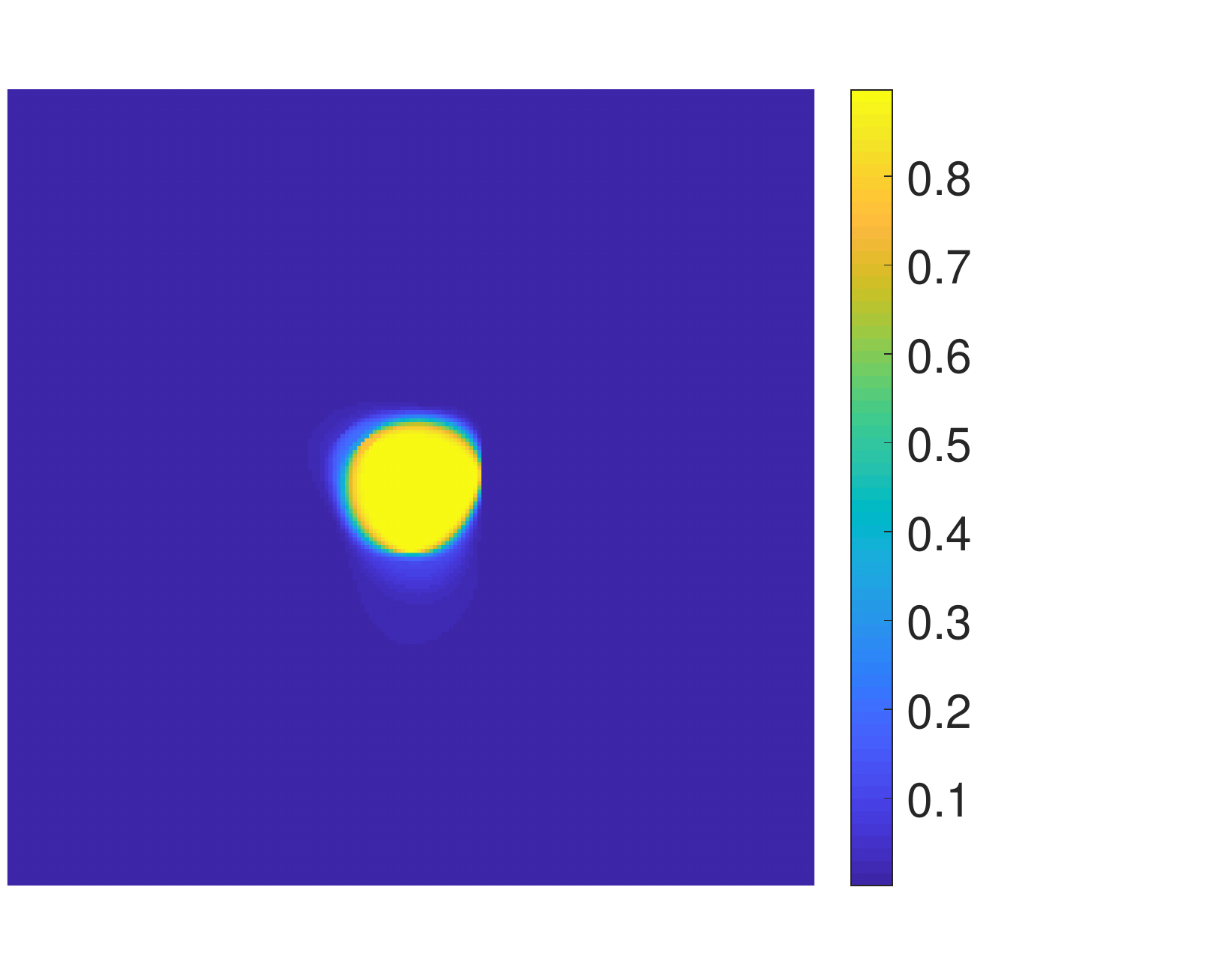}
		\subcaption{\small PSNR$\brackets{\mathbf{c}_0}$ = 27.539}
		\label{Fig:RecoPhantomMass_noise_a}
	\end{subfigure}
	\hfill
	\begin{subfigure}[b]{0.32\textwidth}
		\centering
		\includegraphics[width=1.22\linewidth, trim=4.3cm 8.5cm 0 8cm, clip]{images/Phantom.pdf}
		\subcaption{\small Groundtruth}
	\end{subfigure}
	\hfill
	\begin{subfigure}[b]{0.32\textwidth}
		\centering
		\includegraphics[width=1.22\linewidth, trim=4.3cm 8.5cm 0 8cm, clip]{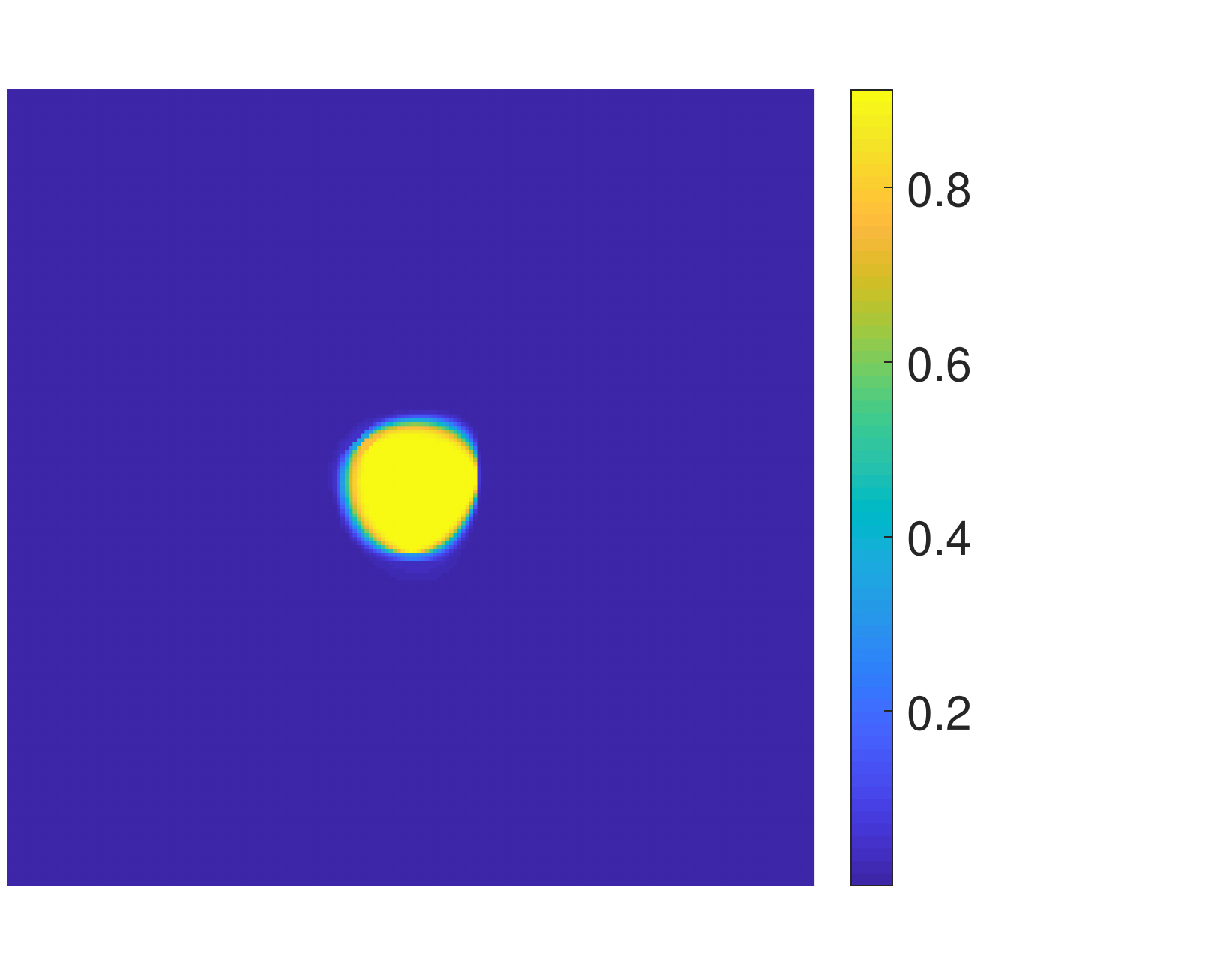}
		\subcaption{\small PSNR$\brackets{\mathbf{c}_0}$ = 27.542}
		\label{Fig:RecoPhantomMass_noise_b}
	\end{subfigure}
	\\
	\begin{subfigure}[b]{0.32\textwidth}
		\centering
		\includegraphics[width=1.22\linewidth, trim=4.3cm 8.5cm 0 8cm, clip]{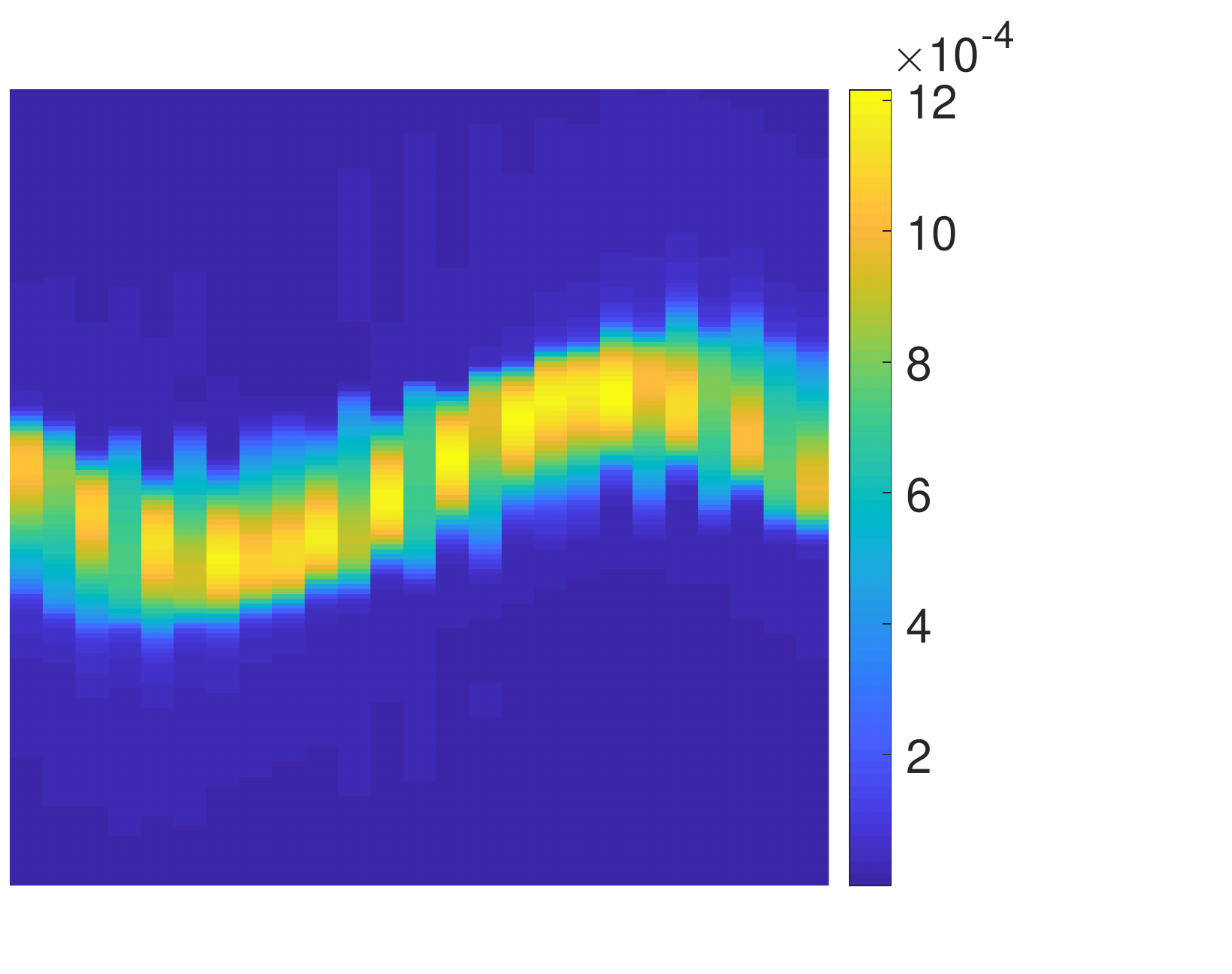}
		\subcaption{\small PSNR$\brackets{\mathbf{c}_0}$ = 27.539}
		\label{Fig:RecoPhantomMass_noise_c}
	\end{subfigure}
	\hfill
	\begin{subfigure}[b]{0.32\textwidth}
		\centering
		\includegraphics[width=1.22\linewidth, trim=4.3cm 8.5cm 0 8cm, clip]{images/Sino_mass.pdf}
		\subcaption{\small Groundtruth}
	\end{subfigure}
	\hfill
	\begin{subfigure}[b]{0.32\textwidth}
		\centering
		\includegraphics[width=1.22\linewidth, trim=4.3cm 8.5cm 0 8cm, clip]{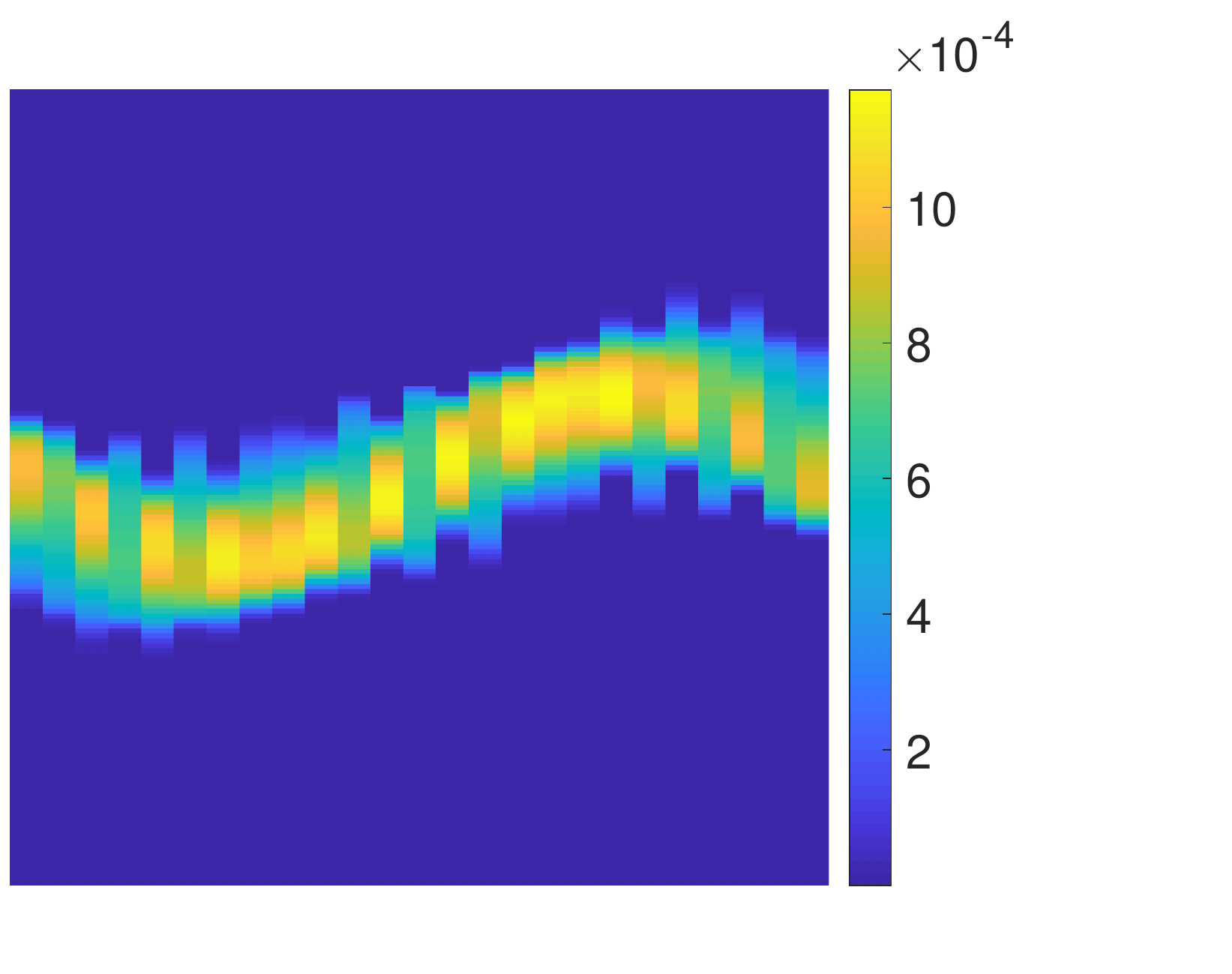}
		\subcaption{\small PSNR$\brackets{\mathbf{c}_0}$ = 27.542}
		\label{Fig:RecoPhantomMass_noise_d}
	\end{subfigure}
	\caption{\small Comparison of groundthruth images to reconstructed phantom (first row) and sinogram (second row) applying $\mathcal{M}_2$ to noisy data without (left column) and with (right column) sparsity constraint on the Radon data.}
	\label{Fig:RecoPhantomMass_noise}
\end{figure}\\
\begin{table}[htbp]
	\centering
	\caption{\small Reconstruction parameters and values}
	\renewcommand{\arraystretch}{1.3}
	\begin{tabular}{|c|c|c|c|c|c|}
		\hline
		\textbf{Figure} & $\alpha_1$ & $\alpha_2$ & $\alpha_3$ & SSIM$\brackets{\mathbf{c}_0}$ & PSNR$\brackets{\mathbf{c}_0}$ \\
		\hline \hline
		& \\[-1.6em]
		\ref{Fig:RecoPhantomMassStatic_ssim} & $2\cdot 10^5$ & $0.1^{3.6}$ & 0 & 0.8628 & 14.52 \\ 
		\ref{Fig:RecoPhantomMassStatic_psnr} & $2\cdot 10^1$ & $0.1^{4.4}$ & 0 & 0.3300 & 17.43 \\
		\hline 
		\ref{Fig:RecoPhantomMass_a} + \ref{Fig:RecoPhantomMass_c}& $2\cdot 10^2$ & $0.1^{4.6}$ & 0 & 0.9682 & 28.24 \\ 
		\ref{Fig:RecoPhantomMass_b} + \ref{Fig:RecoPhantomMass_d} & $2\cdot 10^7$ & $0.1^{1.4}$ & 0 & 0.9197 & 28.97 \\
		\hline 
		\ref{Fig:RecoPhantomMass_noise_a} + \ref{Fig:RecoPhantomMass_noise_c}& $2\cdot 10^7$ & $0.1^{1.4}$ & 0 & 0.8357 & 27.54 \\ 
		\ref{Fig:RecoPhantomMass_noise_b} + \ref{Fig:RecoPhantomMass_noise_d} & $2\cdot 10^7$ & $0.1^{1.4}$ & 20 & 0.9807 & 27.54 \\
		\hline
	\end{tabular}
	\label{table:RecoValues_mass}
	\renewcommand{\arraystretch}{1}
\end{table}
\subsection{Intensity preservation}
Next, we regard reconstruction results for the same motion but under the assumption of intensity preservation. For visualization we still refer to Figure~\ref{Fig:Motion_mass} keeping in mind that the concentration value stays the same no matter which size the object takes. Figure~\ref{Fig:RecoPhantomIntStatic} gives results neglecting the time-dependence of the phantom. Like for mass preservation it seems like the PSNR based parameter choice yields a phantom with shape closer to the searched-for solution (cf. Figure~\ref{Fig:RecoPhantomIntStatic_psnr}). The one related to structural similarity introduces artifacts that could be wrongly interpreted as three additional balls of lower concentration (cf. Figure~\ref{Fig:RecoPhantomIntStatic_ssim}). However, both images are useless for clinical diagnostic and motion information needs to be incorporated.
\begin{figure}[htbp]
	\begin{subfigure}[b]{0.5\textwidth}
		\centering
		\includegraphics[width=1\linewidth, trim=0 8.5cm 0 8cm, clip]{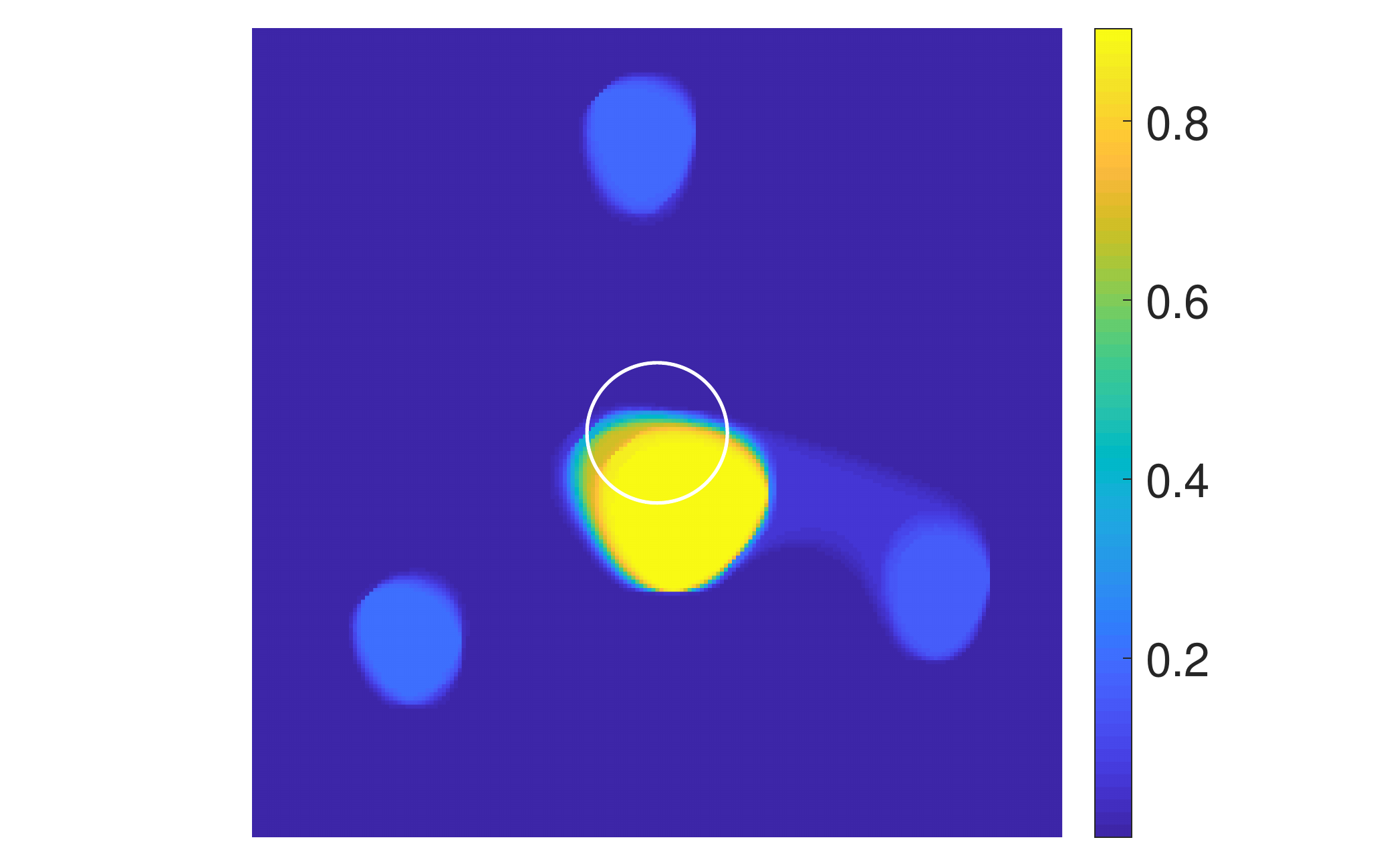}
		\subcaption{\small SSIM$\brackets{\mathbf{c}_0}$ = 0.7862}
		\label{Fig:RecoPhantomIntStatic_ssim}
	\end{subfigure}
	\hfill
	\begin{subfigure}[b]{0.5\textwidth}
		\centering
		\includegraphics[width=1\linewidth, trim=0 8.5cm 0 8cm, clip]{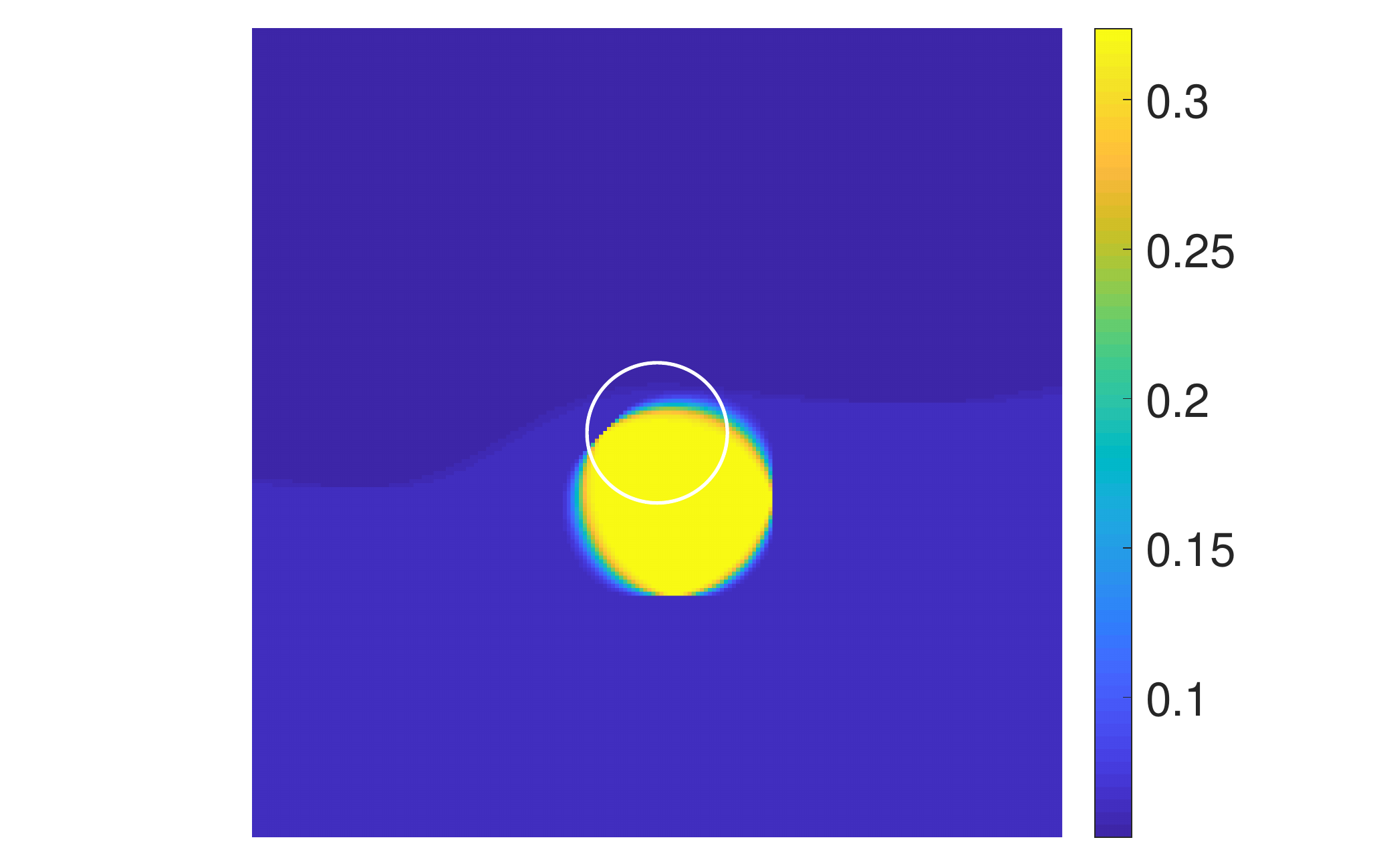}
		\subcaption{\small PSNR$\brackets{\mathbf{c}_0}$ = 15.21}
		\label{Fig:RecoPhantomIntStatic_psnr}
	\end{subfigure}
	\caption{\small Phantom reconstructions applying static reconstruction method $\mathcal{M}_1$ to dynamic data. The white circle indicates shape and location of the chosen reference concentration.}
	\label{Fig:RecoPhantomIntStatic}
\end{figure}\\
This time we have that $\widehat{\widetilde{\mathbf{K}}}_{3,l}\neq0$ and we compare results incorporating this term (method $\mathcal{M}_3$) with those neglecting it (method $\mathcal{M}_2$). In Figure~\ref{Fig:RecoPhantomInt_ssim_} we compare results for the SSIM based parameter choice. In contrast to the results ignoring the dynamics, we find that the circle is reconstructed at the right position. Applying method $\mathcal{M}_2$ still yields three additional balls that could be interpreted as parts of the phantom, whereas for method $\mathcal{M}_3$ we also have artifacts but they can be identified as artifacts. However, in the latter we are confronted with outliers, which might explain the lower SSIM value.
\begin{figure}[htbp]
	\begin{subfigure}[b]{0.32\textwidth}
		\centering
		\includegraphics[width=1.22\linewidth, trim=4.3cm 8.5cm 0 8cm, clip]{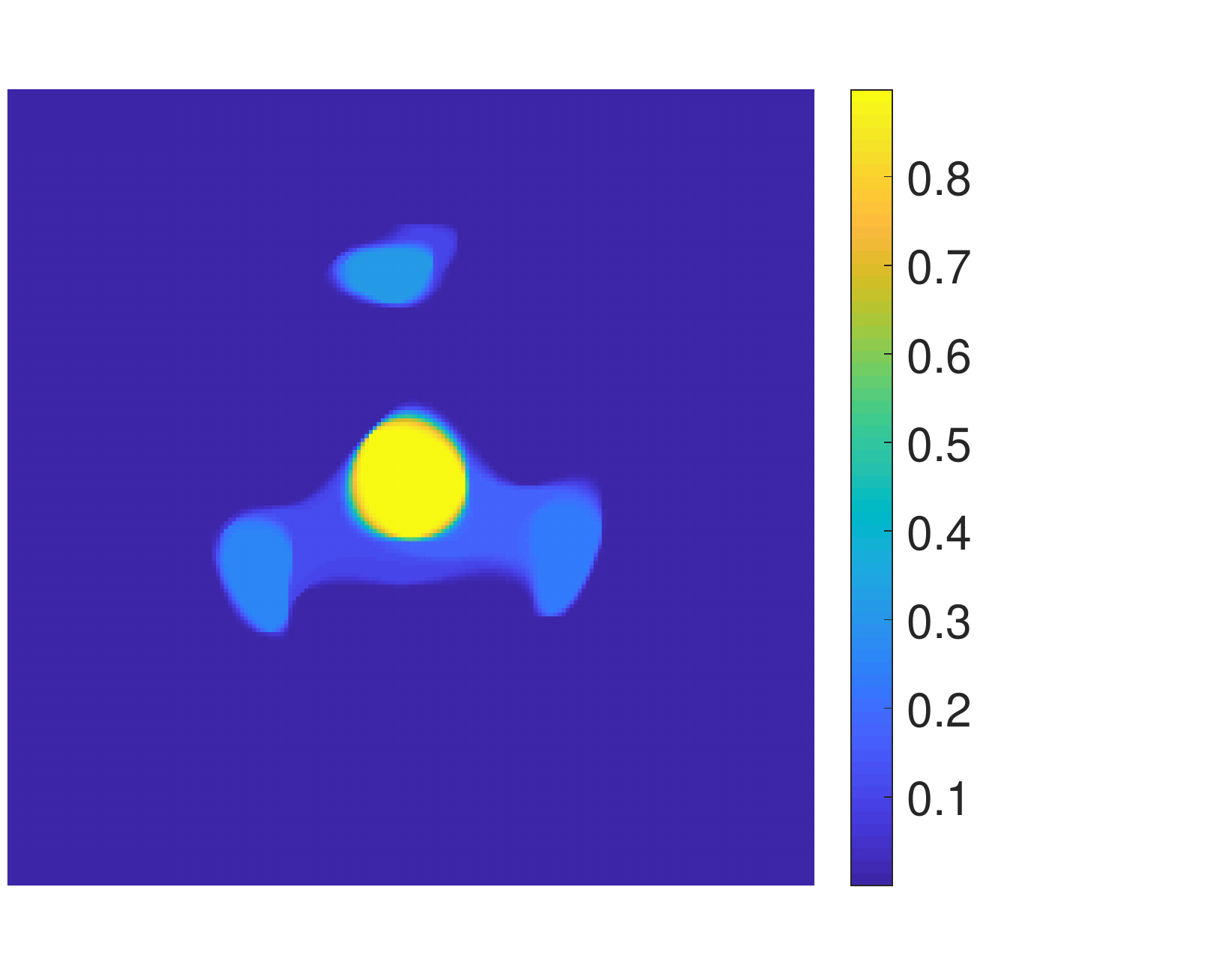}
		\subcaption{\small SSIM$\brackets{\mathbf{c}_0}$ = 0.8955}
		\label{Fig:RecoPhantomInt_ssim}
	\end{subfigure}
	\hfill
	\begin{subfigure}[b]{0.32\textwidth}
		\centering
		\includegraphics[width=1.22\linewidth, trim=4.3cm 8.5cm 0 8cm, clip]{images/Phantom.pdf}
		\subcaption{\small Groundtruth}
	\end{subfigure}
	\hfill
	\begin{subfigure}[b]{0.32\textwidth}
		\centering
		\includegraphics[width=1.22\linewidth, trim=4.3cm 8.5cm 0 8cm, clip]{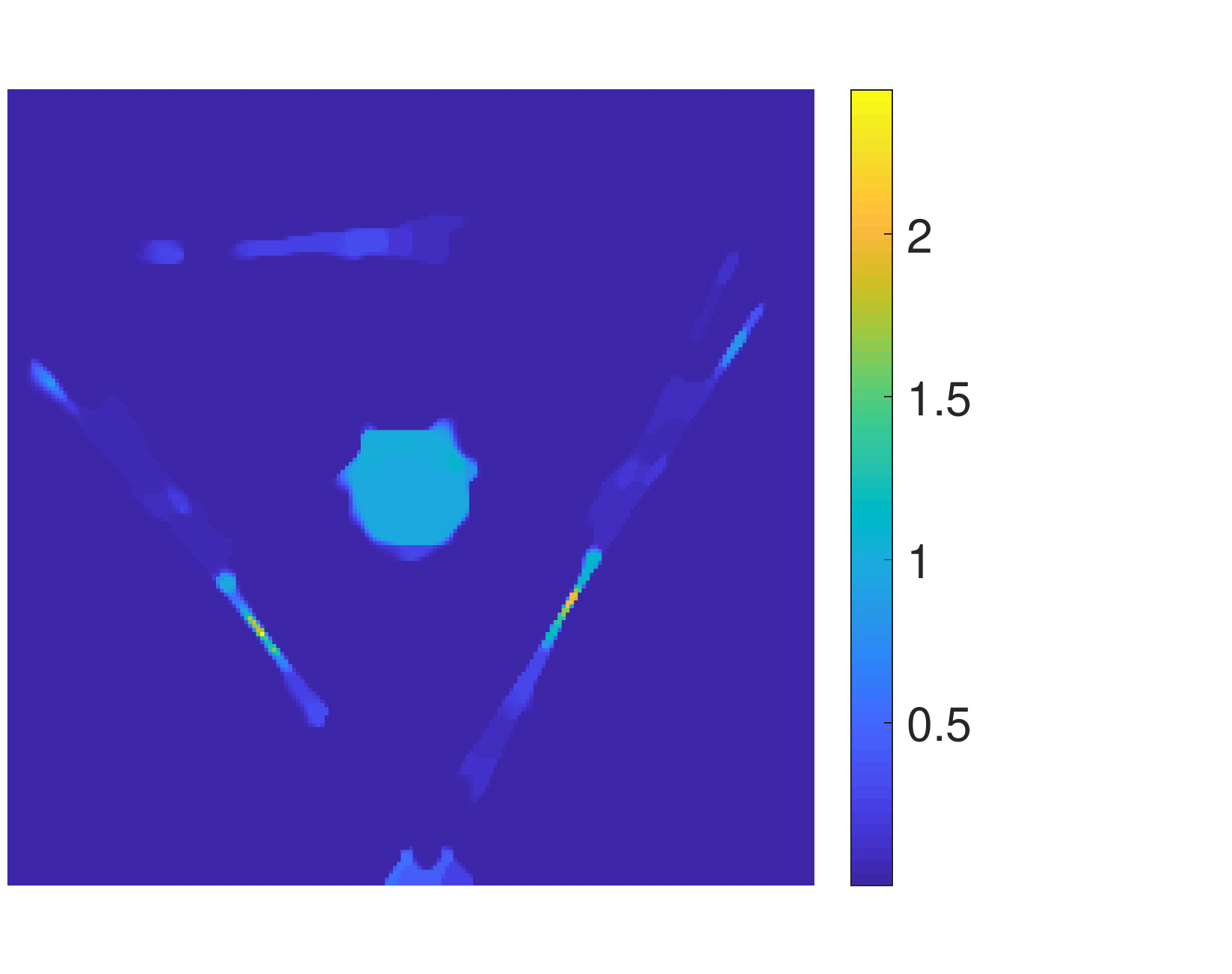}
		\subcaption{\small SSIM$\brackets{\mathbf{c}_0}$ = 0.8465}
		\label{Fig:RecoPhantomIntTwoTerms_ssim}
	\end{subfigure}
	\\
	\begin{subfigure}[b]{0.32\textwidth}
		\centering
		\includegraphics[width=1.22\linewidth, trim=4.3cm 8.5cm 0 8cm, clip]{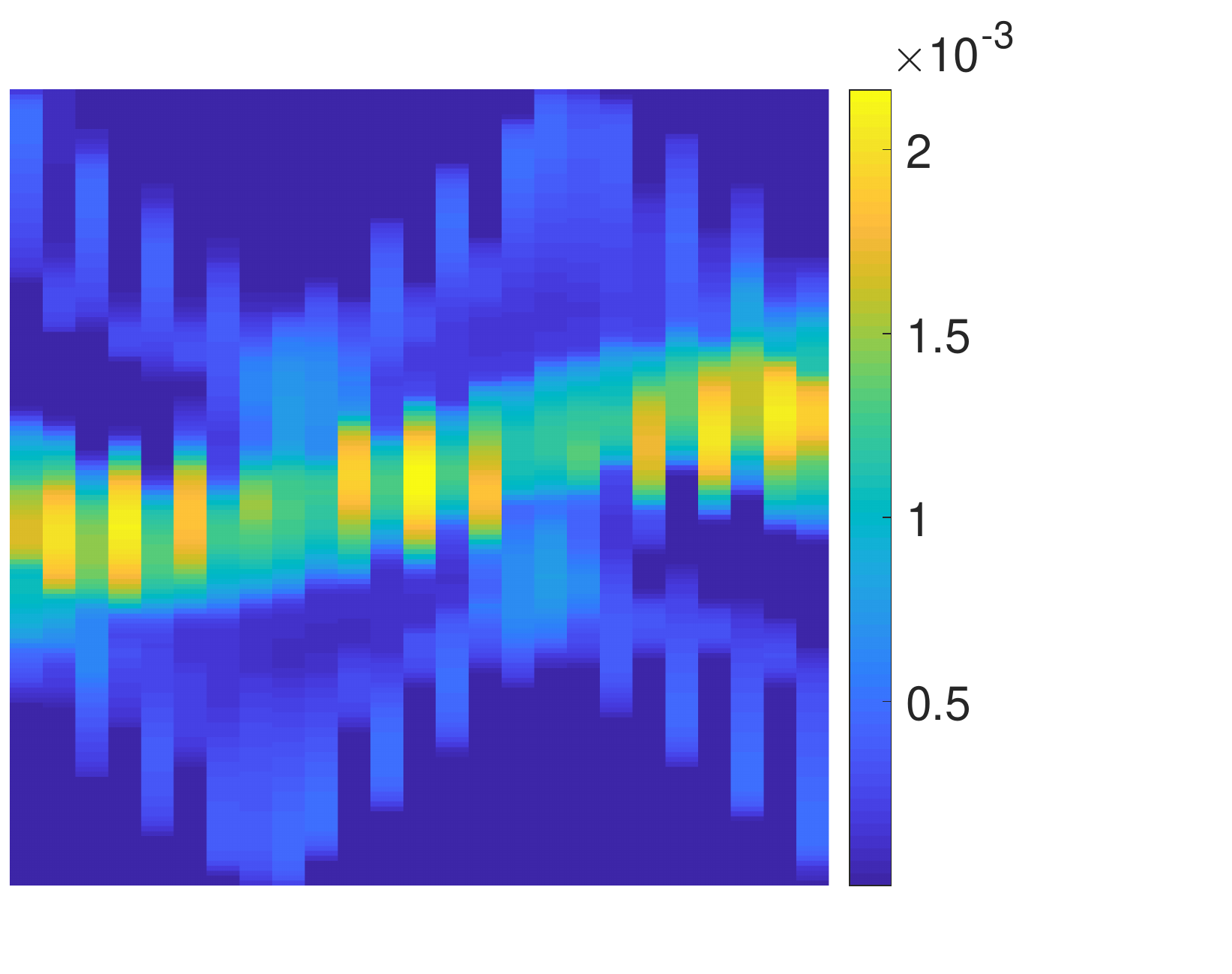}
		\subcaption{\small SSIM$\brackets{\mathbf{c}_0}$ = 0.8955}
		\label{Fig:RecoSinoInt_ssim}
	\end{subfigure}
	\hfill
	\begin{subfigure}[b]{0.32\textwidth}
		\centering
		\includegraphics[width=1.22\linewidth, trim=4.3cm 8.5cm 0 8cm, clip]{images/Sino_int.pdf}
		\subcaption{\small Groundtruth}
	\end{subfigure}
	\hfill
	\begin{subfigure}[b]{0.32\textwidth}
		\centering
		\includegraphics[width=1.22\linewidth, trim=4.3cm 8.5cm 0 8cm, clip]{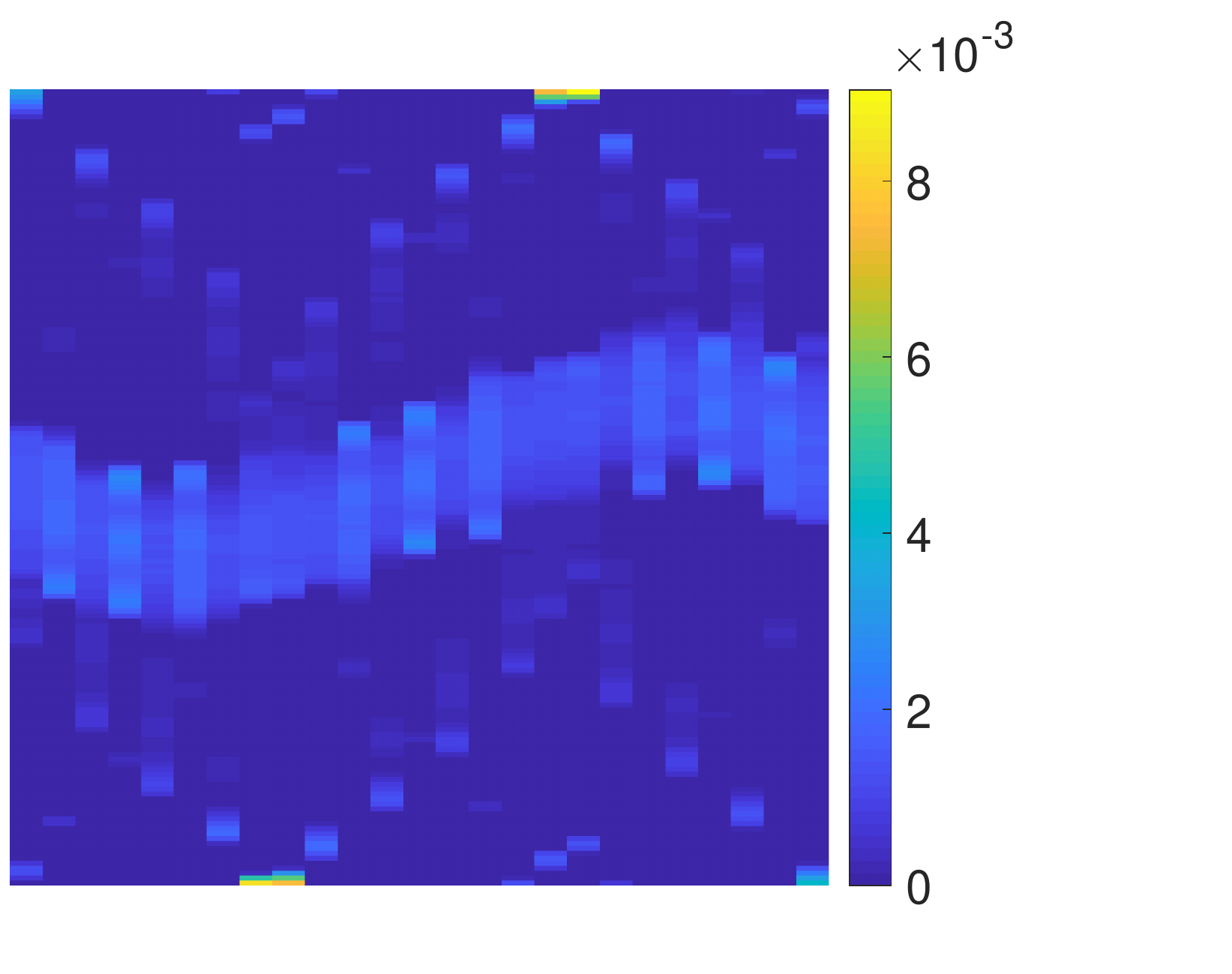}
		\subcaption{\small SSIM$\brackets{\mathbf{c}_0}$ = 0.8465}
		\label{Fig:RecoSinoIntTwoTerms_ssim}
	\end{subfigure}
	\caption{\small Groundtruth images compared to reconstructed phantom (first row) and sinogram (second row) applying $\mathcal{M}_2$ (left column) or $\mathcal{M}_3$ (right column) for SSIM based parameter choice.}
	\label{Fig:RecoPhantomInt_ssim_}
\end{figure}
\begin{figure}[htbp]
	\begin{subfigure}[b]{0.32\textwidth}
		\centering
		\includegraphics[width=1.22\linewidth, trim=4.3cm 8.5cm 0 8cm, clip]{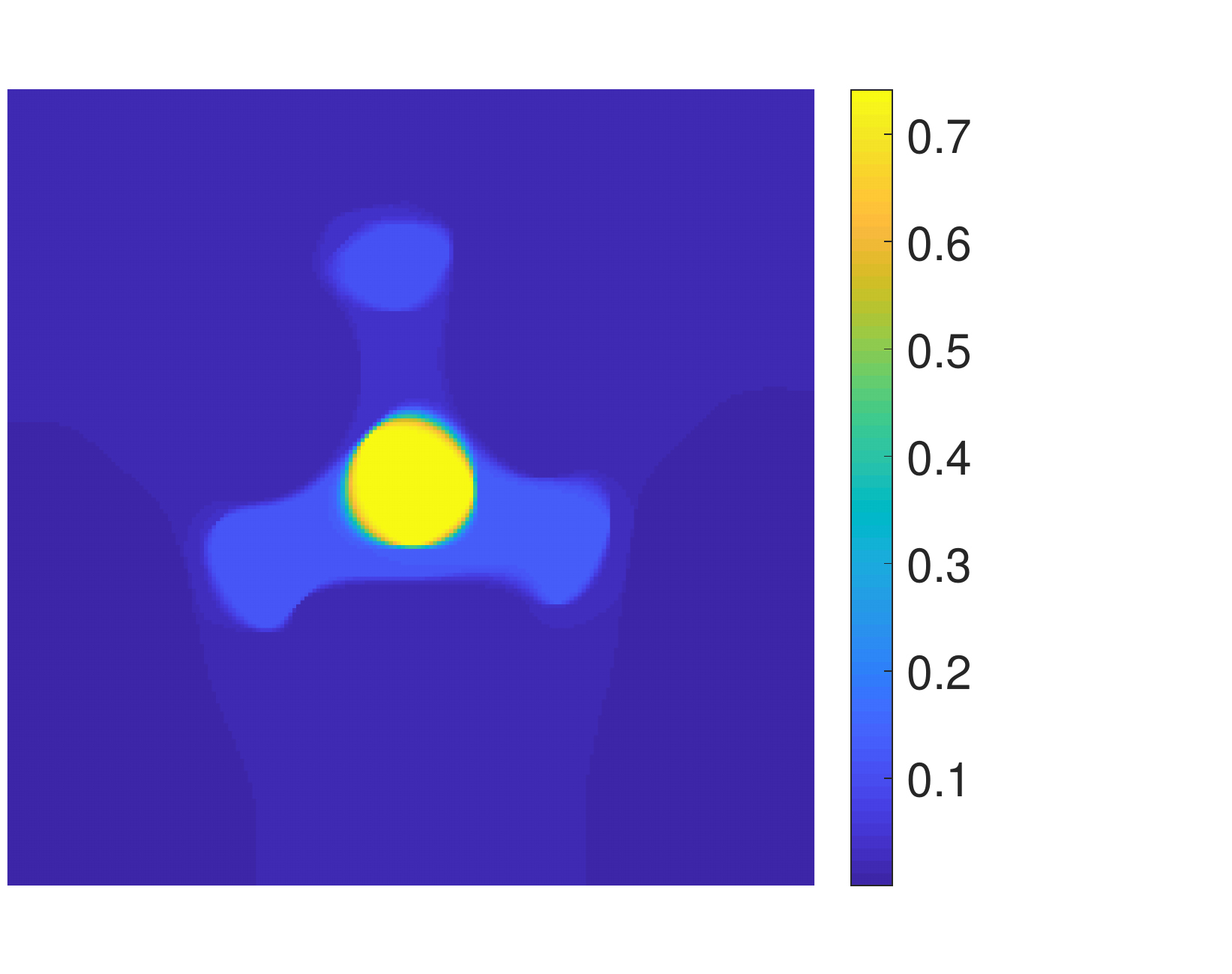}
		\subcaption{\small PSNR$\brackets{\mathbf{c}_0}$ = 23.26}
		\label{Fig:RecoPhantomInt_psnr}
	\end{subfigure}
	\hfill
	\begin{subfigure}[b]{0.32\textwidth}
		\centering
		\includegraphics[width=1.22\linewidth, trim=4.3cm 8.5cm 0 8cm, clip]{images/Phantom.pdf}
		\subcaption{\small Groundtruth}
	\end{subfigure}
	\hfill
	\begin{subfigure}[b]{0.32\textwidth}
		\centering
		\includegraphics[width=1.22\linewidth, trim=4.3cm 8.5cm 0 8cm, clip]{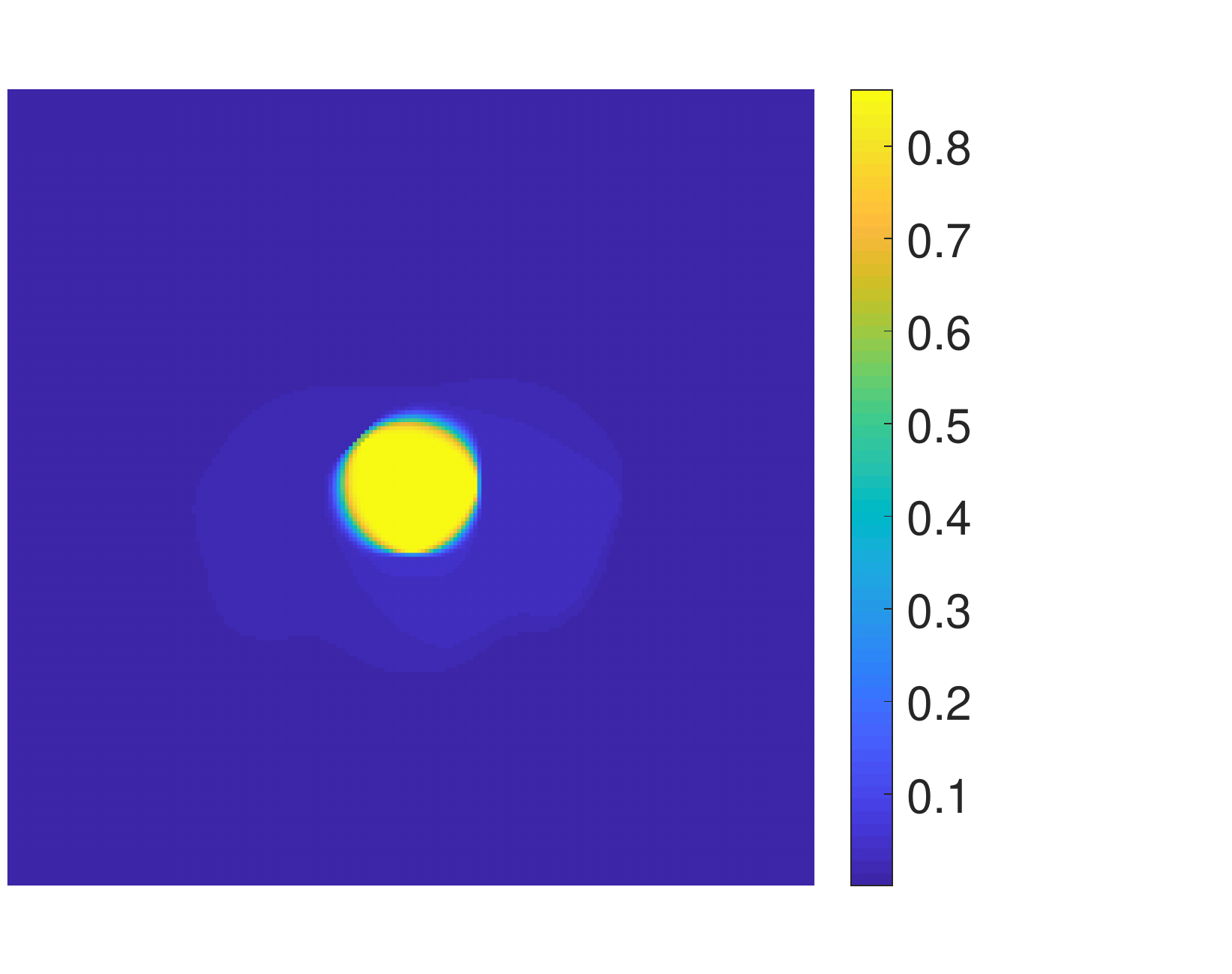}
		\subcaption{\small PSNR$\brackets{\mathbf{c}_0}$ = 27.64}
		\label{Fig:RecoPhantomIntTwoTerms_psnr}
	\end{subfigure}
	\\
	\begin{subfigure}[b]{0.32\textwidth}
		\centering
		\includegraphics[width=1.22\linewidth, trim=4.3cm 8.5cm 0 8cm, clip]{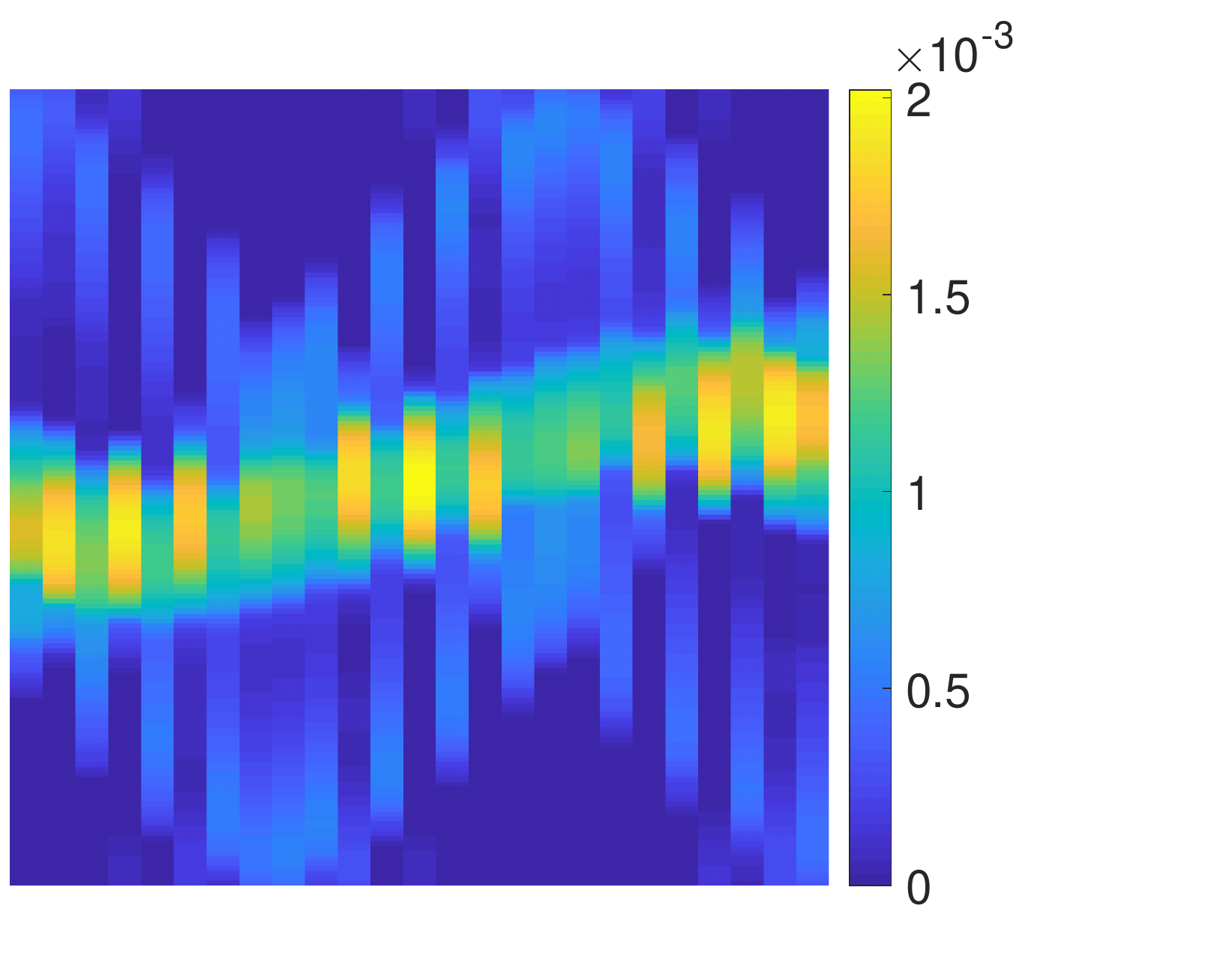}
		\subcaption{\small PSNR$\brackets{\mathbf{c}_0}$ = 23.26}
		\label{Fig:RecoSinoInt_psnr}
	\end{subfigure}
	\hfill
	\begin{subfigure}[b]{0.32\textwidth}
		\centering
		\includegraphics[width=1.22\linewidth, trim=4.3cm 8.5cm 0 8cm, clip]{images/Sino_int.pdf}
		\subcaption{\small Groundtruth}
	\end{subfigure}
	\hfill
	\begin{subfigure}[b]{0.32\textwidth}
		\centering
		\includegraphics[width=1.22\linewidth, trim=4.3cm 8.5cm 0 8cm, clip]{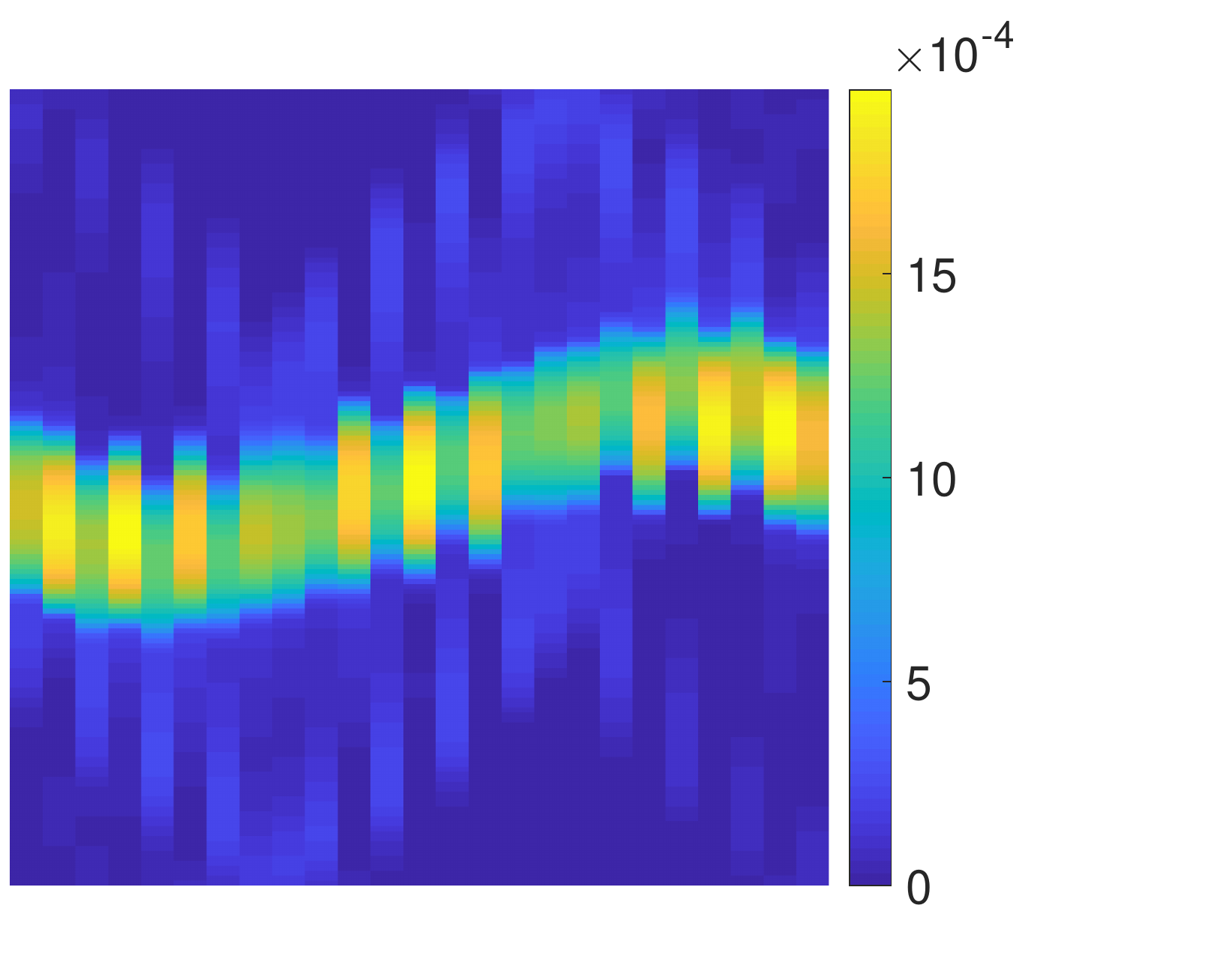}
		\subcaption{\small PSNR$\brackets{\mathbf{c}_0}$ = 27.64}
		\label{Fig:RecoSinoIntTwoTerms_psnr}
	\end{subfigure}
	\caption{\small Groundtruth images compared to reconstructed phantom (first row) and sinogram (second row) applying $\mathcal{M}_2$ (left column) or $\mathcal{M}_3$ (right column) for PSNR based parameter choice.}
	\label{Fig:RecoPhantomInt_psnr_}
\end{figure}
\begin{figure}[htbp]
	\begin{subfigure}[b]{0.5\textwidth}
		\centering
		\includegraphics[width=1\linewidth, trim=0 8.5cm 0 8cm, clip]{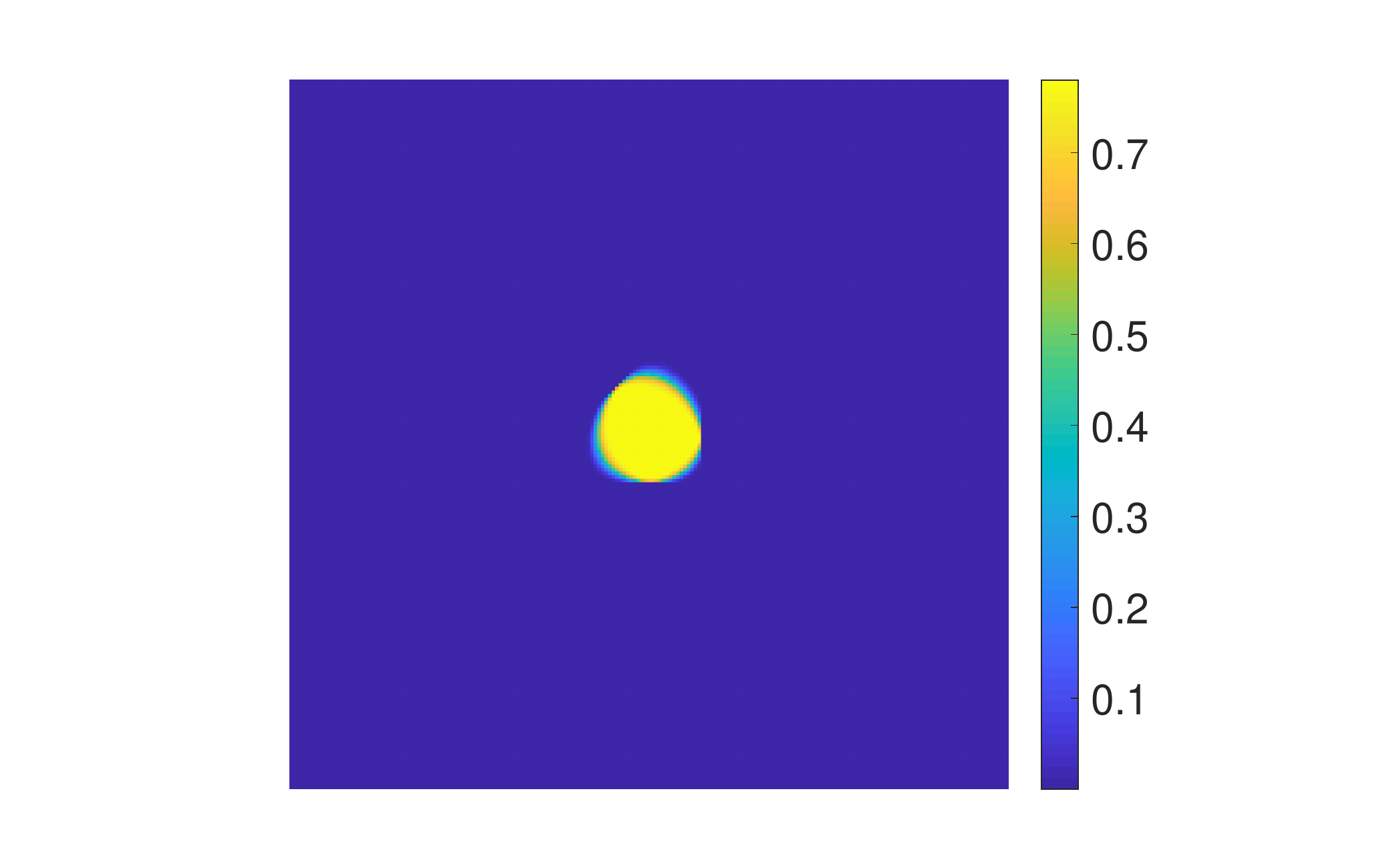}
		\subcaption{\small SSIM$\brackets{\mathbf{c}_0}$ = 0.9703}
	\end{subfigure}
	\hfill
	\begin{subfigure}[b]{0.5\textwidth}
		\centering
		\includegraphics[width=1\linewidth, trim=0 8.5cm 0 8cm, clip]{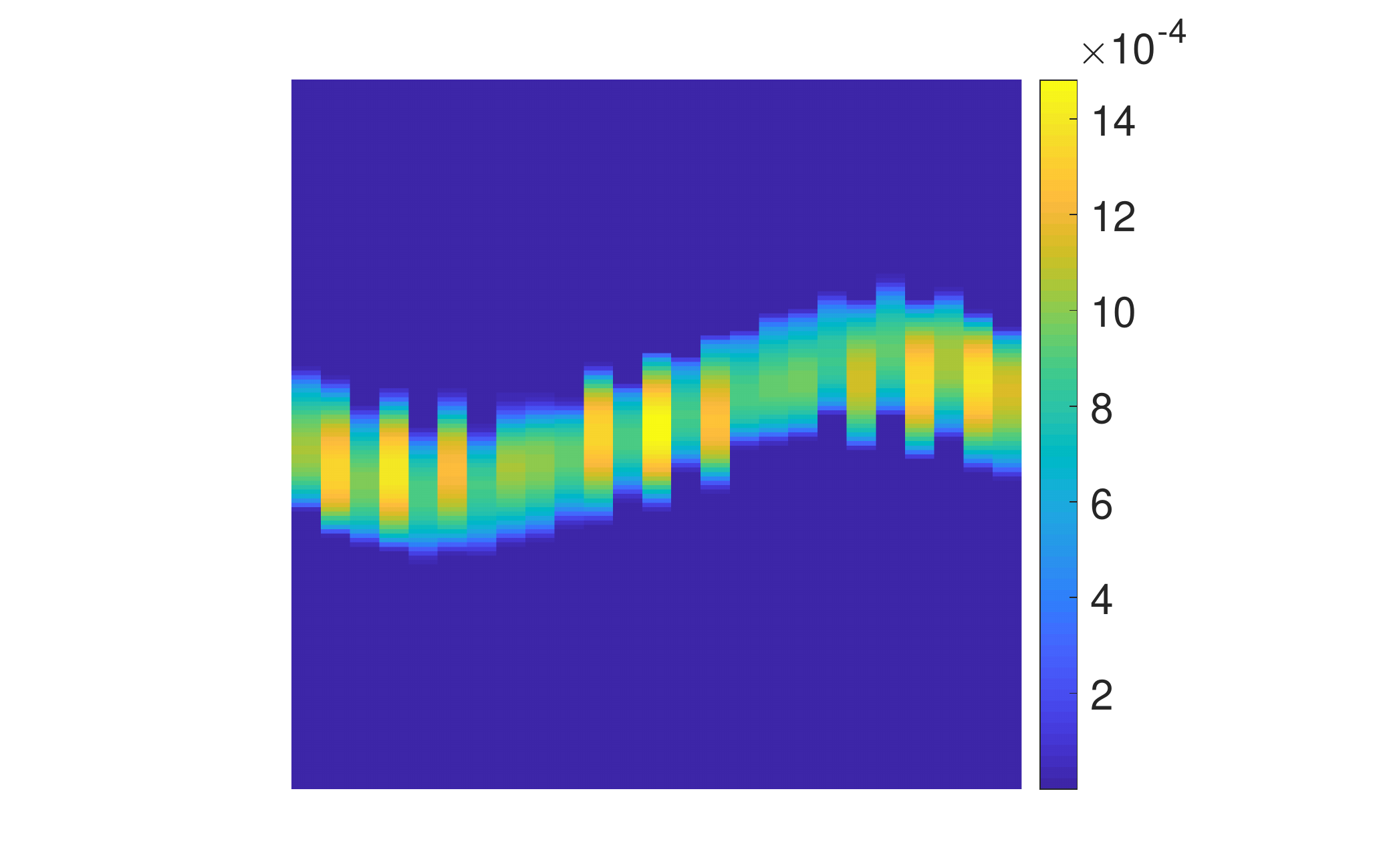}
		\subcaption{\small SSIM$\brackets{\mathbf{c}_0}$ = 0.9703}
	\end{subfigure}
	\caption{\small Phantom and sinogram reconstruction applying method $\mathcal{M}_2$ with $\alpha_3>0$.}
	\label{Fig:RecoPhantomIntSparsity}
\end{figure}
\\
For PSNR based image choice (cf. Figure~\ref{Fig:RecoPhantomInt_psnr_}) we obtain, using method $\mathcal{M}_3$ incorporating an additional term in the forward model, an image without outliers, well-reconstructed edges, and correct location of the phantom. Nevertheless, especially in the sinogram we find background irritations. A possible reason for the errors in the sinogram might be that although we improve the model by adding $\widehat{\widetilde{\mathbf{K}}}_{3,l}$ we still have the inaccuracy resulting from using the small sinogram instead of a needed full one. However, artifacts can be further reduced by applying a sparsity constraint on the sinogram in addition to the TV penalty term acting on the concentration. Therewith, we even obtain for method $\mathcal{M}_2$, neglecting $\widehat{\widetilde{\mathbf{K}}}_{3,l}$, good reconstruction results as can be seen in Figure~\ref{Fig:RecoPhantomIntSparsity}. Here, we did not compare similarity results for different parameter choices, instead we chose $\alpha_1,\;\alpha_2$ as for the SSIM based image choice for method $\mathcal{M}_2$ in combination with setting $\alpha_3>0$. Thus, in the caption we give the corresponding SSIM value. Again, precise parameter choices and related SSIM and PSNR values can be found in Table~\ref{table:RecoValues_int}.
\begin{table}[htbp]
	\centering
	\caption{\small Reconstruction parameters and values}
	\renewcommand{\arraystretch}{1.3}
	\begin{tabular}{|c|c|c|c|c|c|}
		\hline
		\textbf{Figure} & $\alpha_1$ & $\alpha_2$ & $\alpha_3$ & SSIM$\brackets{\mathbf{c}_0}$ & PSNR$\brackets{\mathbf{c}_0}$ \\
		\hline \hline
		& \\[-1.6em]
		\ref{Fig:RecoPhantomIntStatic_ssim} & $2\cdot 10^7$ & $0.1^{2}$ & 0 & 0.7862 & 15.21 \\ 
		\ref{Fig:RecoPhantomIntStatic_psnr} & $2\cdot 10^4$ & $0.1^{1.2}$ & 0 & 0.0327 & 16.93 \\
		\hline 
		\ref{Fig:RecoPhantomInt_ssim} + \ref{Fig:RecoSinoInt_ssim}&  $2\cdot 10^7$ & $0.1^{1.4}$ & 0 & 0.8955 & 21.92 \\
		\ref{Fig:RecoPhantomIntTwoTerms_ssim} + \ref{Fig:RecoSinoIntTwoTerms_ssim}&  $2\cdot 10^2$ & $0.1^{4.4}$ & 0 & 0.8465 & 19.75 \\
		\hline 
		\ref{Fig:RecoPhantomInt_psnr} + \ref{Fig:RecoSinoInt_psnr}&  $2\cdot 10^5$ & $0.1^{1.2}$ & 0 & 0.4658 & 23.26 \\
		\ref{Fig:RecoPhantomIntTwoTerms_psnr} + \ref{Fig:RecoSinoIntTwoTerms_psnr}&  $2\cdot 10^5$ & $0.1^{1.4}$ & 0 & 0.6078 & 27.64 \\
		\hline 
		\ref{Fig:RecoPhantomIntSparsity} &  $2\cdot 10^7$ & $0.1^{1.4}$ & 60 & 0.9703 & 21.85 \\
		\hline
	\end{tabular}
	\label{table:RecoValues_int}
	\renewcommand{\arraystretch}{1}
\end{table}

%%%%%%%%%%%%%%%%%%%%%%%%%%%%%%%%%%
\section{Conclusion}
%%%%%%%%%%%%%%%%%%%%%%%%%%%%%%%%%%
Motivated by the similar sampling geometry for MPI using an FFL scanner and CT, we extended the MPI signal equation to dynamic particle concentrations  using diffeomorphic motion functions as successfully applied in dynamic CT. Compared to the forward model in the static case, we get additional components resulting from the new time-dependencies due to the moving phantom. We were able to link each of these components to adapted versions of the Radon transform and proposed a joint reconstruction of particle concentration and corresponding dynamic Radon data by means of TV regularization in terms of the concentration and an optional penalty term for the corresponding Radon data. Finally, we applied different versions of our method to synthetic data. Among other things, we compared results obtained via neglecting the phantom dynamics with those obtained via incorporating motion information, clearly emphasizing the need for suitable dynamic reconstruction methods in order to get reasonable image reconstructions. \\ 
Further, we regarded an example for mass as well as for intensity preservation. The mass preserving case seems to be easier to handle as we only have one additional term instead of two. However, for our motion the second term was neglectable for both preservation assumptions. We got promising results, which could serve as starting point for future research. For instance, it would be interesting to investigate whether the approach allows extension to more realistic magnetization models in contrast to using the Langevin model. We also accepted further modeling errors by using a reduced sinogram in the reconstruction. According to this, adding sinogram inpainting methods into our Radon-based image reconstruction approach might serve as remedy. So far, we assumed the motion functions to be given. A next step would be to develop methods for determining these from the measurements and a priori information either beforehand of or simultaneously with the image reconstruction. Existing methods for dynamic CT and dynamic inverse problems in general might serve as point of orientation. Lastly, testing and possibly adapting the methods with respect to real data are needed to check applicability for clinical issues.

%%%%%%%%%%%%%%%%%%%%%%%%%%%%%%%%%%
\section*{Acknowledgment}
The authors acknowledge the support by the Deutsche Forschungsgemeinschaft (DFG) within the Research Training Group GRK 2583 "Modeling, Simulation and Optimization of Fluid Dynamic Applications".
%%%%%%%%%%%%%%%%%%%%%%%%%%%%%%%%%%

%%%%%%%%%%%%%%%%%%%%%%%%%%%%%%%%%%
% References
%%%%%%%%%%%%%%%%%%%%%%%%%%%%%%%%%%

\bibliographystyle{plain}
\bibliography{literature}

\end{document}